\documentclass[pra,twocolumn,floatfix,superscriptaddress,longbibliography,nofootinbib]{revtex4-1}

\usepackage{graphicx}
\usepackage{color}
\usepackage{bm}
\usepackage{enumerate}
\usepackage[T1]{fontenc}
\usepackage{bbold}
\usepackage[colorlinks=true,citecolor=blue,linkcolor=magenta]{hyperref}
\usepackage{multirow}

\usepackage[braket, qm]{qcircuit}
\usepackage[ruled]{algorithm}
\usepackage{algpseudocode}
\usepackage{amssymb, amsthm}
\usepackage{physics}

\usepackage[caption = false]{subfig}

\newcommand{\fattheta}{\boldsymbol{\theta}}
\newcommand{\dya}[1]{\ket{#1}\!\bra{#1}}
\newcommand{\tout}{{\text{out}}}
\newcommand{\tin}{{\text{in}}}

\newcommand{\DC}{\mathcal{D}}
\newcommand{\EC}{\mathcal{E}}

\newcommand{\LC}{\mathcal{L}}
\newcommand{\MC}{\mathcal{M}}
\newcommand{\NC}{\mathcal{N}}
\newcommand{\PC}{\mathcal{P}}

\newcommand{\RC}{\mathcal{R}}
\newcommand{\SC}{\mathcal{S}}

\newcommand{\YC}{\mathcal{Y}}

\renewcommand{\geq}{\geqslant}
\renewcommand{\leq}{\leqslant}

\renewcommand{\vec}[1]{\boldsymbol{#1}}  

\newcommand{\ad}{^\dagger}

\newcommand*{\id}{\openone}

\newcommand{\rhot}{\tilde{\rho}}

\newcommand{\tot}{\text{tot}}

\newcommand{\thv}{\vb*{\theta}}

\newcommand{\gindex}{x}
\newcommand{\gdimension}{d}
\newcommand{\indicatorvec}{\vb*{\delta}}
\newcommand{\hscopies}{m}
\newcommand{\polydegree}{D}
\newcommand{\learningrate}{\alpha}



\newtheorem{corollary}{Corollary}
\newtheorem{proposition}{Proposition}

\begin{document}

\preprint{APS/123-QED}

\title{Resource frugal optimizer for quantum machine learning}

\author{Charles Moussa}
\email{c.moussa@liacs.leidenuniv.nl}
\affiliation{Theoretical Division, Los Alamos National Laboratory, Los Alamos, New Mexico 87545, USA}%
\affiliation{LIACS, Leiden University, Niels Bohrweg 1, 2333 CA Leiden, Netherlands}%

\author{Max Hunter Gordon}
\affiliation{Theoretical Division, Los Alamos National Laboratory, Los Alamos, New Mexico 87545, USA}
\affiliation{Instituto de Física Teórica, UAM/CSIC, Universidad Autónoma de Madrid, Madrid 28049, Spain}

\author{Michal Baczyk}
\affiliation{Theoretical Division, Los Alamos National Laboratory, Los Alamos, New Mexico 87545, USA}%
\affiliation{Faculty of Physics, University of Warsaw, ulica Pasteura 5, 02-093 Warsaw, Poland}%

\author{M. Cerezo}
\affiliation{Information Sciences, Los Alamos National Laboratory, Los Alamos, NM 87545, USA}

\author{Lukasz Cincio}
\affiliation{Theoretical Division, Los Alamos National Laboratory, Los Alamos, New Mexico 87545, USA}

\author{Patrick J. Coles}
\affiliation{Theoretical Division, Los Alamos National Laboratory, Los Alamos, New Mexico 87545, USA}
\affiliation{Normal Computing Corporation, New York, New York, USA}


\begin{abstract}
Quantum-enhanced data science, also known as quantum machine learning (QML), is of growing interest as an application of near-term quantum computers. Variational QML algorithms have the potential to solve practical problems on real hardware, particularly when involving quantum data. However, training these algorithms can be challenging and calls for tailored optimization procedures. Specifically, QML applications can require a large shot-count overhead due to the large datasets involved. In this work, we advocate for simultaneous random sampling over both the dataset as well as the measurement operators that define the loss function. We consider a highly general loss function that encompasses many QML applications, and we show how to construct an unbiased estimator of its gradient. This allows us to propose a shot-frugal gradient descent optimizer called Refoqus (REsource Frugal Optimizer for QUantum Stochastic gradient descent).  Our numerics indicate that Refoqus can save several orders of magnitude in shot cost, even relative to optimizers that sample over measurement operators alone. 
\end{abstract}

\maketitle

\section{Introduction}

A new kind of data is emerging in recent times: quantum data. Tabletop quantum experiments and analog quantum simulators produce interesting sets of quantum states that must be characterized. Moreover, the rise of digital quantum computers is leading to the discovery of novel quantum circuits that can, once again, produce quantum states of interest. Quantum sensing, quantum phase diagrams, quantum error correction, and quantum dynamics are some of the areas that stand to benefit from quantum data analysis.

Classical machine learning was developed for the processing of classical data, but it is necessarily inefficient at processing quantum data. This issue has given rise to the field of quantum machine learning (QML)~\cite{biamonte2017quantum,schuld2015introduction}. QML has seen the proposal of parameterized quantum models, such as quantum neural networks~\cite{schuld2014quest,cong2019quantum,abbas2020power,nguyen2022atheory}, that could efficiently process quantum data. Variational QML, which involves classically training a parameterized quantum model, is indeed a leading candidate for implementing QML in the near term.

Variational QML, which we will henceforth refer to as QML for simplicity, has faced various sorts of trainability issues. Exponentially vanishing gradients, known as barren plateaus~\cite{mcclean2018barren,cerezo2020cost,holmes2020barren,holmes2021connecting,sharma2020trainability,marrero2020entanglement,uvarov2020barren,arrasmith2020effect,pesah2020absence}, as well as the prevalence of local minima~\cite{bittel2021training,anschuetz2022beyond} are two issues that can impact the complexity of the training process. Quantum hardware noise also impacts trainability~\cite{wang2020noise,franca2020limitations}. All of these issues contribute to increasing the number of shots and iterations required to minimize the QML loss function. Indeed, a detailed shot-cost analysis has painted a concerning picture~\cite{wecker2015progress}.

It is therefore clear that QML requires careful frugality in terms of the resources expended during the optimization process. Indeed, novel optimizers have been developed in response to these challenges. Quantum-aware optimizers aim to replace off-the-shelf classical optimizers with ones that are specifically tailored to the quantum setting~\cite{stokes2020quantum,koczor2019quantum,nakanishi2020sequential}. Shot-frugal optimizers~\cite{kubler2020adaptive,gu2021adaptive,sweke2020stochastic, tamiya2022stochastic} have been proposed in the context of variational quantum eigensolver (VQE), whereby one can sample over terms in the Hamiltonian instead of measuring every term~\cite{arrasmith2020operator}. While significant progress has been made on such optimizers, particularly for VQE, we argue that very little work has been done to specifically tailor optimizers to the QML setting. The cost functions in QML go well beyond those used in VQE and hence QML requires more general tools.

In this work, we generalize previous shot-frugal and iteration-frugal optimizers, such as those in Refs.~\cite{kubler2020adaptive,gu2021adaptive,arrasmith2020operator}, by extending them to the QML setting. Specifically, we allow for random, weighted sampling over both the input and the output of the loss function estimation circuit. In other words, and as shown in Fig.~\ref{fig_schematic}, we allow for sampling over the dataset as well as over the measurement operators used to define the loss function. Our sampling approach allows us to unlock the frugality (i.e., to achieve the full potential) of adaptive stochastic gradient descent optimizers, such as the recently developed iCANS~\cite{kubler2020adaptive} and gCANS~\cite{gu2021adaptive}.

We discuss how our approach applies to various QML applications such as perceptron-based quantum neural networks~\cite{beer2020training,sharma2020trainability}, quantum autoencoders~\cite{romero2017quantum}, variational quantum principal component analysis (PCA)~\cite{larose2019variational,cerezo2020variational}, and classifiers that employ the mean-squared-error loss function~\cite{schuld2020circuit,sweke2020stochastic}. Each of these applications can be unified under one umbrella by considering a generic loss function with a highly general form. Thus we state our main results for this generic loss function. We establish an unbiased estimator for this loss function and its gradient.  In turn, this allows us to provide convergence guarantees for certain optimization routines, like stochastic gradient descent. Furthermore, we show that for this general loss function one can use the form of the estimator to inform a strategy that distributes shots to obtain the best shot frugal estimates. We also show how to construct an unbiased estimator of the log likelihood loss function, which can be used in gradient-free shot frugal optimization.

Finally, we numerically investigate the performance of our new optimization approach, which we call Refoqus (REsource Frugal Optimizer for QUantum Stochastic gradient descent). For a quantum PCA task and a quantum autoencoder task, Refoqus significantly outperforms state-of-the-art optimizers in terms of the shot resources required. Refoqus even outperforms Rosalin~\cite{arrasmith2020operator} - a shot-frugal optimizer that samples only over measurement operators. Hence, Refoqus will be a crucial tool to minimize the number of shots and iterations required in near-term QML implementations. 

\section{Background}\label{sec_Background}

\subsection{Stochastic Gradient Descent}

One of the most popular optimization approaches is gradient descent, which involves the following update rule for the parameter vector:
\begin{equation}
    \thv^{(t+1)} = \thv^{(t)} - \alpha \grad \LC(\thv^{(t)})\,. \label{eq:gd-update}
\end{equation}
Here, $\LC$ is the loss function, $\alpha$ is the learning rate, and $\thv^{(t)}$ is the parameter vector at iteration $t$. 

Oftentimes one only has access to noisy estimates of the gradient $\grad \LC$, in which case the optimizer is called stochastic gradient descent. In the quantum setting, this situation arises due to shot noise or noise due to sampling from terms in some expansion of the gradient.

\subsection{Parameter Shift Rule}

Estimating the gradient is clearly an essential step in stochastic gradient descent. For this purpose, one useful tool that is often employed in the quantum case is the so-called parameter shift rule~\cite{mitarai2018quantum,schuld2019evaluating}. We emphasize that several assumptions go into this rule.  
Specifically, suppose we assume that the quantum circuit ansatz $U(\thv)$  can be expressed as $U(\thv)=\prod_x e^{-i \theta_x \sigma_x}W_x$ where $W_x$ are unparametrized unitaries, and where $\sigma_x$ are Pauli operators. Moreover, we consider the case when the loss function has the simple form $\LC(\thv) = \expval{U^\dagger(\thv) H U(\thv)}{0}$ for some Hermitian operator $H$. (Note that we will consider more complicated loss functions in this work, and hence we are just stating this special case for background information.) In this case, the parameter shift rule gives
\begin{equation}\label{eq:analytic_derivative}
    \partial_{\gindex} \LC(\thv) :=\frac{\partial  \LC(\thv)}{\partial \theta_{\gindex}}= \frac{\LC(\thv+\frac{\pi}{2} \indicatorvec_{\gindex})-\LC(\thv-\frac{\pi}{2} \indicatorvec_{\gindex})}{2}\,,
\end{equation}
where $\indicatorvec_{\gindex}$ is a unit vector with a one on the  $\gindex$-th component. Equation~\eqref{eq:analytic_derivative} allows one to estimate the gradient by estimating the loss function as specific points on the landscape. Hence it simplifies the procedure to estimate the gradient.

\begin{figure}[t]
\centering
\includegraphics[width=1\columnwidth]{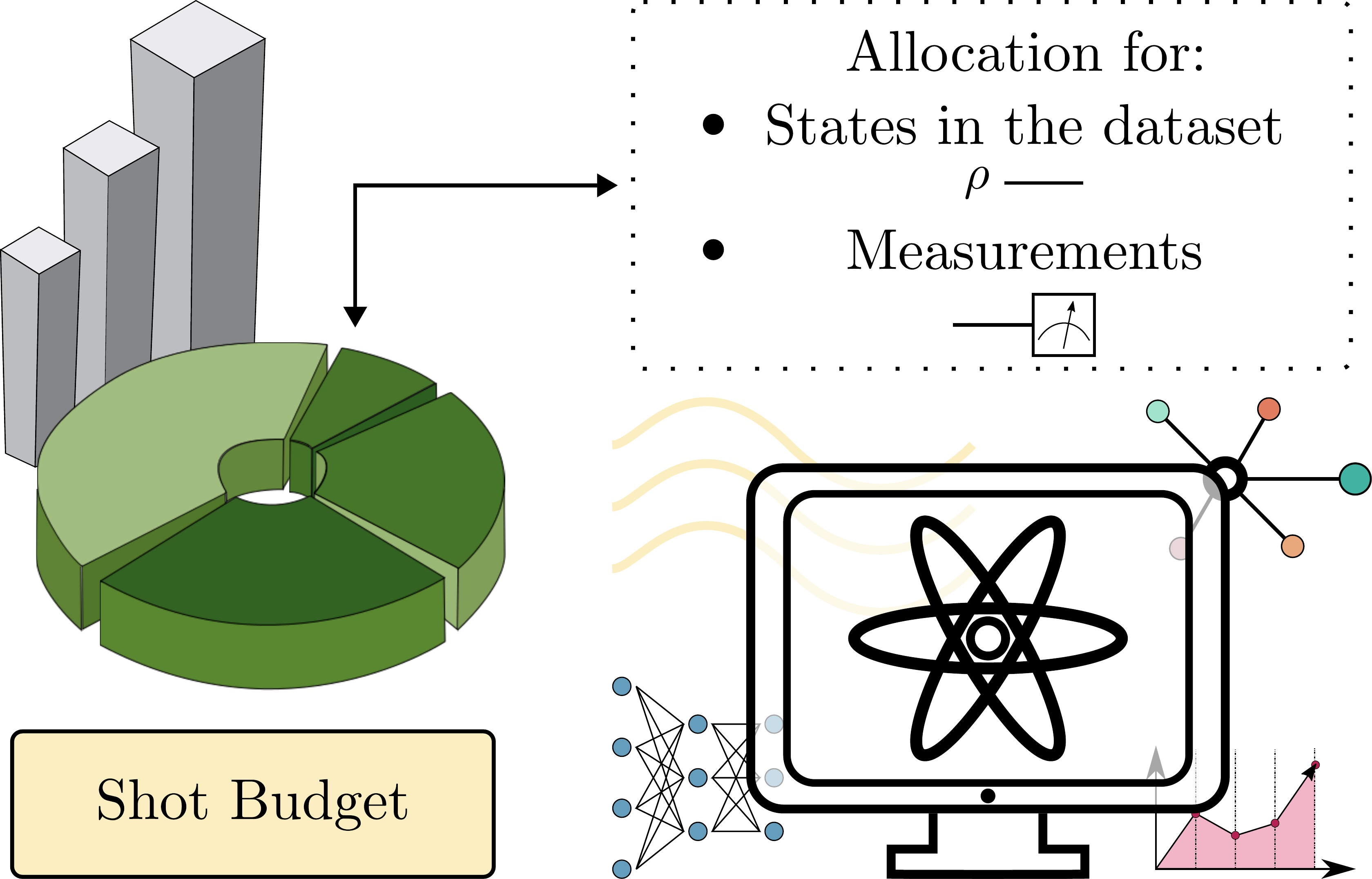}
\caption{\textbf{Schematic illustration of Refoqus.} Measurements, or  shots, are a precious (and expensive) resource in quantum computing. As such, they should be used sparingly and only when absolutely necessary. This is particularly important in variational QML methods where training a model requires continuously calling a  quantum device to estimate the loss function or its gradients. Refoqus provides a shot-frugal optimization paradigm where shots are allocated by sampling over the input data and the measurement operators. }
\label{fig_schematic}
\end{figure}

\subsection{Lipschitz continuity}

The loss function $\LC$ is called Lipschitz continuous if there is some Lipschitz constant $L \ge 0$ that satisfies the bound
\begin{equation}
    \norm{\grad{\LC}(\thv_a) - \grad{\LC}(\thv_b)} \le L \norm{\thv_a - \thv_b}\,, \label{eq:lip-cont}
\end{equation}
for all $\thv_a, \thv_b \in \text{dom}(\LC)$, where $\norm{\cdot}$ is the $\ell_2$ norm. This property provides a recipe for choosing an appropriate learning rate, $\alpha$. Specifically, if~\eqref{eq:lip-cont} holds and we have access to the exact gradient, then choosing $\alpha \le \flatfrac{2}{L}$ is sufficient to guarantee convergence using the update rule in~\eqref{eq:gd-update}.

\subsection{iCANS}

Inspired by an adaptive batch size optimizer used in classical machine learning~\cite{balles2017coupling}, the iCANS (individual coupled adaptive number of shots) optimizer~\cite{kubler2020adaptive} was introduced as an adaptive method for stochastic gradient in the context of the variational quantum eigensolver. It allows the number of shots per partial derivative (i.e., gradient component) to vary individually, hence the name iCANS. 

Consider the gain (i.e., the decrease in the loss function), denoted $\mathcal{G}_{\gindex}$, associated with updating the $\gindex$-th parameter $\theta_{\gindex}$. The goal of iCANS is to maximize the expected gain per shot. That is, for each individual partial derivative, we maximize the shot efficiency:
\begin{equation}
    \gamma_{\gindex} := \frac{\mathbb{E}[\mathcal{G}_{\gindex}]}{s_{\gindex}}; \quad \gindex=1,\ldots,\gdimension\,, \label{eqn:icans-fom}
\end{equation}
where $s_{\gindex}$ is the shot allocation for gradient component $\gindex$ and $\gdimension$ is the number of gradient components. Solving for the optimal shot allocation gives: 
\begin{equation}
    s_{\gindex} = \frac{2L \alpha}{2-L \alpha} \frac{\sigma_{\gindex}^2}{g_{\gindex}^2}.
\end{equation}
Here, $g_{\gindex}$ is an unbiased estimator for the $\gindex$-th gradient component, and $\sigma_{\gindex}$ is the standard deviation of a random variable $X_{\gindex}$ whose sample mean is $g_{\gindex}$. While iCANS often heuristically outperforms other methods, it can have instabilities.

\subsection{gCANS}

Recently, a potential improvement over iCANS was introduced called gCANS (global coupled adaptive number of shots)~\cite{gu2021adaptive}. gCANS considers the expected gain $\mathbb{E}[\mathcal{G}]$ over the entire gradient vector. Then the goal is to maximize the shot efficiency 
\begin{equation}
    \gamma := \frac{\mathbb{E}[\mathcal{G}]}{\sum_{\gindex = 1}^{\gdimension} s_{\gindex}}\,, \label{eqn:gcans-fom}
\end{equation}
where the sum $\sum_{\gindex=1}^{\gdimension} s_{\gindex}$ goes over all components of the gradient. Solving for the optimal shot count then gives:
\begin{equation}
    s_{\gindex} = \frac{2L \alpha}{2-L \alpha} \frac{\sigma_{\gindex} \sum_{\gindex ' = 1}^{\gdimension} \sigma_{\gindex '}}{\norm{\grad{\LC(\thv)}}^2}.
\end{equation}
(Note that an exponential moving average is used to estimate $\sigma_{\gindex}$ and $\norm{\grad{\LC(\thv)}}^2$ as their true value is not accessible.) It was proven that gCANS achieves geometric convergence to the optimum, often reducing the number of shots spent for comparable solutions against its predecessor iCANS.

\subsection{Rosalin}

Shot-frugal optimizers like iCANS and gCANS rely on having an unbiased estimator for the gradient or its components. However, this typically places a hard floor on how many shots must be allocated at each iteration, i.e., going below this floor could result in a biased estimator. This is because the measurement operator $H$ is typically composed of multiple non-commuting terms, each of which must be measured individually. Each of these terms must receive some shot allocation to avoid having a biased estimator. However, having this hard floor on the shot requirement is antithetical to the shot-frugal nature of iCANS and gCANS, and ultimately it handicaps these optimizers' ability to achieve shot frugality.
\par
This issue inspired a recent proposal called Rosalin (Random Operator Sampling for Adaptive Learning with Individual Number of shots)~\cite{arrasmith2020operator}. Rosalin employs weighted random sampling of operators in the measurement operator $H = \sum_j c_j H_j$, which allows one to achieve an unbiased estimator without a hard floor on the shot requirement. (Even a single shot, provided that it is randomly allocated according to an appropriate probability distribution, can lead to an unbiased estimator.) When combined with the shot allocation methods from iCANS or gCANS, the operator sampling methods in Rosalin were shown to be extremely powerful in the context of molecular chemistry problems, which often have a large number of terms in $H$.

We remark that Ref.~\cite{arrasmith2020operator} considered several sampling strategies. Given a budget of $s_{\text{tot}}$ shots and $N$ terms, a simple strategy is to distribute shots per term equally ($s_{j} = s_{\text{tot}} / N$) - referred as uniform deterministic sampling (UDS). Defining $M = \sum_{j} |c_{j}|$, one can also use weighted deterministic sampling (WDS) where the shots are proportionally distributed: $s_{j} = s_{\text{tot}} * \frac{|c_{j}|}{M}$. One can add randomness by using $p_{j} = \frac{|c_{j}|}{M}$ to define a (non-uniform) probability distribution to select which term should be measured. This is referred to as weighted random sampling (WRS). Finally, there exists a hybrid approach where one combines WDS with WRS - referred to as weighted hybrid sampling (WHS). Ref.~\cite{arrasmith2020operator} found that the WRS and WHS strategies performed similarly and they both significantly outperformed the UDS and WDS strategies on molecular ground state problems. Because of these results, we choose to focus on the WRS strategy in our work here.

While Rosalin was designed for chemistry problems, it was not designed for QML, where the number of terms in $H$ is not the only consideration. As discussed below, QML problems involve a (potentially large) dataset of input states. Each input state requires a separate quantum circuit, and hence we are back to the situation of having a hard floor on the shots required due to these multiple input states. This ultimately provides the motivation for our work, which can be viewed as a generalization of Rosalin to the setting of QML.

\section{Framework}
\label{section:framework}
\subsection{Generic Variational QML Framework}

\begin{figure*}[t]
\centering
\includegraphics[width=1\textwidth]{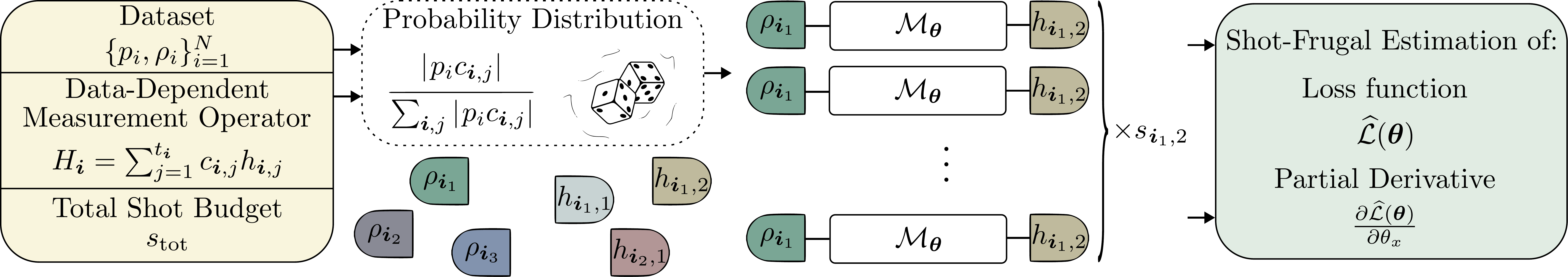}
\caption{\textbf{Illustration of a generic variational QML framework.} Given a dataset of quantum states $\SC$, the input states to the QML model are tensor product states of the form $\rho_{\vec{i}}=\rho_{i_1} \otimes ... \otimes \rho_{i_\hscopies}$. The number of samples $\hscopies$ depends on the QML task at hand. The parameterized quantum model denoted $\mathcal{M}_{\thv}$, acts on $\hscopies$ copies of the input Hilbert space. Finally, an operator $H_{\vec{i}}$ is measured to estimate a quantity to be used when evaluating a loss function $\LC(\thv)$. The latter evaluation is then inputted to a classical optimizer that proposes new parameters $\thv$ in order to minimize the loss. Hence, one can repeat the quantum-classical loop of evaluations and updates until the desired stopping criteria are satisfied.}
\label{fig_Framework}
\end{figure*}

Let us present our general framework for discussing (variational) QML methods; see Fig.~\ref{fig_Framework} for an illustration. This framework is meant to unify multiple literature QML algorithms under one umbrella. We discuss how specific literature algorithms are special cases of this framework in Section~\ref{sct_examples}.

In a generic QML setting, one has a training dataset composed of quantum states:
\begin{equation}
    \SC = \{\rho_i\}_{i=1}^N\,,
\end{equation}
where each $\rho_i$ is a trace-one positive semi-definite matrix, i.e., a density matrix. Each of these training states may come with an associated probability, with the associated probability distribution denoted as
\begin{equation}
    \PC = \{p_i\}_{i=1}^N\,.
\end{equation}

The latter distribution is generally uniform but does not need to be (i.e., datasets with non-uniform probability distributions are used in \cite{liu2018differentiable}).
In variational QML, one trains a parameterized quantum model, which we write as $\mathcal{M}_{\thv}$ for some set of parameters $\thv$. With a large degree of generality, we can assume that $\mathcal{M}_{\thv}$ is a linear, completely-positive map. In general, $\mathcal{M}_{\thv}$ could act on multiple copies ($\hscopies$ copies) of the input Hilbert space. (Multiple copies allow for non-linear operations on the dataset, which can be important in certain classification tasks.) Hence we allow for the output of this action to be of the form: 
\begin{equation}
    \mathcal{M}_{\thv} (\rho_{\vec{i}}) = \mathcal{M}_{\thv} (\rho_{i_1} \otimes ... \otimes \rho_{i_{\hscopies}}) = \mathcal{M}_{\thv}\left( \bigotimes_{\alpha=1}^{\hscopies} \rho_{i_{\alpha}}\right)\,,
\end{equation}
where we employ the notation $\rho_{\vec{i}} := \rho_{i_1} \otimes ... \otimes \rho_{i_{\hscopies}}$.  Given that we are allowing for multiple copies of the input space, one can define an effective dataset $\SC_{\hscopies}$ composed of the tensor product of $m$ states in $\SC$ and an effective probability distribution $\PC_{\hscopies}$.

A QML loss function is then defined in an operational manner so that it could be estimated on a quantum device. This involves considering the aforementioned mathematical objects as well as a measurement operator, or a set of measurement operators. We allow for the measurement operator to be tailored to the input state. Hence we write $H_{\vec{i}}$, with $\vec{i} = \{i_1, ..., i_{\hscopies}\}$, as the measurement operator when the input state on $\hscopies$ copies of the Hilbert space is $\rho_{\vec{i}}=\rho_{i_1} \otimes ... \otimes \rho_{i_{\hscopies}}$. Moreover, each measurement operator can be decomposed into a linear combination of Hermitian matrices that can be directly measured: 
\begin{equation}\label{measurement_decomposed}
    H_{\vec{i}} = \sum_{j=1}^{t_{\vec{i}}} c_{\vec{i},j} h_{\vec{i},j}\,.
\end{equation}

Generically, we could write the loss function as an average over training states chosen from the effective dataset $D_{\hscopies}$:
\begin{equation}\label{eqn_generalloss}
    \LC(\thv) = \sum_{\vec{i}} p_{\vec{i}} \ell (E_{\vec{i}}(\thv))\,.
\end{equation}
Here, $\ell$ is an application-dependent function whose input is a measurable expectation value $E_{\vec{i}}(\thv)$. Specifically, this expectation value is associated with the $\rho_{\vec{i}}$ input state, with the form
\begin{equation}
E_{\vec{i}}(\thv) = \Tr [\mathcal{M}_{\thv} (\rho_{\vec{i}}) H_{\vec{i}}]\,.
\end{equation}

There are many possible forms for the function $\ell$. However, there are multiple QML proposals in the literature that involve a simple linear form:
\begin{equation}\label{eqn_linearform}
\ell (E_{\vec{i}}(\thv)) = E_{\vec{i}}(\thv)\,,
\end{equation}
and in this case, we refer to the overall loss function as a ``linear loss function''. Alternatively, non-linear functions are also possible, and we also consider polynomial functions of the form:
\begin{equation}\label{eqn_polyform}
\ell (E_{\vec{i}}(\thv)) = \sum_{z=0}^{\polydegree} a_z [E_{\vec{i}}(\thv)]^z\,,
\end{equation}
where $\polydegree$ is the degree of the polynomial. In this case, we refer to the loss as a ``polynomial loss function''.

\subsection{Examples of QML loss functions}\label{sct_examples}

Now let us illustrate how various QML loss functions proposed in the literature fall under the previous framework. Crucially, as shown in Fig.~\eqref{fig:applications} these loss functions can be used for a wide range of QML tasks.

\subsubsection{Variational quantum error correction}

Variational quantum algorithms to learn device-tailored quantum error correction codes were discussed in Refs.~\cite{johnson2017qvector,cong2019quantum}. The loss function involves evaluating the input-output fidelity of the code, averaged over a set of input states. The input state is fed into the encoder $\EC_{\thv_1}$, then the noise channel $\NC$ acts, followed by the decoder $\DC_{\thv_2}$. The concatenation of these three channels can be viewed as the overall channel
\begin{equation}
    \MC_{\thv} = \DC_{\thv_2} \circ \NC \circ \EC_{\thv_1}\,,
\end{equation}
with parameter vector $\thv = (\thv_1, \thv_2)$. Then the loss function is given by:
\begin{equation}
    \LC(\thv) = \sum_{\ket{\psi_i} \in \SC } \frac{1}{|\SC|} \Tr [\dya{\psi_i} \MC_{\thv}(\dya{\psi_i})]\,,
\end{equation}
where $\SC$ is some appropriately chosen set of states. It is clear that this loss function is of the form in \eqref{eqn_generalloss} with $\ell$ having the linear form in \eqref{eqn_linearform}.

\begin{figure}[t]
\centering
\includegraphics[width=1\columnwidth]{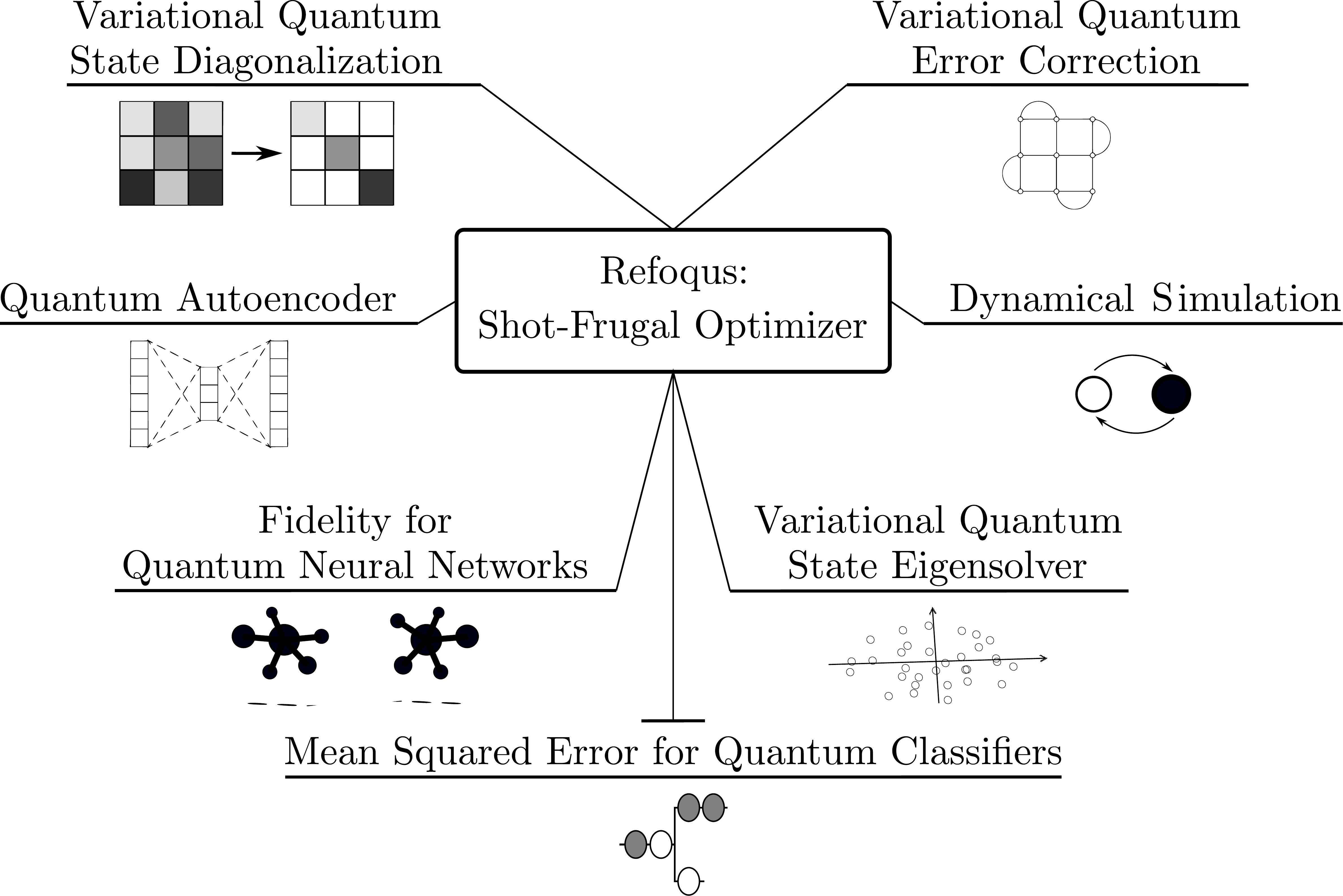}
\caption{\textbf{Applications benefiting from Refoqus}. Many variational quantum algorithms employ training data, and many QML models (which naturally analyze datasets) are variational in nature. We give examples of both of these cases in this figure. Our Refoqus optimizer is relevant in all cases. }
\label{fig:applications}
\end{figure}

\subsubsection{Quantum autoencoder}

Inspired by the success of classical autoencoders, quantum autoencoders~\cite{romero2017quantum, bondarenko2020quantumautoencoders} were proposed to compress quantum data by reducing the number of qubits needed to represent the dataset. Consider a bipartite quantum system $AB$ composed of $n_{\text{A}}$ and $n_{\text{B}}$ qubits, respectively, and let $\{p_i, \ket{\psi_i}\}$ be an ensemble of pure states on $AB$. The quantum autoencoder trains a gate sequence $U(\thv)$ to compress this ensemble into the $A$ subsystem, such that one can recover each state $\ket{\psi_i}$ with high fidelity from the information in subsystem $A$.  One can think of $B$ as the ``trash'' since it is discarded after the action of $U(\thv)$.

The original proposal~\cite{romero2017quantum} employed a loss function that quantified the overlap of the trash state with a fixed pure state:
\begin{align}\label{eq:CG-AE1}
    \LC_{G}(\thv) &=1- \Tr_B[\dya{\vec{0}}\rho_{\text{B}}^{\text{out}}] \\
\label{eq:CG-AE2}   & = \Tr_{AB}[H_{G} U(\thv)\rho_{\text{AB}}^{\text{in}}U(\thv)\ad]\\
&\label{eq:CG-AE3}
=\sum_i  p_i \Tr_{AB}[H_{G} U(\thv)\dya{\psi_i} U(\thv)\ad]\,.
\end{align}
Here, $\rho_{\text{AB}}^{\text{in}}=\sum_i p_i \dya{\psi_i}$ is the ensemble-average input state, $\rho_{\text{B}}^{\text{out}}= \Tr_A[U(\thv)\rho_{\text{AB}}^{\text{in}}U(\thv)\ad]$ is the ensemble-average trash state, and the measurement operator is $H_{G} = \openone_{AB} - \openone_{A} \otimes \ketbra{\vec{0}}$. 

Note that $H_{G}$ is a global measurement operator, meaning it acts non-trivially on all qubits, which can lead to barren plateaus in the training landscape~\cite{cerezo2020cost}. To remedy this issue, Ref.~\cite{cerezo2020cost} proposed a loss function with a local measurement operator acting non-trivially  only on a small number of qubits:
\begin{align}\label{eq:CL-AE}
    \LC_{L}(\thv) &= 1- \frac{1}{n_{\text{B}}}\sum_{j=1}^{n_{\text{B}}}\Tr_B\left[\left(\dya{0}_j\otimes\id_{\overline{j}}\right)\rho_{\text{B}}^{\text{out}}\right] \\
    \label{eq:CL-AE2}   & = \sum_i p_i \Tr_{AB}[H_{L} U(\thv)\dya{\psi_i}U(\thv)\ad]\,,
\end{align}
where $H_{L} = \id_{\text{AB}} - \frac{1}{n_{\text{B}}}\sum_{j=1}^{n_{\text{B}}} \id_{\text{A}}\otimes \dya{0}_j\otimes \id_{\overline{j}}$, and $\id_{\overline{j}}$ is the identity on all qubits in $B$ except the $j$-th qubit.

It is clear from \eqref{eq:CG-AE3} and \eqref{eq:CL-AE2} that both loss functions fall under our framework. Namely, they have the form in \eqref{eqn_generalloss} with $\ell$ having the linear form in \eqref{eqn_linearform}.

\subsubsection{Dynamical simulation}

Recently a QML-based algorithm was proposed for dynamical simulation~\cite{gibbs2022dynamical}. Here the idea is to variationally compile the short-time dynamics into a time-dependent quantum neural network~\cite{cirstoiu2020variational}. Then one uses the trained model to extrapolate to longer times. The training data for the compiling process can be taken to be product states, due to a generalization result from Ref.~\cite{caro2022outofdistribution}. 

Let $U_{\Delta t}$ be a unitary associated with the short-time dynamics. Let $\{\ket{\Psi^{P}_i}\}_{i=1}^N$ be a set of product states used for training, where $\ket{\Psi^{P}_i} = \bigotimes_{j=1}^n \ket{\psi_{i,j}}$, and where $n$ is the number of qubits.  Let $U(\thv)$ be the quantum neural network to be trained. Then the loss function is given by:
\begin{equation}
    \LC(\thv) = \sum_{i=1}^N \frac{1}{N}  \Tr[H_i V(\thv)\dya{\Psi^{P}_i} V(\thv)\ad]\,,
\end{equation}
where we have defined $V(\thv) := U(\thv)\ad U_{\Delta t}$. Here, the measurement operator is given by $H_i = \id - \frac{1}{n}\sum_{j=1}^n \dya{\psi_{i,j}}\otimes \id_{\overline{j}}$, where $\overline{j}$ is the set of all qubits excluding the $j$-th qubit.

Once again, this loss function clearly falls under our framework, having the form in \eqref{eqn_generalloss} with $\ell$ having the linear form in \eqref{eqn_linearform}.

\subsubsection{Fidelity for Quantum Neural Networks}

Dissipative perceptron-based quantum neural networks (DQNNs) were proposed in Ref.~\cite{beer2020training} and their trainability was analyzed in Ref.~ \cite{sharma2020trainability}. The loss function is based on the fidelity between the idealized output state and the actual output state of the DQNN. Specifically, we are given access to training data $\{ \ket{\phi_i^{\tin}}, \ket{\phi_i^{\tout}} \}_{i = 1}^N$, and the DQNN is trained to output a state close to $\ket{\phi_i^{\tout}}$ when the input is $\ket{\phi_i^{\tin}}$. 

For this application, a global loss function was considered with the form
\begin{equation}\label{eq:fidelity_global}
    \LC_G(\thv)=\sum_{i=1}^N \frac{1}{N} \Tr[H_i^G \rho^\tout_i]\,.
\end{equation}
Here, $\rho^\tout_i = \MC_{\thv}(\dya{\phi_i^{\tin}})$ is the output state of the DQNN, which is denoted by $\MC_{\thv}$. The measurement operator is the projector orthogonal to the ideal output state: $H_i^G = \id - \ketbra{\phi_i^{\tout}}{\phi_i^{\tout}}$.

To avoid the issue of barren plateaus, a local loss function was also considered in Ref.~\cite{sharma2020trainability}: 
\begin{equation}\label{eq:fidelity_local}
    \LC_L(\thv)=\sum_{i=1}^N \frac{1}{N} \Tr[H_i^L \rho^\tout_i]\,,
\end{equation}
where the measurement operator is
\begin{equation}\label{eq:localop}
    H_i^L=\id - \frac{1}{n_\tout}\sum_{j=1}^{n_\tout}\dya{\psi^{\tout}_{i,j}}\otimes \id_{\overline{j}}\,.
\end{equation}
This loss function is relevant whenever the ideal output states have a tensor-product form across the $n_\tout$ output qubits, i.e., of the form $\ket{\phi_i^{\tout}} = \ket{\psi^{\tout}_{i,1}} \otimes \cdots \otimes \ket{\psi^{\tout}_{i,n_\tout}} $.

Clearly, these two loss functions fall under our framework, having the form in \eqref{eqn_generalloss} with $\ell$ having the linear form in \eqref{eqn_linearform}.

\subsubsection{Variational quantum state eigensolver}

Near-term methods for quantum principal component analysis have recently been proposed, including the variational quantum state eigensolver (VQSE) \cite{cerezo2020variational} and the variational quantum state diagonalization (VQSD) algorithm. Let us first discuss VQSE.

The goals of VQSE are to estimate the $m$ largest eigenvalues of a density matrix $\rho$ and to find quantum circuits that prepare the associated eigenvectors. When combined with a method to prepare the covariance matrix as a density matrix, VQSE can be used for principal component analysis. Note that such a method was proposed in Ref.~\cite{gordon2022covariance}, where it was shown that choosing $\rho = \sum_i p_i \dya{\psi_i}$ prepares the covariance matrix for a given quantum dataset $\{p_i, \ket{\psi_i}\}$.

The VQSE loss function can be written as an energy:
\begin{align}
    \LC(\thv) &= \Tr [H U(\thv)\rho U(\thv)\ad]\label{eqn_vqse_1} \\
    &= \sum_i p_i \Tr [H U(\thv)\dya{\psi_i} U(\thv)\ad]\label{eqn_vqse_2}
\end{align}
where $U(\thv)$ is a parameterized unitary that is trained to approximately diagonalize $\rho$. Note that we inserted the formula $\rho = \sum_i p_i \dya{\psi_i}$ in order to arrive at \eqref{eqn_vqse_2}.

The measurement operator $H$ is chosen to be non-degenerate over its $m$-lowest energy levels. For example, one can choose a global version of this operator:
\begin{equation} \label{eq:VQSE:global}
    H = \id - \sum_{j=1}^m r_j \dya{e_j},\quad r_j > 0
\end{equation}
or a local version of this operator:
\begin{equation} \label{eq:VQSE:local}
    H = \id - \sum_{j=1}^n r_j Z_j,\quad r_j \in \mathbb{R}\,,
\end{equation} 
where $Z_j$ is the Pauli-$z$ operator on the $j$-th qubit, and with appropriately chosen real coefficients  $r_j$  to achieve non-degeneracy over the $m$-lowest energy levels. Regardless, it is clear that the loss function in \eqref{eqn_vqse_2} falls under our framework, having the form in \eqref{eqn_generalloss} with $\ell$ having the linear form in \eqref{eqn_linearform}.

\subsubsection{Variational quantum state diagonalization}

The goal of the VQSD algorithm~\cite{larose2019variational, kundu2023enhancing} 
is essentially the same as that of VQSE, i.e., to diagonalize a target quantum state $\rho$. Let us use:
\begin{equation} \label{eqn:def-of-rho-prime}
    \rhot := U(\thv) \rho U(\thv)^{\dag}
\end{equation}
to denote the state after the attempted diagonalization. Here we omit the $\thv$ dependency for simplicity of notation. 

In contrast to VQSE, the VQSD loss function depends quadratically on the quantum state. Specifically, global and local cost functions have been proposed, respectively given by
\begin{align}
    \LC_G(\thv) &= \Tr[\rho^2] - \Tr [\mathcal{Z}(\rhot)^2],\\
    \LC_L(\thv) &= \Tr[\rho^2] - \frac{1}{n} \sum_{j=1}^n \Tr [\mathcal{Z}_j(\rhot)^2]\,.
\end{align}
Here, $\mathcal{Z}$ and $\mathcal{Z}_j$ are quantum channels that dephase (i.e., destroy the off-diagonal elements) in the global standard basis and in the local standard basis on qubit $j$, respectively.

We can rewrite the terms in the global loss using:
\begin{align}
     \Tr[\rho^2] &= \Tr[(\rho \otimes \rho) \text{SWAP}]\,, \\   \Tr [\mathcal{Z}(\rhot)^2] &= \Tr[(\rho \otimes \rho) W_G\ad (\ketbra{\vec{0} } \otimes \id) W_G]\,,
\end{align}
where $SWAP$ denotes the swap operator, and 
where $W_G$ corresponds to the layers of CNOTs used in the so-called diagonalized inner product (DIP) test circuit~\cite{larose2019variational}. 
Hence, we obtain:
\begin{align}\label{eq:globalvqsd1}
    \LC_{G}(\thv) &= \Tr[(\rho \otimes \rho) H_{G}]\\
    & = \sum_{i,i'} p_i p_{i'}\Tr[(\dya{\psi_i} \otimes \dya{\psi_{i'}}) H_{G}]\label{eq:globalvqsd2}
\end{align}
where $H_G = \text{SWAP} - U_G\ad (\ketbra{\vec{0} } \otimes \id) U_G$. Note that we inserted the relation $\rho = \sum_i p_i \dya{\psi_i}$ in order to arrive at \eqref{eq:globalvqsd2}. Also note that one can think of the $q_{\vec{i}}:= p_i p_{i'} $ as defining a probability distribution over the index $\vec{i}=\{i,i'\}$. Hence it is clear that \eqref{eq:globalvqsd2} falls under our framework, with $m=2$ copies of the input Hilbert space, and with $\ell$ having the linear form in \eqref{eqn_linearform}.

Similarly, for the local loss, we can write
\begin{align}\label{eq:localvqsd1}
    \LC_{L}(\thv) &= \Tr[(\rho \otimes \rho) H_{L}]\\
    & = \sum_{i,i'} p_i p_{i'}\Tr[(\dya{\psi_i} \otimes \dya{\psi_{i'}}) H_{L}]\label{eq:localvqsd2}
\end{align}
where $H_{L} = \text{SWAP} - (1/n)\sum_{j=1}^n H_{L,j}$, and $ H_{L,j}$ is the Hermitian operator that is measured for the partial diagonalized inner product (PDIP) test~\cite{larose2019variational}. Once again, the local loss falls under our framework, with $\hscopies = 2$ copies of the input Hilbert space, and with $\ell$ having the linear form in \eqref{eqn_linearform}.

\subsubsection{Mean squared error for quantum classifiers}

The mean squared error (MSE)  loss function is widely employed in classical machine learning. Moreover, it has been used in the context of quantum neural networks, e.g., in Ref.~\cite{cong2019quantum}. Given a set of labels $y_i \in \mathbb{R}$ for a dataset $\{\rho_i\}_{i=1}^N$, and a set of predictions from a machine learning model denoted $\tilde{y}_i(\thv)$, the MSE loss is computed as:
\begin{equation}\label{mse_general}
    \LC(\thv) = \frac{1}{N}\sum_{i=1}^N (y_i - \tilde{y}_i(\thv))^2 \,.
\end{equation}

In the case of a quantum model, there is freedom in specifying how to compute the prediction $\tilde{y}_i(\thv)$. Typically, this will be estimated via an expectation value, such as $\tilde{y}_i(\thv) = \Tr[\MC_{\thv}(\rho_i) H_i] = E_i(\thv)$. In this case, the loss function in \eqref{mse_general} would be a quadratic function of expectation value $E_i(\thv)$. Hence, this falls under our framework, with $\ell$ having the polynomial form in \eqref{eqn_polyform} with degree $\polydegree = 2$. 

\subsection{Possible extensions to our framework}

Up until now, we have considered loss functions that have a linear or polynomial dependence on measurable expectation values, $E_i(\thv)$. As shown above, this form encompasses many loss functions used in the literature and we will show that cost functions of this form enable cheap gradient estimation and therefore shot frugal optimization. 

However, there are also more complicated loss functions that have a less trivial dependence on measurable expectation values. One such loss function is the log-likelihood which is commonly used in classification tasks \cite{abbas2020power,thanasilp2021subtleties}.

Let us assume one is provided with a dataset of the form $\SC=\{\rho_i,y_i\}_{i=1}^N$ where $\rho_i$ are quantum states, and $y_i$ are discrete real-valued labels. We aim to train a parametrized model  whose goal is to predict the correct labels. Here one cannot directly apply the techniques previously used, as discrete outputs mean we cannot simply evaluate gradients. Therefore we cannot directly employ the results in~\cite{ balles2017coupling,kubler2020adaptive}. Hence, a different shot-frugal method must be employed. 

Given $\SC$, the likelihood function is defined as 
\begin{equation}
    \Pr(\SC|\thv):=\prod_{i=1}^N\Pr(y_i|\rho_i;\thv)=\prod_{i=1}^Np_i\,,
\end{equation}
where we have defined $p_i=\Pr(y_i|\rho_i;\thv)$ and where it is assumed that the model's predictions only depend on the current input. From here, we can define the negative log-likelihood loss function as
\begin{equation}\label{eq:neg-log-loss}
    \LC(\thv)=-\log\left(\Pr(\SC|\thv)\right)=-\sum_{i=1}^N\log\left(p_i\right)\,.
\end{equation}

In Appendix~\ref{app:log_likelihood} we present a conceptually different approach for shot frugal optimization, which does not rely on the unbiased estimation of gradients. Rather, this approach relies on the construction of an unbiased estimator for the log-likelihood loss function. In particular, we follow the results in~\cite{van2020unbiased} and use inverse binomial sampling in the context of quantum machine learning. We note that this extension is of independent interest to gradient-free shot-frugal optimizers (as one cannot leverage gradients for problems with discrete outputs) but leave further exploration to future work.

\section{Unbiased estimators for gradients of QML losses}

Gradient-based optimizers are commonly used when optimizing QML models. In shot frugal versions of these approaches, one of the key aspects is the construction of an unbiased estimator of the gradient \cite{sweke2020stochastic, kubler2020adaptive,gu2021adaptive}. There are two types of loss functions we will consider in this work; those that have a linear dependence on expectation values and those that have a non-linear dependence given by a polynomial function. We will see that for these two types of losses one can construct unbiased estimators of the gradients and that it is also possible to define sampling strategies that massively reduce the number of shots needed to evaluate such an estimator. 

Previous work considered how to construct unbiased estimators of QML gradients. However, the shot frugal resource allocation strategies presented were sub-optimal \cite{sweke2020stochastic} and leave room for improvement. Furthermore, the more sophisticated shot allocation methods presented in~\cite{arrasmith2020operator} were purpose-built for VQE-type cost functions. Here we unify approaches from these two works and show how one can employ more sophisticated shot allocation strategies in a general QML setting. 

First, we consider the simpler case of linear loss functions where we show one can directly employ the parameter shift rule~\cite{mitarai2018quantum,schuld2019evaluating} to construct an unbiased estimator for the gradient. Then we turn our attention to polynomial loss functions where we present a general form for an unbiased estimator of the gradient. In both cases, we show that shots can be allocated according to the expansion coefficients in the expressions we derive. These in turn depend on the coefficients in the operators to be measured and the set of quantum states used in the QML data set. Such shot-allocation schemes are an important ingredient in the design of our QML shot-frugal optimizer in the next section. We use the U-statistic formalism~\cite{hoeffding1948class} to construct unbiased estimators for the loss function in each case. For an introduction to U-statistics, we refer the reader to Appendix~\ref{app:U-statistics}.

\subsection{Loss functions linear in the quantum circuit observables}
\subsubsection{Using the parameter shift rule}
Loss functions that have a linear dependence on expectation values of observables are straightforward to consider. As shown in Ref.~\cite{sweke2020stochastic}, the parameter shift rule can be used directly to construct unbiased estimators of the gradient of these loss functions. Previously, we wrote our general linear loss function as
\begin{equation}
    \LC(\thv) = \sum_{\vec{i}} p_{\vec{i}} \ell (E_{\vec{i}}(\thv))\,.
\end{equation}
Let us not consider the case where 
\begin{equation}
E_{\vec{i}}(\thv) = \Tr (\mathcal{M}_{\thv} (\rho_{\vec{i}}) H_{\vec{i}})\, \quad
\textnormal{and} \quad
\ell (E_{\vec{i}}(\thv)) = E_{\vec{i}}(\thv)\,.\nonumber
\end{equation}

By expanding the measurement operator as $H_{\vec{i}} = \sum_{j=1}^{t_{\vec{i}}} c_{\vec{i},j} h_{\vec{i},j}$, we can then write this loss function as follows
\begin{equation}
    \LC(\thv)= \sum_{\vec{i}, j} q_{\vec{i},j} \langle h_{\vec{i},j}(\thv) \rangle,
\end{equation}
where $q_{\vec{i}, j} = p_{\vec{i}} c_{\vec{i}, j} $ and $\Tr[ \mathcal{M}_{\thv}(\rho_{\vec{i}}) h_{\vec{i}, j}] =\langle h_{\vec{i},j} (\thv) \rangle$.

Consider the partial derivative of this loss function with respect to the parameter $\theta_{\gindex}$, \begin{equation}
    \frac{\partial \LC}{\partial \theta_{\gindex}} = \sum_{\vec{i}, j} q_{\vec{i},j} \frac{\partial \langle h_{\vec{i},j}(\thv)\rangle}{\partial \theta_{\gindex}}\,. 
\end{equation}
Suppose we assume that the quantum channel $\mathcal{M}_{\thv}(\rho_{i})$ is a unitary channel:
\begin{equation}
    \mathcal{M}_{\thv}(\rho_{i}) = U(\thv)\rho_{i} U(\thv)\ad\,,
\end{equation}
where $U(\thv)$ is a trainable quantum circuit whose parametrized gates are generated by Pauli operators. This assumption allows us to directly apply the  parameter shift rule (see Sec.~\ref{sec_Background}). This leads to
\begin{equation}
        \frac{\partial \LC}{\partial \theta_{\gindex}} = \frac{1}{2}\sum_{\vec{i}, j} q_{\vec{i},j} \big( \langle h_{\vec{i},j}(\thv+\indicatorvec_{\gindex}\frac{\pi}{2})\rangle -  \langle h_{\vec{i},j}(\thv-\indicatorvec_{\gindex}\frac{\pi}{2})\rangle\big)\,.
\end{equation}

Therefore, an unbiased estimator for the gradient can be obtained by combining two unbiased estimators for the loss function. Defining $\widehat{g}_{\gindex}(\thv)$ to be an unbiased estimator of the $\gindex$-th component of the gradient,
\begin{equation}\label{eq:gradient_estimator}
    \widehat{g}_{\gindex}(\thv) = \frac{1}{2}[\widehat{\LC}(\thv+\indicatorvec_{\gindex}\frac{\pi}{2}) - \widehat{\LC}(\thv+\indicatorvec_{\gindex}\frac{\pi}{2})]\,,
\end{equation}
where $\widehat{\LC}(\thv \pm \indicatorvec_{\gindex}\frac{\pi}{2})$ are unbiased estimators for the loss function at the different shifted parameter values needed when employing the parameter shift rule. This means that for loss functions that are linear in the expectation values recovered from the quantum circuit, one can then use an unbiased estimator for the cost evaluated at different parameter values to return an unbiased estimator for the gradient. 

We can therefore distribute the shots according to the coefficients $q_{\vec{i},j}$ when evaluating a single or multi-shot estimate of $\widehat{\LC}(\thv \pm \indicatorvec_{\gindex}\frac{\pi}{2})$. This can be done by constructing a probability distribution according to the probabilities $\epsilon_{\vec{i}, j} = |q_{\vec{i},j}|/ \sum_{\vec{i},j} |q_{\vec{i},j}|$. Note that this distribution strategy relies on the construction of two unbiased estimators of the loss function, which are then combined. In the next section, we show how one can construct such estimators. In practice, we will use the same total number of shots to evaluate both estimators as they have the same expansion weights and are therefore equally important. 

\subsubsection{Unbiased estimators for Linear loss functions}
\label{unbiasedestimatorslinearsection}
Let $s_{\text{tot}}$ represent the total number of shots. We denote $\widehat{\mathcal{E}}_{\vec{i},j} $ as the estimator for $\langle h_{\vec{i},j} (\thv) \rangle$, and $\widehat{\LC}(\thv)$ the estimator for the loss. That is,
\begin{equation}\label{eq:estimators}
    \widehat{\LC}(\thv)= \sum_{\vec{i}, j} q_{\vec{i},j}  \widehat{\mathcal{E}}_{\vec{i},j} \,,\quad\text{with}\quad
    \widehat{\mathcal{E}}_{\vec{i},j} =  \frac{1}{\mathbb{E}[s_{\vec{i},j}]}\sum_{k=1}^{s_{\vec{i},j}} r_{\vec{i},j,k}\,.
\end{equation}

Here, $s_{\vec{i}, j}$ is the number of shots allocated to the measurement of $\langle h_{\vec{i},j} (\thv) \rangle$. Note that $s_{\vec{i},j}$ may be a random variable. As we will work in terms of the total shot budget for the estimation, $s_{\text{tot}}$, we impose $\sum_{\vec{i}, j} s_{\vec{i}, j} = s_{\text{tot}} $. Also, each $r_{\vec{i},j,k}$ is an independent random variable associated with the $k$-th single-shot measurement of $\langle h_{\vec{i},j} (\thv) \rangle$. We will assume that $\mathbb{E}[s_{\vec{i},j}]>0$ for all $\vec{i},j$. Using these definitions we can show that this defines an unbiased estimator for the loss function in the following proposition.
\begin{proposition}
\label{prop-unbiased-mother-cost}
Let $\widehat{\LC}$ be the estimator defined in Eq.~\eqref{eq:estimators}. $\widehat{\LC}$ is an unbiased estimator for the cost function $\LC(\thv)$ defined in Eq.~\eqref{eqn_generalloss}.
\end{proposition}

\begin{proof}
\begin{align}
    \mathbb{E}[\widehat{\LC}] = \mathbb{E}[\sum_{\vec{i}, j} q_{\vec{i},j}   \frac{1}{\mathbb{E}[s_{\vec{i},j}]}\sum_{k=1}^{s_{\vec{i},j}} r_{\vec{i},j,k}] \nonumber\\
    = \sum_{\vec{i}, j} q_{\vec{i},j}   \frac{1}{\mathbb{E}[s_{\vec{i},j}]} \mathbb{E}[\sum_{k=1}^{s_{\vec{i},j}} r_{\vec{i},j,k}].
\end{align}
Here it is useful to recall Wald's equation. Wald's equation states that the expectation value of the sum of $N$ real-valued, identically distributed, random variables, $X_{i}$, can be expressed as
\begin{equation}
    \mathbb{E}[\sum\limits_{i=1}^{N} X_{i}] = \mathbb{E}[N]\mathbb{E}[X_{1}],
\end{equation}
where $N$ is a random variable that does not depend on the terms of the sum. In our case, each shot is indeed independent and sampled from the same distribution. Furthermore, the total number of shots does not depend on the sequence of single-shot measurements. Therefore,
\begin{align}
    \mathbb{E}[\widehat{\LC}]= \sum_{\vec{i}, j} q_{\vec{i},j}   \frac{\mathbb{E}[s_{\vec{i},j}]}{\mathbb{E}[s_{\vec{i},j}]} \langle h_{\vec{i},j} (\thv) \rangle = \LC(\thv).
\end{align}
\end{proof}
Alternatively, one can see that the estimator defined in Eq.~\eqref{eq:estimators} has the form of a degree-$1$ U-statistic for $\widehat{\LC}$. From the above result, we arrive at the following corollary.
\begin{corollary}
Let $\widehat{g}_{x}(\thv)$ be the estimator defined in Eq.~\eqref{eq:gradient_estimator}. $\widehat{g}_{x}(\thv)$ is an unbiased estimator for the $\gindex$-th component of the gradient.
\end{corollary}
\begin{proof}
The proof follows by taking the expectation values of Eq.~\eqref{eq:gradient_estimator} and employing the result from Prop.~\ref{prop-unbiased-mother-cost} and the parameter shift rule.
\end{proof}

We also obtained the variance of the estimator constructed above whose derivation can be found in Appendix~\ref{app:varlinear}. Such a quantity can be used to compare different strategies for allocating shots among Hamiltonian terms. Although we do not compare variances, we use elements of the proof in Appendix~\ref{app:proof-prop-unbiased-simple-mother-cost} for the following construction of an unbiased estimator of the MSE in Appendix~\ref{app:mse}.

Having now constructed an estimator for the loss function as outlined in the previous section combining two single or multi-shot estimates of this function evaluated at the required parameter values will lead to an unbiased estimator of the gradient. 

\subsection{Loss functions with polynomial dependence on the quantum circuit observables}
\subsubsection{Constructing an unbiased estimator of the gradient}
In the case of non-linear dependence on the observables produced by a quantum circuit, estimating the gradient is not as simple as simply applying to parameter shift rule. However, we can still derive estimators for the gradient when the non-linearity is described by a polynomial function.  We begin with the general expression for the loss functions we consider in this work
\begin{equation}
    \LC(\thv) = \sum_{\vec{i}} p_{\vec{i}} \ell (E_{\vec{i}}(\thv))\,.
\end{equation}
Now constructing a polynomial loss function of degree $\polydegree$ requires that
\begin{equation}
\ell (E_{\vec{i}}(\thv)) = \sum_{z=0}^{\polydegree} a_z [E_{\vec{i}}(\thv)]^z\,,
\end{equation}
which leads to the expression of the loss function, 
\begin{equation}
    \LC(\thv) = \sum_{\vec{i},z} p_{\vec{i},z}\big[\sum_{j} c_{\vec{i},j} \langle h_{\vec{i}, j} (\thv)\rangle\big]^{z}\,,
\end{equation}
where $p_{\vec{i},z} = p_{\vec{i}}a_{z}$. Taking the derivative with respect to $\theta_{\gindex}$ leads to
\begin{equation}
\frac{\partial \LC(\thv)}{\partial \theta_{\gindex}} = \sum_{\vec{i},z} p_{\vec{i},z} z \big(\sum_{j} c_{\vec{i},j} \langle h_{\vec{i}, j} (\thv)\rangle\big)^{z-1} \big(\sum_{j'} c_{\vec{i},j'} \frac{\partial\langle h_{\vec{i}, j'} (\thv)\rangle}{\partial \theta_{\gindex}}\big)\,.
\end{equation}
Using the multinomial theorem and given $J$ Hamiltonian terms, we can expand the second sum in the above expression,
\begin{align}\label{eq:gradient}
    \frac{\partial \LC(\thv)}{\partial \theta_{\gindex}} = &\sum_{\vec{i},z} p_{\vec{i},z} z \sum_{b_1+ b_2+ \cdots +b_J=z-1} \binom{z-1}{b_1, b_2, \cdots b_J} \\ \nonumber
    &\prod_{j} (c_{\vec{i},j} \langle h_{\vec{i}, j} (\thv)\rangle)^{b_j} \big(\sum_{j'} c_{\vec{i},j'} \frac{\partial\langle h_{\vec{i}, j'} (\thv)\rangle}{\partial \theta_{\gindex}}\big),
\end{align}
where $\binom{z-1}{b_1, b_2, \cdots b_J} = \frac{(z-1)!}{b_1! b_2! \cdots b_J!}$ and $b_j$ are non-negative integers. Therefore, we need to construct unbiased estimators of the terms  $\langle h_{\vec{i}, j} (\thv)\rangle^{b_{j}}$ and use the previously established gradient estimators with the parameter shift rule. Then we will need to consider how to distribute shots among this estimator. Rewriting Eq.~\eqref{eq:gradient} leads to
\begin{align}\label{eq:gradient2}
    \frac{\partial \LC(\thv)}{\partial \theta_{\gindex}} =  &\sum_{\vec{i},z}\sum_{b_1+ b_2+ \cdots +b_J=z-1} p_{\vec{i},z} z  \frac{(z-1)!}{b_1! b_2! \cdots b_J!} \\ \nonumber
    &\sum_{j'} \prod_j c_{\vec{i},j}^{b_j} c_{\vec{i},j'}  \bigg[\frac{\partial\langle h_{\vec{i}, j'} (\thv)\rangle}{\partial \theta_{\gindex}} \langle h_{\vec{i}, j} (\thv)\rangle^{b_j}\bigg]  \,.
\end{align}
Therefore we can distribute the shots according to the magnitude of the expansion terms $\prod_j c_{\vec{i},j}^{b_j}c_{\vec{i}, j'} p_{\vec{i},z} z  \frac{(z-1)!}{b_1! b_2! \cdots b_J!}$. Once again we can construct an unbiased estimator with the normalized magnitude of these terms defining the probabilities. We now explore how to construct an unbiased estimator for the term $\langle h_{\vec{i}, j} (\thv)\rangle^{b_{j}}$, which is essential to construct an unbiased estimator for the gradients of these kinds of loss functions.

\subsubsection{Constructing an unbiased estimator of polynomial terms}
In order to construct an unbiased estimator for the gradient we need an unbiased estimator for terms of the form
\begin{equation} \label{eq:grad_product_term}
    \frac{\partial\langle h_{\vec{i}, j} (\thv)\rangle}{\partial \theta_{\gindex}} \prod_j  \langle h_{\vec{i}, j} (\thv)\rangle^{b_j}\,.
\end{equation}

We can use the parameter shift rule for the term on the left-hand side. For the term in the product as each $j$ index corresponds to a different operator in the Hamiltonian, each term will be independent. Therefore, we need to construct estimators for terms of the form  $\langle h_{\vec{i}, j} (\thv)\rangle^{z}$. This leads us to the following proposition.

\begin{proposition}
Let  $\widehat{\xi}_{\vec{i},j}$ be an estimator defined as
\begin{equation}\label{eq:poly_unbiased_estimator}
     \widehat{\xi}_{\vec{i},j} = \frac{1}{\mathbb{E}[\binom{s_{\vec{i},j}}{z}]} \sum h^{*}(r_{\vec{i}, j, k_{\alpha_1}},...,r_{\vec{i}, j, k_{\alpha_z}}),
\end{equation}
where the summation is over all subscripts $1 \leq \alpha_{1} < \alpha_{2} < ... < \alpha_{z} \leq s_{\vec{i},j}$ and $h^{*}(r_{\vec{i}, j, k_{1}},...,r_{\vec{i}, j, k_{z}}) = \prod_{\beta=1}^{z} r_{\vec{i}, j, k_{\alpha_{\beta}}}$. $\widehat{\xi}_{\vec{i},j}$ is an unbiased estimator for the term  $\langle h_{\vec{i}, j} (\thv)\rangle^{z}$ estimated with $s_{\vec{i},j} $ shots where $s_{\vec{i},j} \geq z$.
\end{proposition}
\begin{proof}
Eq.~\eqref{eq:poly_unbiased_estimator} is inspired by the form of a U-statistic for $\langle h_{\vec{i}, j} (\thv)\rangle^{b_{j}}$, see Appendix~\ref{app:U-statistics} for more details. Taking the expectation value and using the U-statistic formalism one arrives at the desired result.
\end{proof}

Bringing this all together we can formulate a proposition regarding an unbiased estimator for gradients of loss functions with polynomial dependence.
\begin{proposition}
\begin{align}
    \hat{g}_{\gindex}(\thv) &= \sum_{\vec{i},z}\sum_{b_1+ b_2+ \cdots +b_J=z-1} \\ \nonumber
    & p_{\vec{i},z} z  \frac{(z-1)!}{b_1! b_2! \cdots b_J!} \\ \nonumber
    &\sum_{j'} \prod_{j=1}^{J} c_{\vec{i},j}^{b_j} c_{\vec{i},j'}  \frac{1}{\mathbb{E}[s_{\vec{i},j, j', \vec{b}}]}\sum_{k=1}^{s_{\vec{i},j, j', \vec{b}}} \big( r^{+}_{\vec{i},j,k} - r^{-}_{\vec{i},j,k} \big)  \\ \nonumber
    &\frac{1}{  \mathbb{E}[ \binom{s_{\vec{i},j, j', \vec{b}}}{ b_{j}}]} \sum \prod_{\beta=1}^{z} r_{\vec{i}, j, k_{\alpha_{\beta}}}\,,
\end{align}
where 
\begin{equation}
    \mathbb{E}[r^{\pm}_{\vec{i},j,k}] = \langle h_{\vec{i},j}(\thv \pm \indicatorvec_{\gindex}\frac{\pi}{2})\rangle
\end{equation}
and $\vec{b} = (b_1, b_2, ..., b_J)$. The final sum is over all subscripts $1 \leq \alpha_{1} < \alpha_{2} < ... < \alpha_{z} \leq s_{\vec{i},j}$ and $h^{*}(r_{\vec{i}, j, k_{1}},...,r_{\vec{i}, j, k_{z}}) = \prod_{\beta=1}^{z} r_{\vec{i}, j, k_{\alpha_{\beta}}}$. $\hat{g}_{\gindex}(\thv)$ is a unbiased estimator for $\frac{\partial \LC(\thv)}{\partial \theta_{\gindex}}$
\end{proposition}
\begin{proof}
This can be seen to be true by taking the expectation values and invoking the propositions previously established. 
\end{proof}

These same techniques can be used to construct unbiased estimators of the loss function, which we expand on in Appendix~\ref{app:unbiasedlosspolyD}. In order to provide a concrete example of using the above propositions, we consider the special case of the MSE loss function.

\subsubsection{Constructing an estimator for the gradient of the MSE loss function}

To clarify the notation used above, let us focus on the special case of the  MSE loss function. We consider a slightly more general form than the MSE cost introduced above in Eq.~\eqref{mse_general}. Consider a set of labels $y_{\vec{i}} \in \mathbb{R}$ for a dataset $\{\rho_{\vec{i}}\}_{\vec{i}=1}^N$, composed of the tensor product of $m$ states. The set of predictions from the quantum machine learning model are denoted $\tilde{y}_i(\thv) = \sum_j c_{\vec{i},j} \langle h_{\vec{i}, j} (\thv)\rangle $. Therefore, we can write the loss as
\begin{equation}\label{loss_mse_eg}
    \LC_{MSE}(\thv)= \sum_{\vec{i}} p_{\vec{i}} \ \big[y_{\vec{i}} - \sum_j c_{\vec{i},j} \langle h_{\vec{i}, j} (\thv)\rangle \big]^2\,.
\end{equation}
Expanding the previous equation leads to
\begin{align}
    \LC_{MSE}(\thv &)= \sum_{\vec{i}} p_{\vec{i}} \ \big[y_{\vec{i}}^{2} - y_{\vec{i}} \sum_j c_{\vec{i},j} \langle h_{\vec{i}, j} (\thv)\rangle \nonumber \\
    &\quad \quad \quad + \big(\sum_j c_{\vec{i},j} \langle h_{\vec{i}, j} (\thv)\rangle  \big)^{2} \big].
\end{align}
Evaluating the partial derivative with respect to $\theta_{\gindex}$ gives, 
\begin{align}\label{eq:gradthree}
    \frac{\partial \LC_{MSE}(\thv)}{\partial \theta_{\gindex}}& = \sum_{\vec{i}, j, j'}  - p_{\vec{i}} c_{\vec{i},j} \frac{\partial\langle h_{\vec{i}, j} (\thv)\rangle}{\partial \theta_{\gindex}} \\ \nonumber 
    &+ 2p_{\vec{i}}c_{\vec{i}, j'}c_{\vec{i},j}\langle h_{\vec{i}, j'} (\thv)\rangle \frac{\partial\langle h_{\vec{i}, j} (\thv)\rangle}{\partial \theta_{\gindex}} \,.
\end{align}
We can distribute the total number of shots $s_{tot}$ among these terms according to the relative magnitudes of the $p_{\vec{i}}c_{\vec{i},j'}$ and $2p_{\vec{i}}c_{\vec{i},j'}c_{\vec{i},j}$ coefficients. This can be achieved by sampling from a multinomial probability distribution where the normalized magnitude of these coefficients defines the probabilities. 

It is important to note that some estimators have a different minimum number of shots than others. For example, the estimator for the last term in the above equation, involving a gradient and a direct expectation value, requires a total of $3$ for different circuit evaluations. The estimator for this term can be written as 
\begin{equation}
        \widehat{\mathcal{D}}_{\vec{i},j,j'} =\frac{1}{2} \widehat{\mathcal{E}}_{\vec{i},j'}\big(\widehat{\mathcal{E}}^{+}_{\vec{i},j} - \widehat{\mathcal{E}}^{-}_{\vec{i},j}\big)\,.
\end{equation}
Therefore, in this case, the minimum number of shots given to any term can be set to $3$ to ensure every estimation of any term in  Eq.~\eqref{loss_mse_eg} will always be unbiased. We expand on this consideration below.

\subsection{Distributing the shots among estimator terms}
As previously mentioned the shots can be distributed according to a multinomial distribution with probabilities given by the magnitude of the constant factors that appear in the expression for the above estimators. We note that the shots assigned to a given term may be zero. In order to ensure the estimate of the above term using a given number of shots we need to ensure the estimate of each term is also unbiased. One needs at least $2$ shots to produce an unbiased estimate of the gradient. For terms of the form $\langle h_{\vec{i}, j} (\thv)\rangle^{b_{j}}$ one needs at least $b_{j}$ shots, corresponding to the degree of U-statistic. Therefore, care needs to be taken when distributing shots to ensure that each term measured has sufficiently many. One way to ensure this is the case is to distribute shots in multiples of the largest number of shots needed to evaluate any one term. Any leftover shots can then be distributed equally across the terms to be measured.

Each term itself may consist of a product of several expectation values, each with a different required minimum number of shots. The shots distribution with each term can be selected to correspond to this required minimum number of shots per term. To make this concrete consider a term of the form in Eq.~\eqref{eq:grad_product_term}. Estimating the gradient will take at least $2$ shots. Estimating the product term will take at least $\sum_{j=1}^{J} b_{j}$ shots. If we are given $s_{\vec{i},j}$ shots to use to estimate this term we can assign $\lfloor 2(s_{\vec{i},j}/(2+\sum_{j=1}^{J} b_{j}) \rfloor$ to the first term and $\lfloor \sum_{j=1}^{J} b_{j}(s_{\vec{i},j}/(2+\sum_{j=1}^{J} b_{j}) \rfloor$ to the product term, distributing any remaining shots equally among both.

\section{The Refoqus optimizer}

Now that we have defined unbiased estimators of the gradient, we have the tools needed to present our new optimizer. We call our optimizer Refoqus, which stands for REsource Frugal Optimizer for QUantum Stochastic gradient descent. This optimizer is tailored to QML tasks. QML tasks have the potential to incur large shot overheads because of large training datasets as well as measurement operators composed of a large number of terms. With Refoqus, we achieve shot frugality first by leveraging the gCANS~\cite{gu2021adaptive} rule for allocating shots at each iteration step and second by allocating these shots to individual terms in the loss function via random sampling over the training dataset and over measurement operators. This random sampling allows us to remove the shot floor imposed by deterministic strategies (while still achieving an unbiased estimator for the gradient), hence unlocking the full shot frugality of our optimizer.

The key insight needed for the construction of the Refoqus protocol is that cost functions that have a linear or polynomial dependence on measurable hermitian matrices, and their gradients can always be written in a form consisting of a summation of different measurable quantities with expansion coefficients. These coefficients can in turn be used to define a multinomial probability distribution to guide the allocation of the shots to each term when evaluating the loss function or its gradient.

We outline the Refoqus optimizer in Algorithm~\ref{alg:refoqus}. Given a number of shots to distribute among Hamiltonian terms, one evaluates the gradient components and the corresponding variances with the \textit{iEvaluate} subroutine (Line $3$). We follow the gCANS procedure to compute the shot budget for each iteration (Line $12$). We refer to \cite{gu2021adaptive} for more details on gCANS. The iterative process stops until the total shot budget has been used.

The hyperparameters that we employ are similar to those of Rosalin \cite{arrasmith2020operator}, and will come with similar recommendations on how to set them. For instance, the Lipshitz constant $L$ bounds the largest possible value of the derivative. Hence, it depends on the loss expression but can be set as $M = \sum_{\vec{i}, j} |q_{\vec{i},j}|$ according to \cite{kubler2020adaptive}. Moreover, a learning rate satisfying $0 < \alpha < 2 / L$ can be used.

\section{Convergence Guarantees}
\label{section:convergence}
The framework for Refoqus leverages the structure and the update rule of the gCANS optimizer presented in~\cite{gu2021adaptive}. Therefore, we can apply the same arguments introduced to show geometric convergence. We repeat the arguments and assumptions needed for this convergence result here for convenience.
\begin{proposition}
Provided the loss function satisfies the assumptions stated below the stochastic gradient descent routine used in Refoqus achieves geometric convergence to the optimal value of the cost function. That is,
\begin{equation}
    \mathbb{E}[\LC(\thv)^{(t)}] - \LC^{*} = \mathcal{O}(\gamma^{t})
\end{equation}
where $t$ labels the iteration, $\LC^{*}$ is the optimal value of the loss function, and $0<\gamma<1$.
\end{proposition}

The update rule used in this work and introduced in~\cite{gu2021adaptive} and the proof presented here can be directly applied. This update rule guarantees fast convergence in expectation to the optimal value of sufficiently smooth loss functions. The underlying assumption of this result is that the loss function is strongly convex and has Lipschitz-continuous gradients. In most realistic QML scenarios the optimization landscape is not convex, however, if the optimizer were to settle into a convex region then fast convergence is expected.

The exact assumptions needed in order to ensure the geometric convergence of Refoqus to the global minima are as follows:
\begin{enumerate}
    \item $\mathbb{E}[g_{\gindex}(\thv)] = \frac{\partial \LC(\thv)}{\partial \theta_{\gindex}}, \forall \gindex \in [\gdimension]$.
    \item Var$[g_{\gindex}(\thv)] = \frac{\text{Var}[X_{\gindex}]}{s_{\gindex}}$, where $X_{\gindex}$ is the sampling-based estimator of $g_{x}(\thv)$.
    \item $\LC(\thv)$ is $\mu$-strongly convex.
    \item $\LC(\thv)$ has $L$-Lipschitz continuous gradient.
    \item $\alpha$ is a constant learning rate satisfying $0 < \alpha < \text{min}\{1/L, 2/\mu\}$. 
    \item An ideal version of the update rule holds, that is:
    \begin{align}
        & s_{\gindex} = \frac{2L\alpha}{2 - L\alpha}\frac{\sigma_{\gindex} ( \sum_{\gindex ' = 1}^{\gdimension}\sigma_{\gindex '})}{||\grad{\LC(\thv)||^{2}}} \quad \forall \gindex \in [\gdimension],  \nonumber \\ 
        & \text{where} \quad \sigma_{\gindex} = \sqrt{\text{Var}[X_{\gindex}]} \,.\nonumber
    \end{align}
\end{enumerate}

In the previous sections, we have shown how to construct unbiased estimators for the gradient, satisfying the first assumption. The second assumption is satisfied as the estimate of the gradient is constructed by calculating the mean of several sampling-based estimates. Assumption 5 is satisfied by construction. Assumption 6 is an idealized version of the update rule used in Refoqus. However, the gradient magnitude and estimator variances cannot be exactly known in general, so these quantities are replaced by exponentially moving averages to predict the values of $\sigma_{\gindex}$ and $||\grad{\LC(\thv)}||^{2}$. Assumption 4 is also satisfied for all the loss functions we consider in this work. Finally, as previously noted assumption $5$ is not expected to hold in QML landscapes, which are non-convex in general. 

\begin{algorithm}[H]
\caption{The optimization loop used in \textit{Refoqus}. The function $iEvaluate(\fattheta, \vec{s}, \{f(p_{\vec{i}},c_{\vec{i},j})\})$ evaluates the
the gradient at $\thv$ returning, a vector of the gradient estimates $\vec{g}$ and their variances $\vec{S}$. The vectors are calculated using $s_{0}s_{\gindex}$ shots to estimate the $\gindex$-th component of the gradient and its variance, where $s_{0}$ is the minimum number of shots required to obtain an unbiased estimator. Note that shots for each component estimation are distributed in multiples of $s_{0}$, according to a multinomial distribution determined by the expansion coefficients that define the gradient estimator. These expansion coefficients are defined by the function $f(p_{\vec{i}},c_{\vec{i},j})$, which returns the probability of measuring the term corresponding to the coefficients $p_{\vec{i}}$, $c_{\vec{i},j}$. For the case of linear loss functions $f(p_{\vec{i}},c_{\vec{i},j}) = \frac{|p_{\vec{i}}c_{\vec{i},j}|}{\sum_{\vec{i},j}p_{\vec{i}}c_{\vec{i},j}}$.}\label{alg:refoqus}
\begin{algorithmic}[1]
\Statex \textbf{Input:} Learning rate $\learningrate$, starting point $\fattheta_0$, min number of shots per estimation $s_{\min}$, number of shots that can be used in total $s_{max}$, Lipschitz constant $L$, running average constant $\mu$, a vector of the least number of shots needed for each gradient estimate $s_{0}$, which is loss function dependent and $f(p_{\vec{i}},c_{\vec{i},j})$, which is also loss function dependent.
\State initialize: $\fattheta \gets \fattheta_0 $, $s_{\tot} \gets 0$, $\vec{g} \gets \ (0,...,0)^T $, $\vec{S} \gets \ (0,...,0)^T $,
$\vec{s} \gets (s_{\min} ,... ,s_{\min})^T$, $\vec{\chi}' \gets (0,...,0)^T$, $\vec{\chi} \gets (0,...,0)^T$, $\vec{\xi} \gets (0,...,0)^T$, $\vec{\xi}' \gets (0,...,0)^T$, $t\gets 0$
\While{$s_{\tot} < s_{max}$}
    \State $\vec{g}, \vec{S} \gets iEvaluate(\fattheta, \vec{s}, \{f(p_{\vec{i}},c_{\vec{i},j})\})$
    \State $s_{\tot} \gets s_{\tot} + \sum_{\gindex} s_{0} s_{\gindex}$
    \State $\vec{\chi'} \gets \mu \vec{\chi} + (1-\mu) \vec{g}$
    \State $\vec{\xi'} \gets \mu \vec{\xi} + (1-\mu) \vec{S}$
    \State $\vec{\xi} \gets \vec{\xi'}/(1-\mu^{t+1})$
    \State $\vec{\chi} \gets \vec{\chi'}/(1-\mu^{t+1})$
    \State $\vec{\theta} \gets \vec{\theta} - \learningrate \vec{g}$
\For{$ \gindex \in [1,..., \gdimension]$}
    \State $s_{\gindex} \gets \left\lceil\frac{2L\learningrate}{2-L\learningrate} \frac{\xi_{\gindex} \sum_{\gindex} \xi_{\gindex}}{||\vec{\chi}||^2}\right\rceil$
\EndFor
\State $t\gets t + 1$
\EndWhile
\end{algorithmic}
\end{algorithm}

\section{Numerical Results}

\subsection{Quantum PCA}

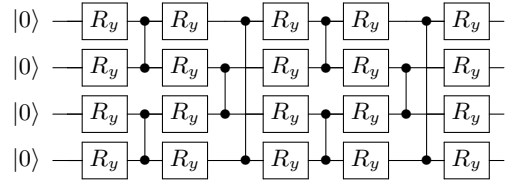
\begin{figure}[t]
\centering
\[
\Qcircuit @C=.6em @R=.4em {
\lstick{\ket{0}} & \qw & \gate{R_y} \qw & \control \qw    & \gate{R_y} \qw &         \qw    & \control \qw    & \gate{R_y} \qw&\control \qw    &\gate{R_y} \qw&         \qw    & \control \qw    & \gate{R_y} \qw&\qw \\
\lstick{\ket{0}}& \qw&\gate{R_y} \qw&\control \qw\qwx&\gate{R_y} \qw&\control \qw    &         \qw\qwx&\gate{R_y}  \qw&\control \qw\qwx&\gate{R_y} \qw&\control \qw    &         \qw\qwx&\gate{R_y} \qw &\qw  \\ 
\lstick{\ket{0}}& \qw&\gate{R_y} \qw&\control \qw    &\gate{R_y} \qw&\control \qw\qwx&         \qw\qwx&\gate{R_y}  \qw&\control \qw    &\gate{R_y} \qw&\control \qw\qwx&         \qw\qwx&\gate{R_y} \qw&\qw \\
\lstick{\ket{0}}& \qw&\gate{R_y} \qw&\control \qw\qwx&\gate{R_y} \qw&         \qw    &\control \qw\qwx&\gate{R_y} \qw&\control \qw\qwx&\gate{R_y} \qw&         \qw    &\control \qw\qwx&\gate{R_y} \qw&\qw\\
}
\]
\caption{\textbf{Hardware-efficient ansatz with $2$ layers used for VQSE on $4$ qubits.} Each $R_y$ rotation is independently parametrized according to $ R_y(\theta) = e^{-iY \theta / 2} $.}
\label{pqcdiagram}
\end{figure}

We benchmark the performance of several optimizers when applied to using VQSE~\cite{cerezo2020variational} for quantum PCA~\cite{lloyd2014quantum} on molecular ground states. We perform quantum PCA on $3$ quantum datasets: ground states of the $H_2$ molecule in the sto-3g basis ($4$ qubits), $H_2$ in the 6-31g basis ($8$ qubits) and $\text{BeH}_2$ in the sto-3g basis ($14$ qubits). We use $101$ circuits, per dataset. The corresponding covariance matrix can be expressed as $ \frac{1}{101} \sum_{i=0}^{100} \ket{\psi}_{i} \bra{\psi}_{i} $. We optimize using the local cost introduced in Eq.~\eqref{eq:VQSE:local} and repeated here for convenience:
\begin{equation} 
    H = \id - \sum_{j=1}^n r_j Z_j,\quad r_j \in \mathbb{R}\,
\end{equation} 
taking coefficients $r_j = 1.0 + 0.2 (j-1)$ and $p_i=\frac{1}{101}$ following the presentation in~\cite{cerezo2020variational}. Hence the latter coefficients form the $q_{i,j} = p_{i}r_{j}$ terms used for shot-allocation strategies, as outlined below.

When implementing Rosalin and Refoqus we use the weighted random sampling strategy for allocating shots as it demonstrated superior performance against other strategies in the numerical results of \cite{arrasmith2020operator}. Hence, the $ q_{i,j}$ coefficients that appear in the loss function (and the gradient estimators) are used to construct a multinomial probability distribution defined by the terms $\frac{|q_{i,j}|}{M}$, where $M = \sum_{i, j} |q_{i,j}|$. This distribution is used to probabilistically allocate the number of shots given to each term for each gradient estimation.

\begin{figure*}[!ht]
\centering
\subfloat[]{\includegraphics[width=0.67\columnwidth]{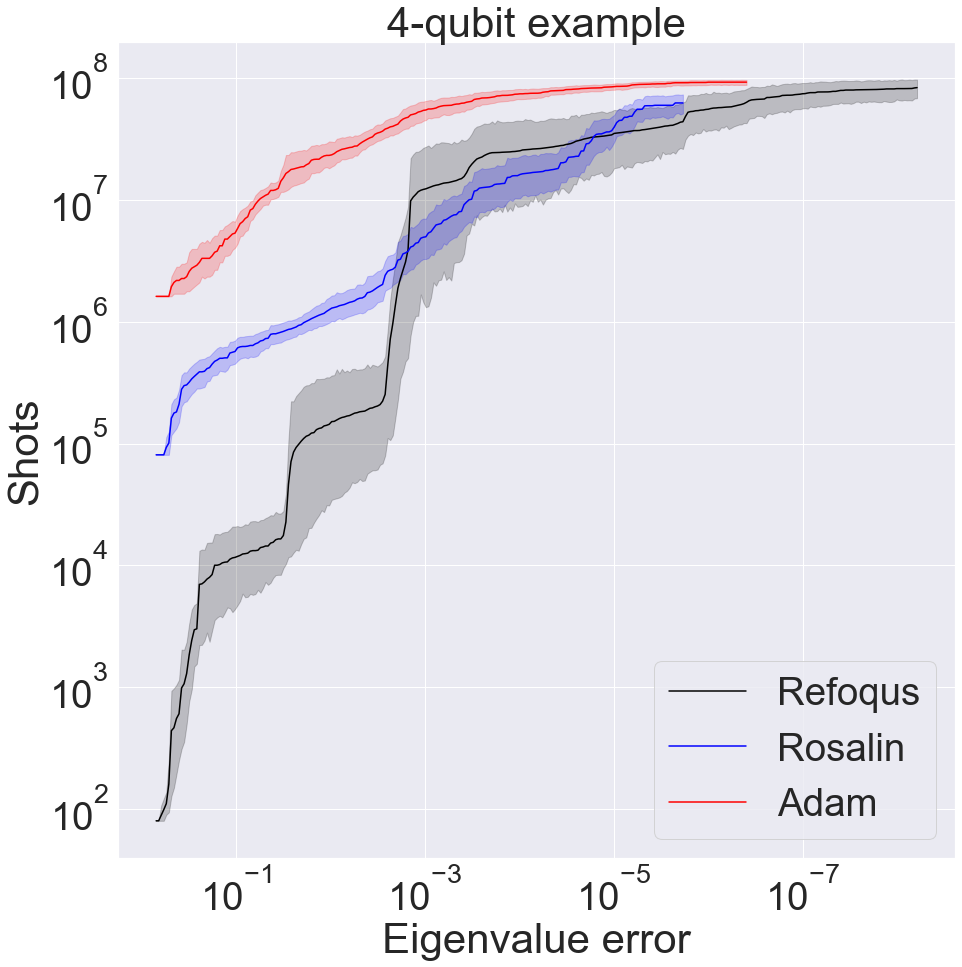}}
\subfloat[]{\includegraphics[width=0.67\columnwidth]{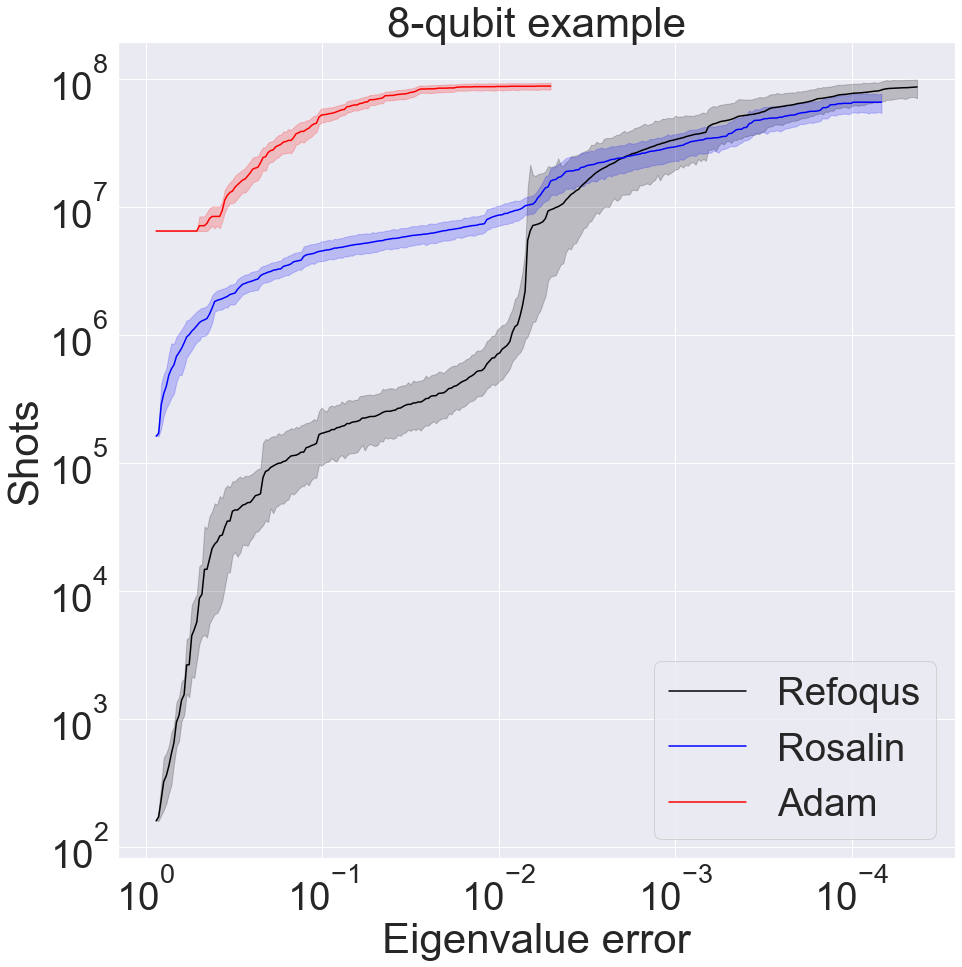}}
\subfloat[]{\includegraphics[width=0.67\columnwidth]{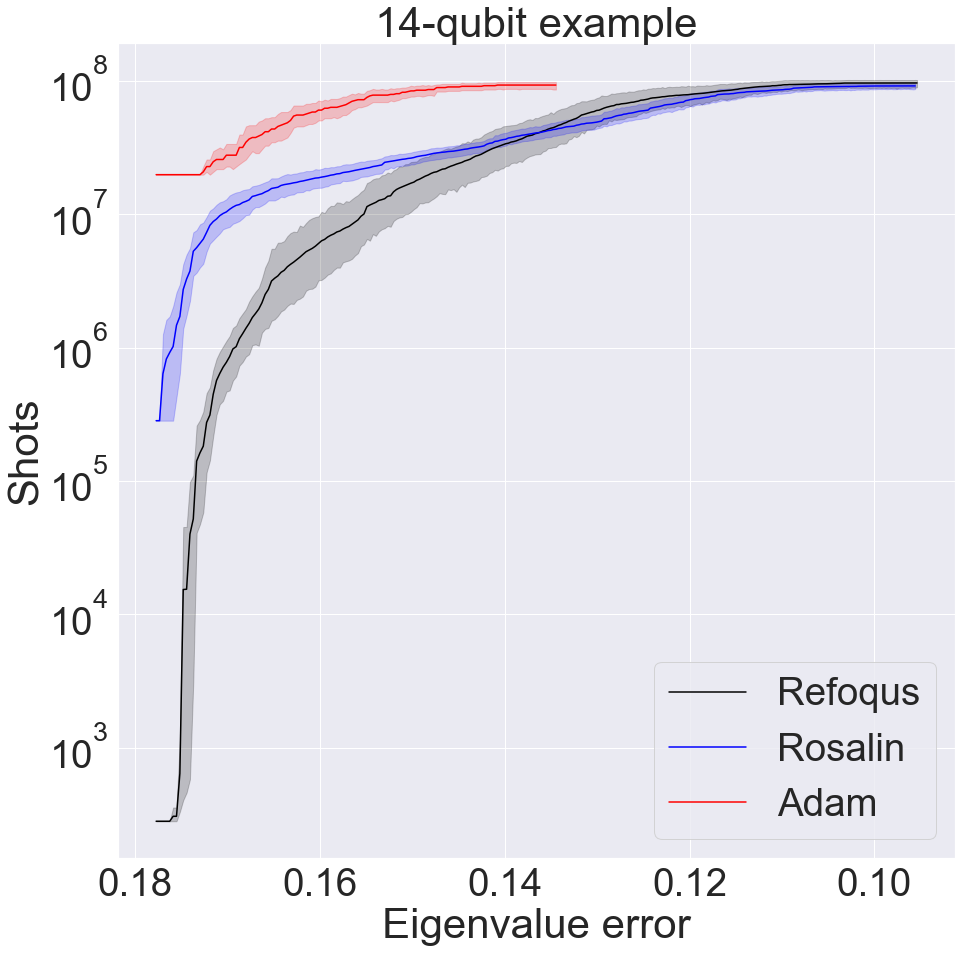}}
\\
\subfloat[]{\includegraphics[width=0.67\columnwidth]{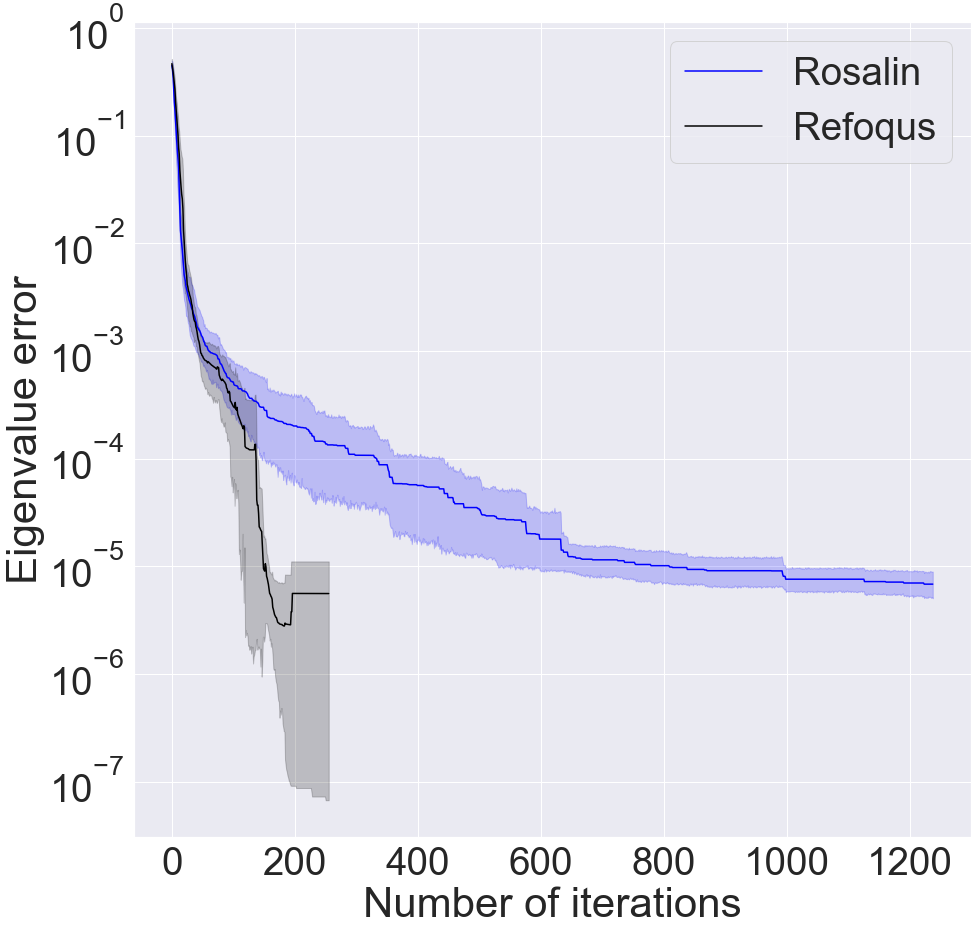}}
\subfloat[]{\includegraphics[width=0.67\columnwidth]{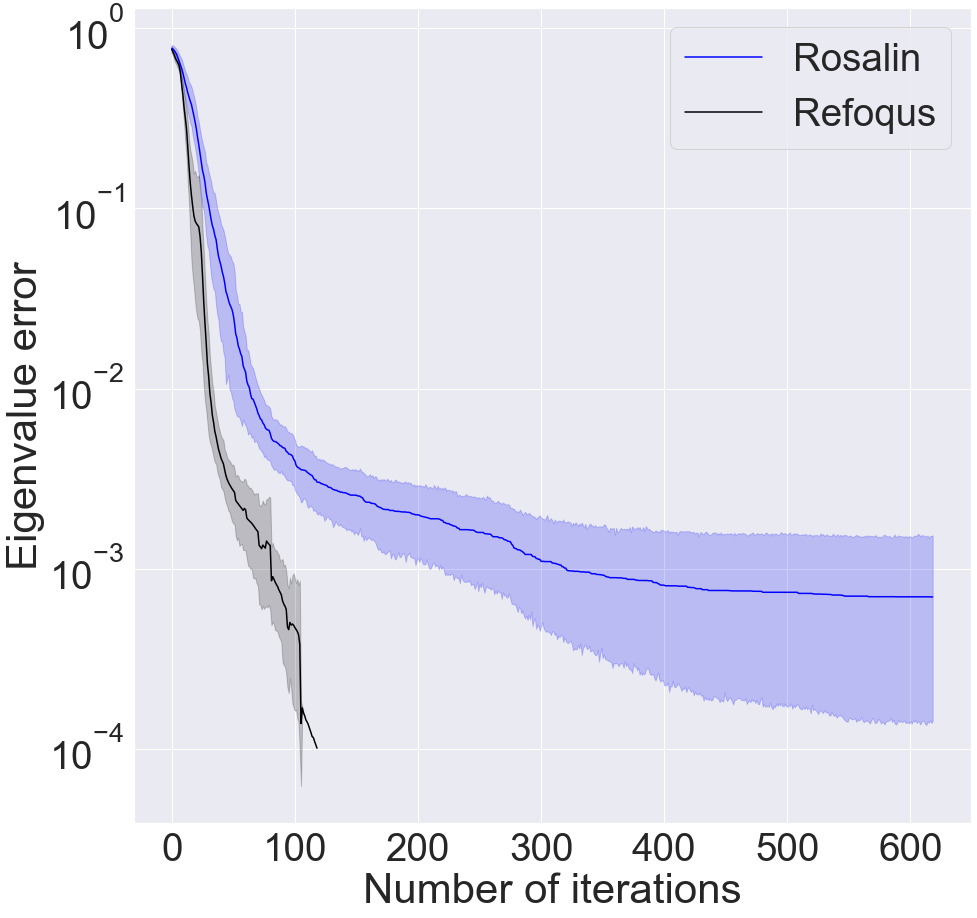}}
\subfloat[]{\includegraphics[width=0.67\columnwidth]{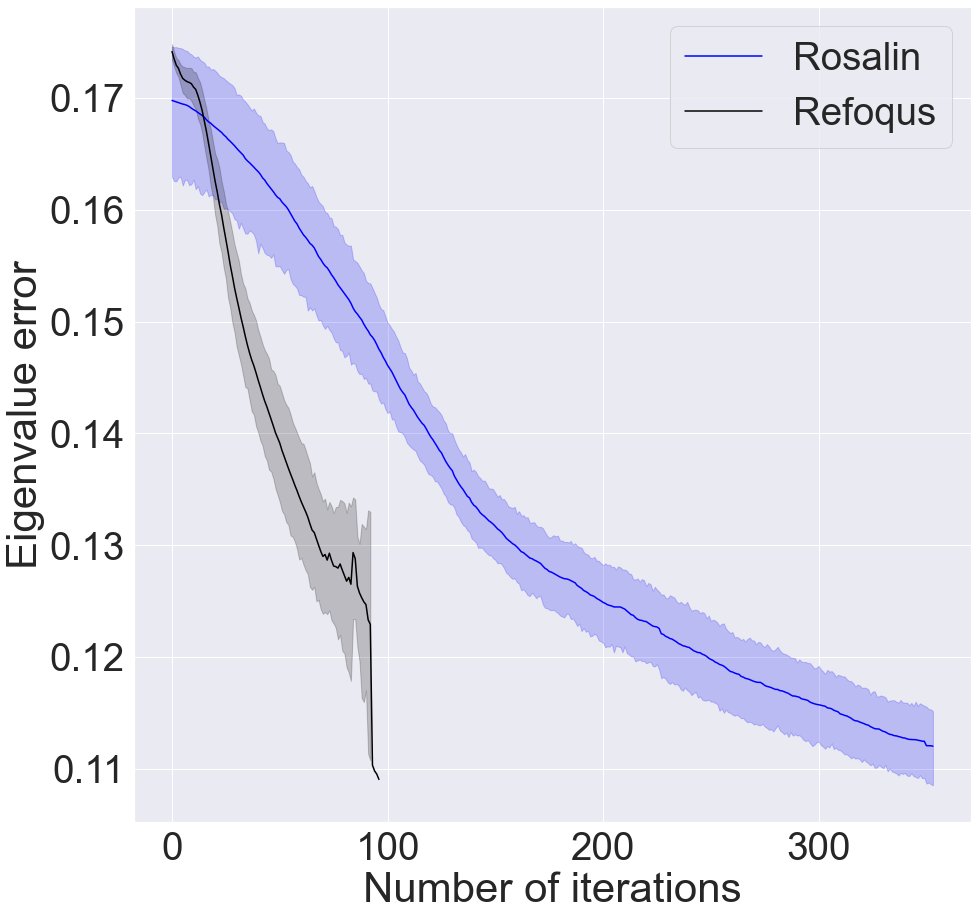}}
\caption{\textbf{VQSE for quantum PCA of $\text{H}_{2}$ molecular ground states (a, b, d, and e) and $\text{BeH}_2$ (c and f)}. The upper plots show the budget of shots spent against the best eigenvalue error achieved by the optimizers. The lower plots show the number of iterations used to spend the total shot budget. We display the results obtained from $20$ independent optimization runs on a data set of $101$ circuits representing molecular ground states calculated using the Adam, Rosalin, and Refoqus optimizers and compare their performance. We show the median (solid lines) and $95\%$ confidence intervals (shaded regions) over the $20$ different random initializations.}
\label{vqseperfs}
\end{figure*}

We use a circuit consisting of $2$ layers of a hardware-efficient ansatz with depth $4$, depicted in Fig.~\ref{pqcdiagram}, with $20$, $40$, and $70$ parameters for the $4$, $8$ and $16$ qubit problems respectively. These parameters are optimized in order to minimize the VQSE cost function using the Adam, Rosalin, and Refoqus optimizers\footnote{For Adam, we perform $100$ shots per circuit in our simulations. For Rosalin and Refoqus, we use a minimum of $2$ shots per circuit (this leads to $s_{\min} = 202$ and $s_{\min} = 2$ for Rosalin and Refoqus respectively) in conjunction with WRS.}. We use the absolute error in estimating the eigenvalues of the system to benchmark the overall performance of the output of the optimized circuit as this is desired output of the algorithm. Given the exact $16$ highest eigenvalues $\lambda_i$, and their estimation $\tilde{\lambda_i}$ given current parameters $\theta$, the latter is computed as $\epsilon_{\lambda} = \sum_{i=1}^{16} (\lambda_i - \tilde{\lambda_i})^2$.

Figure~\ref{vqseperfs} shows the results obtained when running each algorithm $20$ times with different random initialization of the variational ansatz, up to a total shot budget of $10^8$. In general, we see that Refoqus is the best-performing optimizer, achieving a best-case accuracy while using fewer shots for all shot budgets. In summary, for the $4$ qubit system, a median eigenvalue error of $ 6.8 \times 10^{-6}$, $5.15 \times 10^{-6}$ and $ 6.08 \times 10^{-7}$ was obtained using Adam, Rosalin and Refoqus respectively. Although we found better minimal value with Adam compared to Rosalin ($3.79 \times 10^{-7}$ and  $1.74 \times 10^{-6}$ respectively), Rosalin is more advantageous at lower shot budgets.  However, Refoqus achieved a minimal error value of $6.13 \times 10^{-9}$ and demonstrates a clear advantage over both Adam and Rosalin. On both the $8$ and $14$ qubit systems, we observe a  similar trend although eigenvalue errors are worse overall. This is due to the variational ansatz being kept at a fixed depth while increasing the size of the problem, which leads to worse performance. Nonetheless, Refoqus appears to clearly match or outperform both Adam and Rosalin in the number of shots required to reach a given accuracy. 

Additionally, Refoqus requires fewer parameter updates (iterations) when compared to Rosalin. Indeed, for the $4$, $8$ and $14$ qubits problems respectively, a median of $117, 89$, and $80$ iterations were used by Refoqus against $1217$, $618$, and $353$ for Rosalin. We note that this may be interpreted as another desirable feature of the optimizer, as this minimizes the number of iterations needed in the quantum-classical feedback loop to arrive at a solution of the same (or better) quality. Although the number of iterations is not a bottleneck in and of itself, in current hardware the quantum-classical feedback can prove restrictive meaning fewer iterations are favorable for real-device implementation. 

Finally, we note that the variance of the Refoqus results appears to be larger, suggesting that sampling over more terms can lead to a larger variety in the quality of the optimization obtained. Nevertheless, the advantage in shot frugality that arises from sampling over more terms is clear and despite this larger variance, Refoqus is clearly the best-performing optimizer. 

\subsection{Quantum Autoencoder}

In addition to the previous quantum PCA results, we also performed a quantum autoencoder task on the same molecular ground states as before. We optimize using the local cost introduced in Eq.~\eqref{eq:CL-AE} where half of the qubits are ``trashed''. For the $4$ qubit system, the ansatz used is inspired from~\cite{schuld2020circuit}. Its implementation is available in Pennylane under the name ``StronglyEntanglingLayers''~\cite{bergholm2018pennylane}, and we use $3$ layers. For the other systems, we use the same ansatz used for the quantum PCA task with $3$ layers too, as we obtained better cost values. Figure~\ref{qaperfs} shows the results obtained when running each algorithm $10$ times with different random initialization of the variational ansatz, up to a total shot budget of $10^6$ shots. We do not use Adam though as only two iterations would be done with the latter budget. 
In summary, for the $4$ qubit system, a median eigenvalue error of $8.4 \times 10^{-2}$ was obtained using Rosalin using a median of $1468$ iterations, and $2.0 \times 10^{-4}$ in $208$ iterations with Refoqus. The latter achieved a minimal cost value of $4.4 \times 10^{-5}$, and $1.7 \times 10^{-2}$ for Rosalin. On the $8$ qubit system, a median eigenvalue error of $3.3 \times 10^{-1}$ was obtained using Rosalin using a median of $342$ iterations, and $1.3 \times 10^{-3}$ in $65$ iterations with Refoqus. The latter achieved a minimal cost value of $4.3 \times 10^{-4}$, and $1.6 \times 10^{-1}$ for Rosalin. Finally, on the $14$ qubit system, a median eigenvalue error of $4.7 \times 10^{-1}$ was obtained using Rosalin using a median of $89$ iterations, and $1.5 \times 10^{-2}$ in $24$ iterations with Refoqus. The latter achieved a minimal cost value of $8.8 \times 10^{-3}$, and $3.2 \times 10^{-1}$ for Rosalin. Hence, we witness again that Refoqus outperforms Rosalin. 

\begin{figure*}[!ht]
\centering
\subfloat[]{\includegraphics[width=0.67\columnwidth]{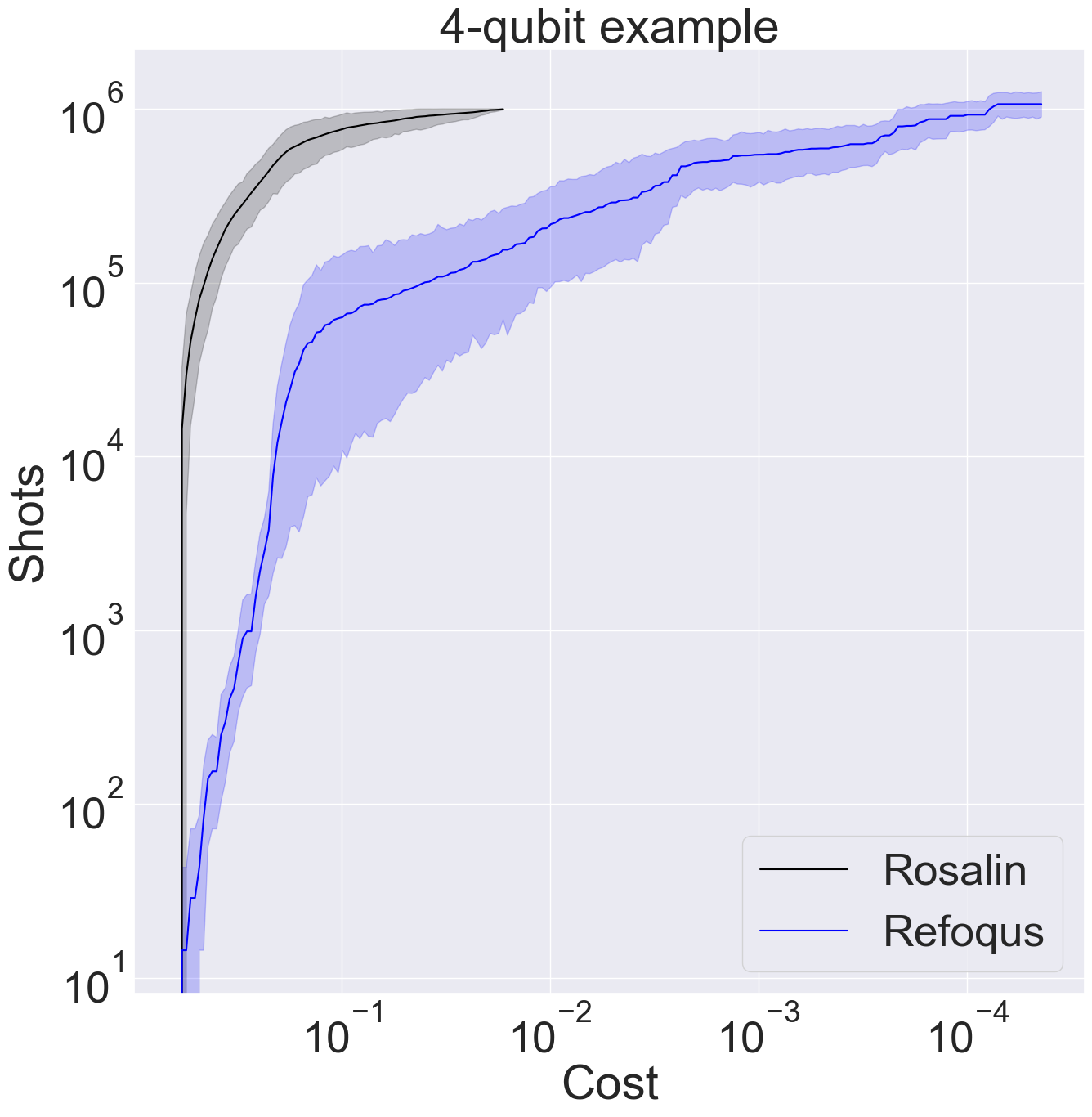}}
\subfloat[]{\includegraphics[width=0.67\columnwidth]{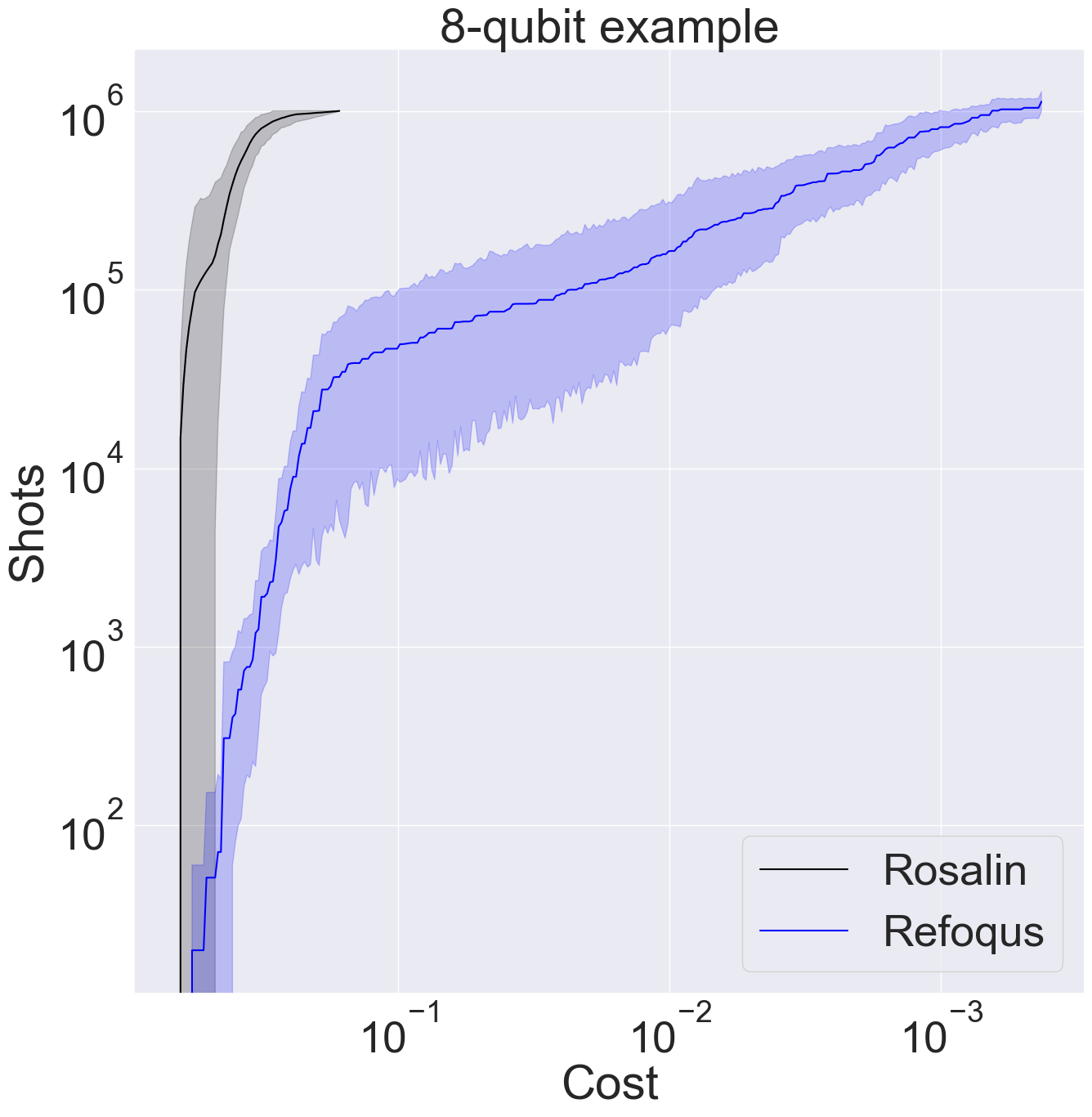}}
\subfloat[]{\includegraphics[width=0.67\columnwidth]{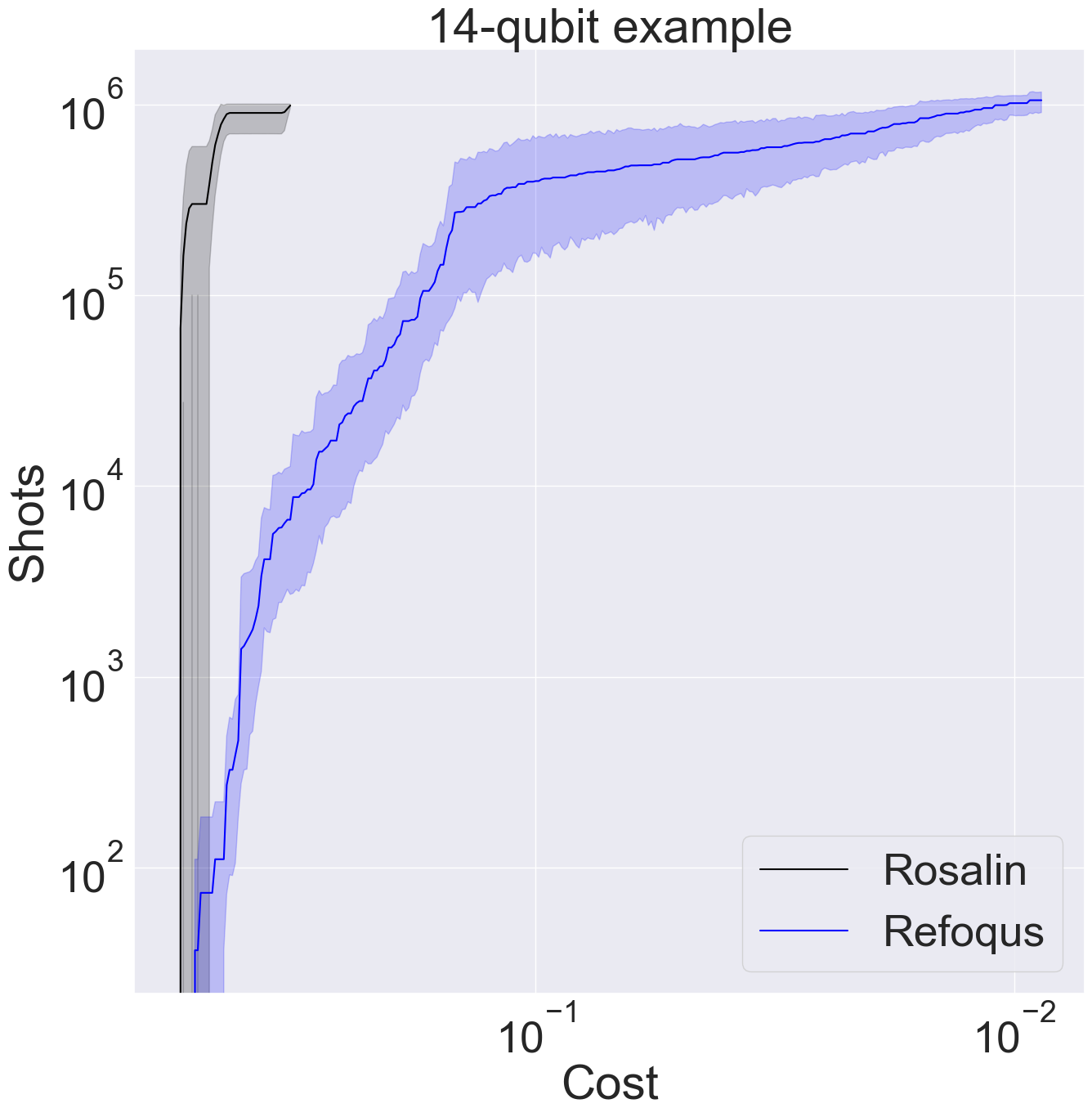}}
\\
\subfloat[]{\includegraphics[width=0.67\columnwidth]{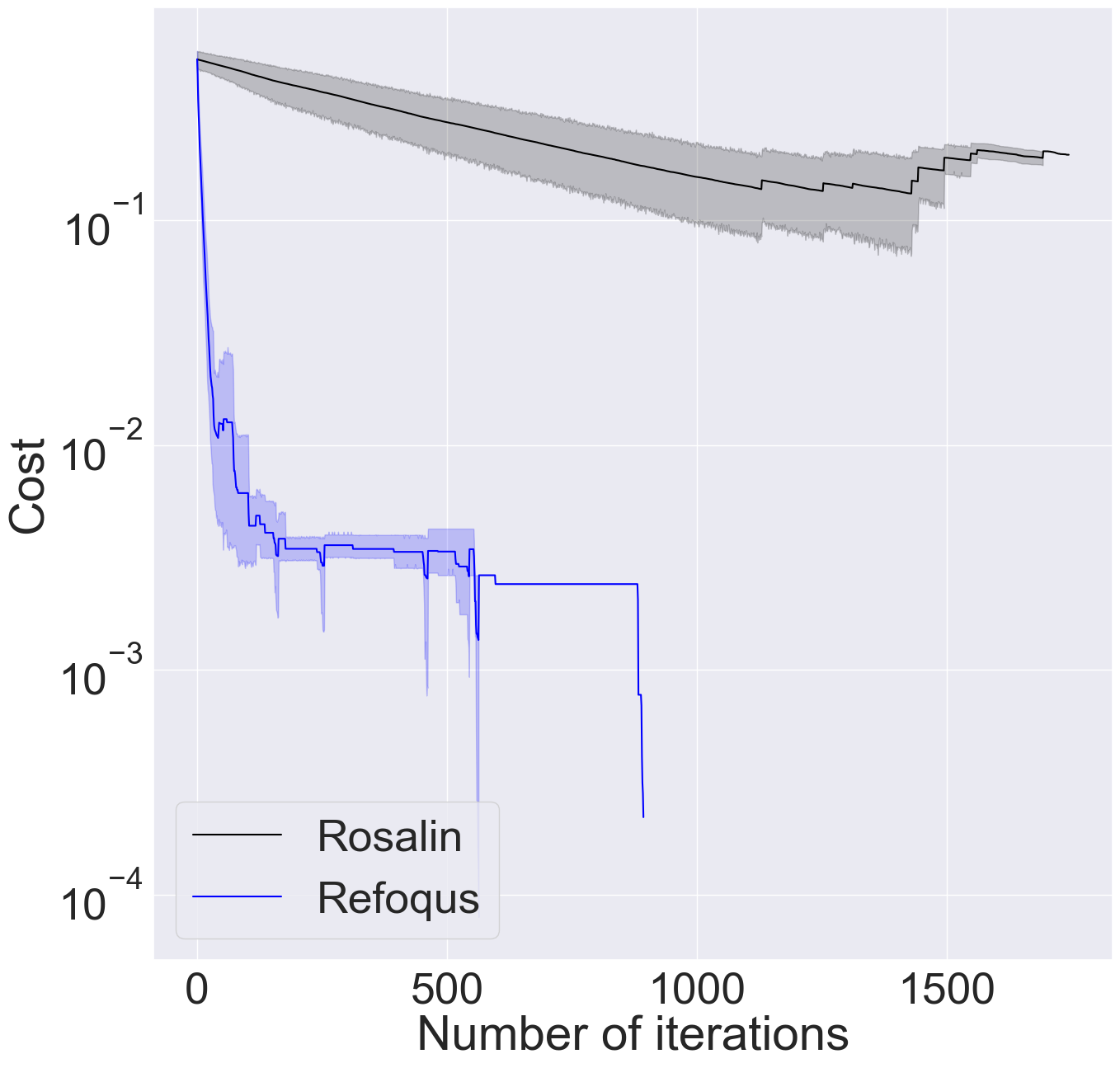}}
\subfloat[]{\includegraphics[width=0.67\columnwidth]{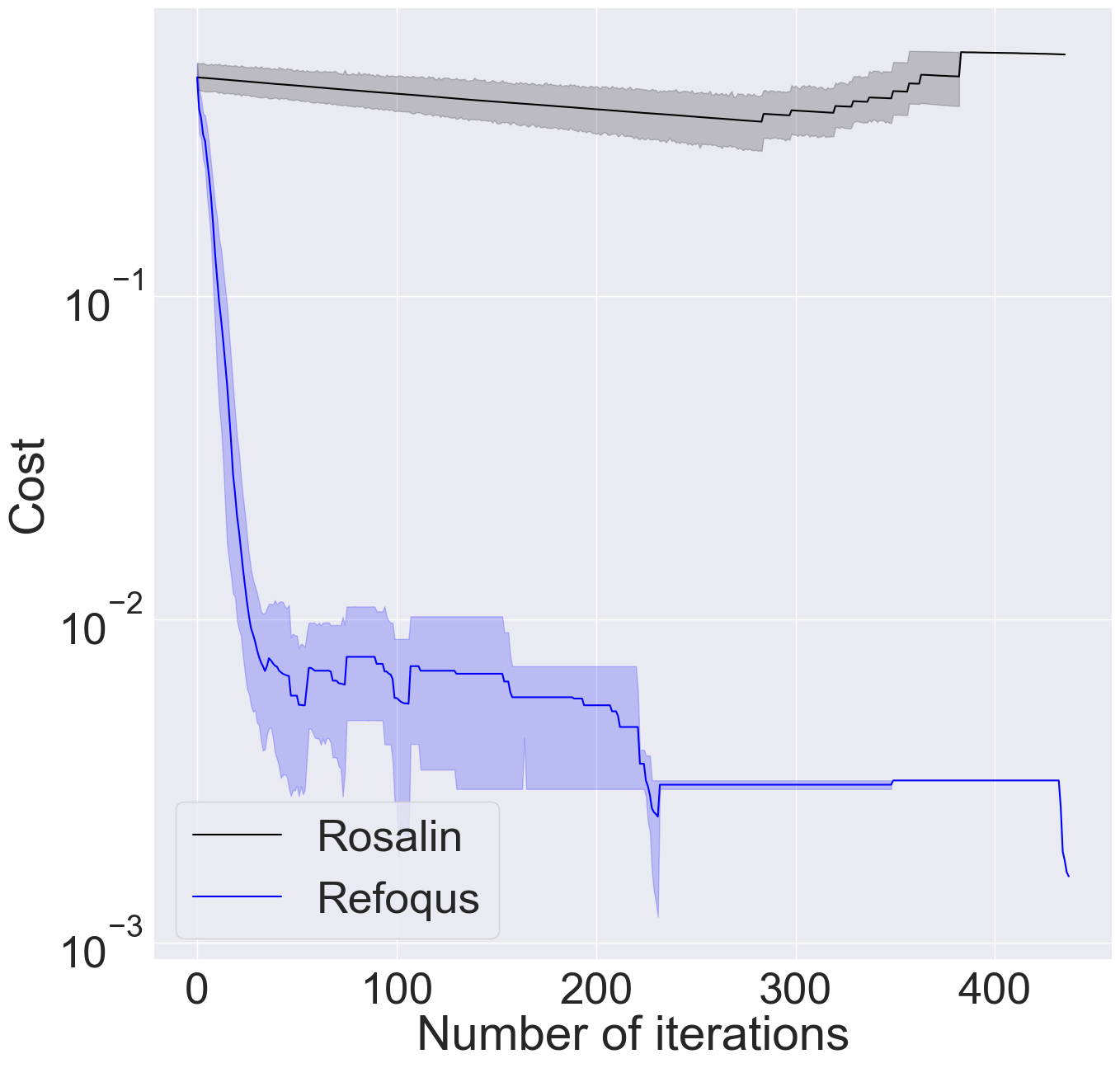}}
\subfloat[]{\includegraphics[width=0.67\columnwidth]{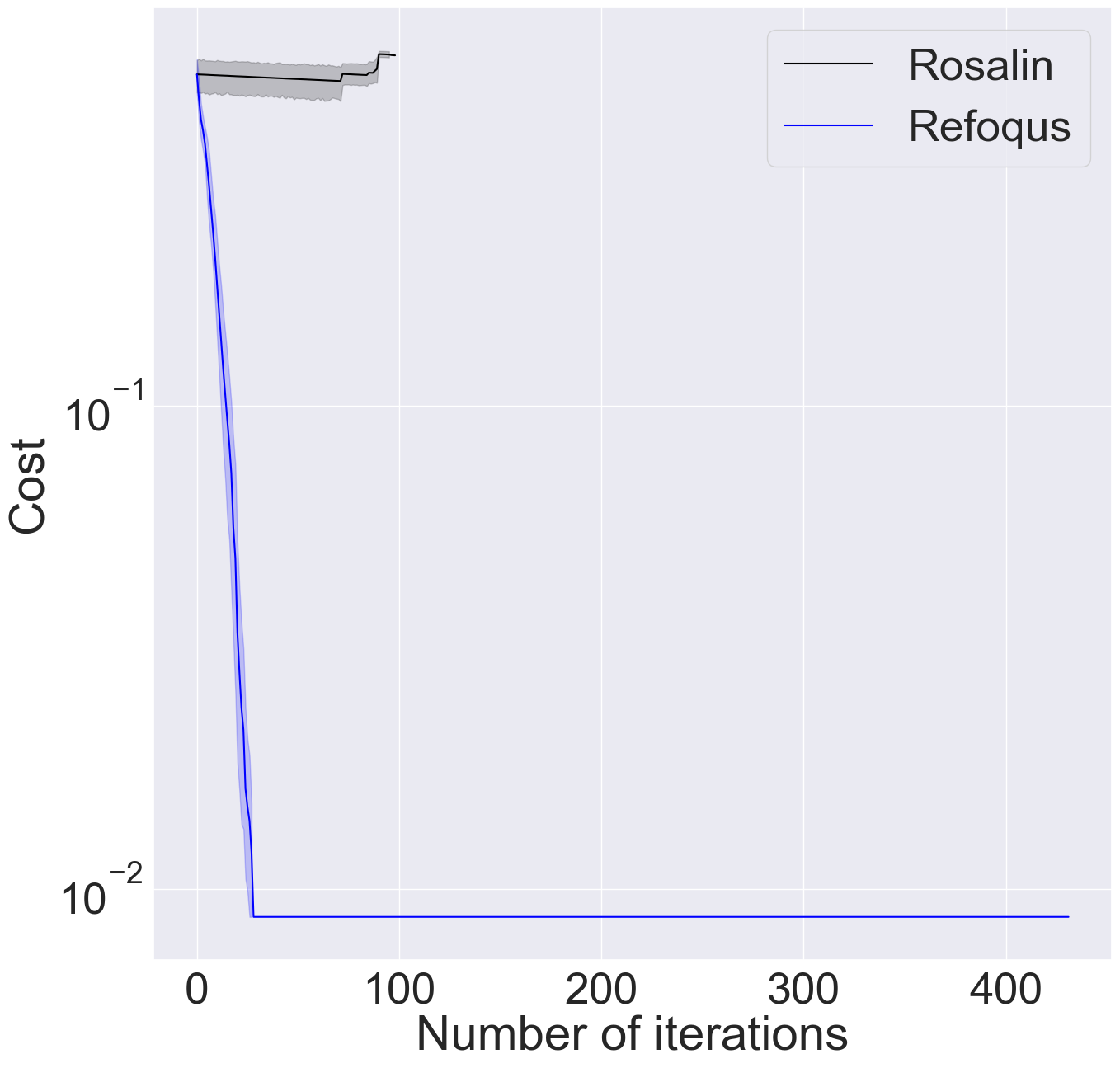}}
\caption{\textbf{Quantum autoencoder applied on $\text{H}_{2}$ molecular ground states (a, b, d, and e) and $\text{BeH}_2$ (c and f)}. The upper plots show the budget of shots spent against the quantum autoencoder cost achieved by the optimizers. The lower plots show the number of iterations used to spend the total shot budget. We display the results obtained from $10$ independent optimization runs on a data set of $101$ circuits representing molecular ground states calculated using the Rosalin, and Refoqus optimizers and compare their performance. We show the median (solid lines) and $95\%$ confidence intervals (shaded regions) over the $10$ different random initializations.}
\label{qaperfs}
\end{figure*}
 
\section{Discussion}

QML algorithms present a different paradigm for data processing which is particularly well suited to quantum data. VQAs are a key contender in giving a near-term useful quantum advantage. However, both VQAs and QML models require long run times and large resource overheads during training (as many iterations and shots to achieve respectable performances). To address these challenges, we propose Refoqus as a shot-frugal gradient-based optimizer based on state and operator sampling. We outline many cost functions that are of interest to the community and are easily captured within our framework for Refoqus. Our new optimizer leverages the loss function of the problem to allocate shots randomly when evaluating gradients. This randomness allows us to make resource-cheap gradient update steps, unlocking shot-frugality. We have shown that Refoqus comes with geometric convergence guarantees under specific assumptions, elaborated in Section~\ref{section:convergence}. Additionally, when applying our optimizer to a QML task, namely a quantum PCA task and a quantum autoencoder task, we obtain significantly better performance in terms of the number of shots and the number of iterations needed to obtain a given accuracy. Although we only applied Refoqus to a quantum PCA task and a quantum autoencoder task, we expect similar results to hold for other problems fitting our framework. Indeed, other cost functions follow a similar pattern and once there are many weighed terms, the sampling approach is expected to return similar results.

A potential future research direction is to extend our analysis to more complicated, non-linear loss functions such as the log-likelihood, exploring in more detail how to introduce shot frugality to gradient-free optimizers~\cite{zhang2022quark, bonet2023performance, kulshrestha2023learning}. Furthermore, applying our optimizer to a problem of interest on a real device is an interesting and logical next step. However, the quality of results can be significantly impacted in noisy settings \cite{wecker2015progress}. To address the latter problem, our optimizer can be combined with several noise-handling procedures when using it on a real device~\cite{endo2018practical, endo2021hybrid, huang2023near, cao2021nisq, bonet2018low, strikis2020learningbased, lowe2020unified, botelho2022error}. It would also be interesting to introduce different techniques to reduce the number of shots spent such as Bayesian optimization \cite{tamiya2022stochastic}.  Finally, other empirical studies related to algorithm selection and configuration~\cite{huang2023near, perezsalinas2023analyzing, bonet2023performance, moussa2020toquantum, moussa2022hyperparameter, moussa2022unsupervised, ito2023latency} can be applied to Refoqus on different datasets \cite{perrier2022qdataset, moussa2022hyperparameter, placidi2023mnisq}.

\acknowledgements 

This work was supported by the U.S. Department of Energy (DOE) through a quantum computing program sponsored by the Los Alamos National Laboratory (LANL) Information Science \& Technology Institute. M.H.G. was partially supported by the Laboratory Directed Research and Development (LDRD) program under project number 20210116DR. M.C. acknowledges support from the LDRD program of LANL under project number 20230049DR. L.C. was partially supported by the U.S. DOE, Office of Science, Office of Advanced Scientific Computing Research, under the Accelerated Research in Quantum Computing (ARQC) program. P.J.C. acknowledges support from the LANL ASC Beyond Moore’s Law project.

\bibliography{quantum}

\begin{thebibliography}{67}%
\makeatletter
\providecommand \@ifxundefined [1]{%
 \@ifx{#1\undefined}
}%
\providecommand \@ifnum [1]{%
 \ifnum #1\expandafter \@firstoftwo
 \else \expandafter \@secondoftwo
 \fi
}%
\providecommand \@ifx [1]{%
 \ifx #1\expandafter \@firstoftwo
 \else \expandafter \@secondoftwo
 \fi
}%
\providecommand \natexlab [1]{#1}%
\providecommand \enquote  [1]{``#1''}%
\providecommand \bibnamefont  [1]{#1}%
\providecommand \bibfnamefont [1]{#1}%
\providecommand \citenamefont [1]{#1}%
\providecommand \href@noop [0]{\@secondoftwo}%
\providecommand \href [0]{\begingroup \@sanitize@url \@href}%
\providecommand \@href[1]{\@@startlink{#1}\@@href}%
\providecommand \@@href[1]{\endgroup#1\@@endlink}%
\providecommand \@sanitize@url [0]{\catcode `\\12\catcode `\$12\catcode
  `\&12\catcode `\#12\catcode `\^12\catcode `\_12\catcode `\%12\relax}%
\providecommand \@@startlink[1]{}%
\providecommand \@@endlink[0]{}%
\providecommand \url  [0]{\begingroup\@sanitize@url \@url }%
\providecommand \@url [1]{\endgroup\@href {#1}{\urlprefix }}%
\providecommand \urlprefix  [0]{URL }%
\providecommand \Eprint [0]{\href }%
\providecommand \doibase [0]{https://doi.org/}%
\providecommand \selectlanguage [0]{\@gobble}%
\providecommand \bibinfo  [0]{\@secondoftwo}%
\providecommand \bibfield  [0]{\@secondoftwo}%
\providecommand \translation [1]{[#1]}%
\providecommand \BibitemOpen [0]{}%
\providecommand \bibitemStop [0]{}%
\providecommand \bibitemNoStop [0]{.\EOS\space}%
\providecommand \EOS [0]{\spacefactor3000\relax}%
\providecommand \BibitemShut  [1]{\csname bibitem#1\endcsname}%
\let\auto@bib@innerbib\@empty
\bibitem [{\citenamefont {Biamonte}\ \emph {et~al.}(2017)\citenamefont
  {Biamonte}, \citenamefont {Wittek}, \citenamefont {Pancotti}, \citenamefont
  {Rebentrost}, \citenamefont {Wiebe},\ and\ \citenamefont
  {Lloyd}}]{biamonte2017quantum}%
  \BibitemOpen
  \bibfield  {author} {\bibinfo {author} {\bibfnamefont {J.}~\bibnamefont
  {Biamonte}}, \bibinfo {author} {\bibfnamefont {P.}~\bibnamefont {Wittek}},
  \bibinfo {author} {\bibfnamefont {N.}~\bibnamefont {Pancotti}}, \bibinfo
  {author} {\bibfnamefont {P.}~\bibnamefont {Rebentrost}}, \bibinfo {author}
  {\bibfnamefont {N.}~\bibnamefont {Wiebe}},\ and\ \bibinfo {author}
  {\bibfnamefont {S.}~\bibnamefont {Lloyd}},\ }\href
  {https://doi.org/10.1038/nature23474} {\bibfield  {journal} {\bibinfo
  {journal} {Nature}\ }\textbf {\bibinfo {volume} {549}},\ \bibinfo {pages}
  {195} (\bibinfo {year} {2017})}\BibitemShut {NoStop}%
\bibitem [{\citenamefont {Schuld}\ \emph {et~al.}(2015)\citenamefont {Schuld},
  \citenamefont {Sinayskiy},\ and\ \citenamefont
  {Petruccione}}]{schuld2015introduction}%
  \BibitemOpen
  \bibfield  {author} {\bibinfo {author} {\bibfnamefont {M.}~\bibnamefont
  {Schuld}}, \bibinfo {author} {\bibfnamefont {I.}~\bibnamefont {Sinayskiy}},\
  and\ \bibinfo {author} {\bibfnamefont {F.}~\bibnamefont {Petruccione}},\
  }\href {https://doi.org/10.1080/00107514.2014.964942} {\bibfield  {journal}
  {\bibinfo  {journal} {Contemporary Physics}\ }\textbf {\bibinfo {volume}
  {56}},\ \bibinfo {pages} {172} (\bibinfo {year} {2015})}\BibitemShut
  {NoStop}%
\bibitem [{\citenamefont {Schuld}\ \emph {et~al.}(2014)\citenamefont {Schuld},
  \citenamefont {Sinayskiy},\ and\ \citenamefont
  {Petruccione}}]{schuld2014quest}%
  \BibitemOpen
  \bibfield  {author} {\bibinfo {author} {\bibfnamefont {M.}~\bibnamefont
  {Schuld}}, \bibinfo {author} {\bibfnamefont {I.}~\bibnamefont {Sinayskiy}},\
  and\ \bibinfo {author} {\bibfnamefont {F.}~\bibnamefont {Petruccione}},\
  }\href {https://doi.org/10.1007/s11128-014-0809-8} {\bibfield  {journal}
  {\bibinfo  {journal} {Quantum Information Processing}\ }\textbf {\bibinfo
  {volume} {13}},\ \bibinfo {pages} {2567} (\bibinfo {year}
  {2014})}\BibitemShut {NoStop}%
\bibitem [{\citenamefont {Cong}\ \emph {et~al.}(2019)\citenamefont {Cong},
  \citenamefont {Choi},\ and\ \citenamefont {Lukin}}]{cong2019quantum}%
  \BibitemOpen
  \bibfield  {author} {\bibinfo {author} {\bibfnamefont {I.}~\bibnamefont
  {Cong}}, \bibinfo {author} {\bibfnamefont {S.}~\bibnamefont {Choi}},\ and\
  \bibinfo {author} {\bibfnamefont {M.~D.}\ \bibnamefont {Lukin}},\ }\href
  {https://doi.org/10.1038/s41567-019-0648-8} {\bibfield  {journal} {\bibinfo
  {journal} {Nature Physics}\ }\textbf {\bibinfo {volume} {15}},\ \bibinfo
  {pages} {1273} (\bibinfo {year} {2019})}\BibitemShut {NoStop}%
\bibitem [{\citenamefont {Abbas}\ \emph {et~al.}(2021)\citenamefont {Abbas},
  \citenamefont {Sutter}, \citenamefont {Zoufal}, \citenamefont {Lucchi},
  \citenamefont {Figalli},\ and\ \citenamefont {Woerner}}]{abbas2020power}%
  \BibitemOpen
  \bibfield  {author} {\bibinfo {author} {\bibfnamefont {A.}~\bibnamefont
  {Abbas}}, \bibinfo {author} {\bibfnamefont {D.}~\bibnamefont {Sutter}},
  \bibinfo {author} {\bibfnamefont {C.}~\bibnamefont {Zoufal}}, \bibinfo
  {author} {\bibfnamefont {A.}~\bibnamefont {Lucchi}}, \bibinfo {author}
  {\bibfnamefont {A.}~\bibnamefont {Figalli}},\ and\ \bibinfo {author}
  {\bibfnamefont {S.}~\bibnamefont {Woerner}},\ }\href
  {https://doi.org/10.1038/s43588-021-00084-1} {\bibfield  {journal} {\bibinfo
  {journal} {Nature Computational Science}\ }\textbf {\bibinfo {volume} {1}},\
  \bibinfo {pages} {403} (\bibinfo {year} {2021})}\BibitemShut {NoStop}%
\bibitem [{\citenamefont {Nguyen}\ \emph {et~al.}(2022)\citenamefont {Nguyen},
  \citenamefont {Schatzki}, \citenamefont {Braccia}, \citenamefont {Ragone},
  \citenamefont {Larocca}, \citenamefont {Sauvage}, \citenamefont {Coles},\
  and\ \citenamefont {Cerezo}}]{nguyen2022atheory}%
  \BibitemOpen
  \bibfield  {author} {\bibinfo {author} {\bibfnamefont {Q.~T.}\ \bibnamefont
  {Nguyen}}, \bibinfo {author} {\bibfnamefont {L.}~\bibnamefont {Schatzki}},
  \bibinfo {author} {\bibfnamefont {P.}~\bibnamefont {Braccia}}, \bibinfo
  {author} {\bibfnamefont {M.}~\bibnamefont {Ragone}}, \bibinfo {author}
  {\bibfnamefont {M.}~\bibnamefont {Larocca}}, \bibinfo {author} {\bibfnamefont
  {F.}~\bibnamefont {Sauvage}}, \bibinfo {author} {\bibfnamefont {P.~J.}\
  \bibnamefont {Coles}},\ and\ \bibinfo {author} {\bibfnamefont
  {M.}~\bibnamefont {Cerezo}},\ }\href {https://arxiv.org/abs/2210.08566}
  {\bibfield  {journal} {\bibinfo  {journal} {arXiv preprint arXiv:2210.08566}\
  } (\bibinfo {year} {2022})}\BibitemShut {NoStop}%
\bibitem [{\citenamefont {McClean}\ \emph {et~al.}(2018)\citenamefont
  {McClean}, \citenamefont {Boixo}, \citenamefont {Smelyanskiy}, \citenamefont
  {Babbush},\ and\ \citenamefont {Neven}}]{mcclean2018barren}%
  \BibitemOpen
  \bibfield  {author} {\bibinfo {author} {\bibfnamefont {J.~R.}\ \bibnamefont
  {McClean}}, \bibinfo {author} {\bibfnamefont {S.}~\bibnamefont {Boixo}},
  \bibinfo {author} {\bibfnamefont {V.~N.}\ \bibnamefont {Smelyanskiy}},
  \bibinfo {author} {\bibfnamefont {R.}~\bibnamefont {Babbush}},\ and\ \bibinfo
  {author} {\bibfnamefont {H.}~\bibnamefont {Neven}},\ }\href
  {https://doi.org/10.1038/s41467-018-07090-4} {\bibfield  {journal} {\bibinfo
  {journal} {Nature {C}ommunications}\ }\textbf {\bibinfo {volume} {9}},\
  \bibinfo {pages} {1} (\bibinfo {year} {2018})}\BibitemShut {NoStop}%
\bibitem [{\citenamefont {Cerezo}\ \emph {et~al.}(2021)\citenamefont {Cerezo},
  \citenamefont {Sone}, \citenamefont {Volkoff}, \citenamefont {Cincio},\ and\
  \citenamefont {Coles}}]{cerezo2020cost}%
  \BibitemOpen
  \bibfield  {author} {\bibinfo {author} {\bibfnamefont {M.}~\bibnamefont
  {Cerezo}}, \bibinfo {author} {\bibfnamefont {A.}~\bibnamefont {Sone}},
  \bibinfo {author} {\bibfnamefont {T.}~\bibnamefont {Volkoff}}, \bibinfo
  {author} {\bibfnamefont {L.}~\bibnamefont {Cincio}},\ and\ \bibinfo {author}
  {\bibfnamefont {P.~J.}\ \bibnamefont {Coles}},\ }\href
  {https://doi.org/10.1038/s41467-021-21728-w} {\bibfield  {journal} {\bibinfo
  {journal} {Nature {C}ommunications}\ }\textbf {\bibinfo {volume} {12}},\
  \bibinfo {pages} {1} (\bibinfo {year} {2021})}\BibitemShut {NoStop}%
\bibitem [{\citenamefont {Holmes}\ \emph {et~al.}(2021)\citenamefont {Holmes},
  \citenamefont {Arrasmith}, \citenamefont {Yan}, \citenamefont {Coles},
  \citenamefont {Albrecht},\ and\ \citenamefont
  {Sornborger}}]{holmes2020barren}%
  \BibitemOpen
  \bibfield  {author} {\bibinfo {author} {\bibfnamefont {Z.}~\bibnamefont
  {Holmes}}, \bibinfo {author} {\bibfnamefont {A.}~\bibnamefont {Arrasmith}},
  \bibinfo {author} {\bibfnamefont {B.}~\bibnamefont {Yan}}, \bibinfo {author}
  {\bibfnamefont {P.~J.}\ \bibnamefont {Coles}}, \bibinfo {author}
  {\bibfnamefont {A.}~\bibnamefont {Albrecht}},\ and\ \bibinfo {author}
  {\bibfnamefont {A.~T.}\ \bibnamefont {Sornborger}},\ }\href
  {https://doi.org/10.1103/PhysRevLett.126.190501} {\bibfield  {journal}
  {\bibinfo  {journal} {Physical Review Letters}\ }\textbf {\bibinfo {volume}
  {126}},\ \bibinfo {pages} {190501} (\bibinfo {year} {2021})}\BibitemShut
  {NoStop}%
\bibitem [{\citenamefont {Holmes}\ \emph {et~al.}(2022)\citenamefont {Holmes},
  \citenamefont {Sharma}, \citenamefont {Cerezo},\ and\ \citenamefont
  {Coles}}]{holmes2021connecting}%
  \BibitemOpen
  \bibfield  {author} {\bibinfo {author} {\bibfnamefont {Z.}~\bibnamefont
  {Holmes}}, \bibinfo {author} {\bibfnamefont {K.}~\bibnamefont {Sharma}},
  \bibinfo {author} {\bibfnamefont {M.}~\bibnamefont {Cerezo}},\ and\ \bibinfo
  {author} {\bibfnamefont {P.~J.}\ \bibnamefont {Coles}},\ }\href
  {https://doi.org/10.1103/PRXQuantum.3.010313} {\bibfield  {journal} {\bibinfo
   {journal} {PRX Quantum}\ }\textbf {\bibinfo {volume} {3}},\ \bibinfo {pages}
  {010313} (\bibinfo {year} {2022})}\BibitemShut {NoStop}%
\bibitem [{\citenamefont {Sharma}\ \emph {et~al.}(2022)\citenamefont {Sharma},
  \citenamefont {Cerezo}, \citenamefont {Cincio},\ and\ \citenamefont
  {Coles}}]{sharma2020trainability}%
  \BibitemOpen
  \bibfield  {author} {\bibinfo {author} {\bibfnamefont {K.}~\bibnamefont
  {Sharma}}, \bibinfo {author} {\bibfnamefont {M.}~\bibnamefont {Cerezo}},
  \bibinfo {author} {\bibfnamefont {L.}~\bibnamefont {Cincio}},\ and\ \bibinfo
  {author} {\bibfnamefont {P.~J.}\ \bibnamefont {Coles}},\ }\href
  {https://doi.org/10.1103/PhysRevLett.128.180505} {\bibfield  {journal}
  {\bibinfo  {journal} {Physical Review Letters}\ }\textbf {\bibinfo {volume}
  {128}},\ \bibinfo {pages} {180505} (\bibinfo {year} {2022})}\BibitemShut
  {NoStop}%
\bibitem [{\citenamefont {Marrero}\ \emph {et~al.}(2021)\citenamefont
  {Marrero}, \citenamefont {Kieferov{\'a}},\ and\ \citenamefont
  {Wiebe}}]{marrero2020entanglement}%
  \BibitemOpen
  \bibfield  {author} {\bibinfo {author} {\bibfnamefont {C.~O.}\ \bibnamefont
  {Marrero}}, \bibinfo {author} {\bibfnamefont {M.}~\bibnamefont
  {Kieferov{\'a}}},\ and\ \bibinfo {author} {\bibfnamefont {N.}~\bibnamefont
  {Wiebe}},\ }\href {https://doi.org/10.1103/PRXQuantum.2.040316} {\bibfield
  {journal} {\bibinfo  {journal} {PRX Quantum}\ }\textbf {\bibinfo {volume}
  {2}},\ \bibinfo {pages} {040316} (\bibinfo {year} {2021})}\BibitemShut
  {NoStop}%
\bibitem [{\citenamefont {Uvarov}\ and\ \citenamefont
  {Biamonte}(2021)}]{uvarov2020barren}%
  \BibitemOpen
  \bibfield  {author} {\bibinfo {author} {\bibfnamefont {A.}~\bibnamefont
  {Uvarov}}\ and\ \bibinfo {author} {\bibfnamefont {J.~D.}\ \bibnamefont
  {Biamonte}},\ }\href {https://doi.org/10.1088/1751-8121/abfac7} {\bibfield
  {journal} {\bibinfo  {journal} {Journal of Physics A: Mathematical and
  Theoretical}\ }\textbf {\bibinfo {volume} {54}},\ \bibinfo {pages} {245301}
  (\bibinfo {year} {2021})}\BibitemShut {NoStop}%
\bibitem [{\citenamefont {Arrasmith}\ \emph {et~al.}(2021)\citenamefont
  {Arrasmith}, \citenamefont {Cerezo}, \citenamefont {Czarnik}, \citenamefont
  {Cincio},\ and\ \citenamefont {Coles}}]{arrasmith2020effect}%
  \BibitemOpen
  \bibfield  {author} {\bibinfo {author} {\bibfnamefont {A.}~\bibnamefont
  {Arrasmith}}, \bibinfo {author} {\bibfnamefont {M.}~\bibnamefont {Cerezo}},
  \bibinfo {author} {\bibfnamefont {P.}~\bibnamefont {Czarnik}}, \bibinfo
  {author} {\bibfnamefont {L.}~\bibnamefont {Cincio}},\ and\ \bibinfo {author}
  {\bibfnamefont {P.~J.}\ \bibnamefont {Coles}},\ }\href
  {https://doi.org/10.22331/q-2021-10-05-558} {\bibfield  {journal} {\bibinfo
  {journal} {Quantum}\ }\textbf {\bibinfo {volume} {5}},\ \bibinfo {pages}
  {558} (\bibinfo {year} {2021})}\BibitemShut {NoStop}%
\bibitem [{\citenamefont {Pesah}\ \emph {et~al.}(2021)\citenamefont {Pesah},
  \citenamefont {Cerezo}, \citenamefont {Wang}, \citenamefont {Volkoff},
  \citenamefont {Sornborger},\ and\ \citenamefont {Coles}}]{pesah2020absence}%
  \BibitemOpen
  \bibfield  {author} {\bibinfo {author} {\bibfnamefont {A.}~\bibnamefont
  {Pesah}}, \bibinfo {author} {\bibfnamefont {M.}~\bibnamefont {Cerezo}},
  \bibinfo {author} {\bibfnamefont {S.}~\bibnamefont {Wang}}, \bibinfo {author}
  {\bibfnamefont {T.}~\bibnamefont {Volkoff}}, \bibinfo {author} {\bibfnamefont
  {A.~T.}\ \bibnamefont {Sornborger}},\ and\ \bibinfo {author} {\bibfnamefont
  {P.~J.}\ \bibnamefont {Coles}},\ }\href
  {https://doi.org/10.1103/PhysRevX.11.041011} {\bibfield  {journal} {\bibinfo
  {journal} {Physical Review X}\ }\textbf {\bibinfo {volume} {11}},\ \bibinfo
  {pages} {041011} (\bibinfo {year} {2021})}\BibitemShut {NoStop}%
\bibitem [{\citenamefont {Bittel}\ and\ \citenamefont
  {Kliesch}(2021)}]{bittel2021training}%
  \BibitemOpen
  \bibfield  {author} {\bibinfo {author} {\bibfnamefont {L.}~\bibnamefont
  {Bittel}}\ and\ \bibinfo {author} {\bibfnamefont {M.}~\bibnamefont
  {Kliesch}},\ }\href {https://doi.org/10.1103/PhysRevLett.127.120502}
  {\bibfield  {journal} {\bibinfo  {journal} {Phys. Rev. Lett.}\ }\textbf
  {\bibinfo {volume} {127}},\ \bibinfo {pages} {120502} (\bibinfo {year}
  {2021})}\BibitemShut {NoStop}%
\bibitem [{\citenamefont {Anschuetz}\ and\ \citenamefont
  {Kiani}(2022)}]{anschuetz2022beyond}%
  \BibitemOpen
  \bibfield  {author} {\bibinfo {author} {\bibfnamefont {E.~R.}\ \bibnamefont
  {Anschuetz}}\ and\ \bibinfo {author} {\bibfnamefont {B.~T.}\ \bibnamefont
  {Kiani}},\ }\href {https://doi.org/10.1038/s41467-022-35364-5} {\bibfield
  {journal} {\bibinfo  {journal} {Nature Communications}\ }\textbf {\bibinfo
  {volume} {13}},\ \bibinfo {pages} {7760} (\bibinfo {year}
  {2022})}\BibitemShut {NoStop}%
\bibitem [{\citenamefont {Wang}\ \emph {et~al.}(2021)\citenamefont {Wang},
  \citenamefont {Fontana}, \citenamefont {Cerezo}, \citenamefont {Sharma},
  \citenamefont {Sone}, \citenamefont {Cincio},\ and\ \citenamefont
  {Coles}}]{wang2020noise}%
  \BibitemOpen
  \bibfield  {author} {\bibinfo {author} {\bibfnamefont {S.}~\bibnamefont
  {Wang}}, \bibinfo {author} {\bibfnamefont {E.}~\bibnamefont {Fontana}},
  \bibinfo {author} {\bibfnamefont {M.}~\bibnamefont {Cerezo}}, \bibinfo
  {author} {\bibfnamefont {K.}~\bibnamefont {Sharma}}, \bibinfo {author}
  {\bibfnamefont {A.}~\bibnamefont {Sone}}, \bibinfo {author} {\bibfnamefont
  {L.}~\bibnamefont {Cincio}},\ and\ \bibinfo {author} {\bibfnamefont {P.~J.}\
  \bibnamefont {Coles}},\ }\href {https://doi.org/10.1038/s41467-021-27045-6}
  {\bibfield  {journal} {\bibinfo  {journal} {Nature {C}ommunications}\
  }\textbf {\bibinfo {volume} {12}},\ \bibinfo {pages} {1} (\bibinfo {year}
  {2021})}\BibitemShut {NoStop}%
\bibitem [{\citenamefont {Stilck~Fran{\c{c}}a}\ and\ \citenamefont
  {Garcia-Patron}(2021)}]{franca2020limitations}%
  \BibitemOpen
  \bibfield  {author} {\bibinfo {author} {\bibfnamefont {D.}~\bibnamefont
  {Stilck~Fran{\c{c}}a}}\ and\ \bibinfo {author} {\bibfnamefont
  {R.}~\bibnamefont {Garcia-Patron}},\ }\href
  {https://doi.org/10.1038/s41567-021-01356-3} {\bibfield  {journal} {\bibinfo
  {journal} {Nature Physics}\ }\textbf {\bibinfo {volume} {17}},\ \bibinfo
  {pages} {1221} (\bibinfo {year} {2021})}\BibitemShut {NoStop}%
\bibitem [{\citenamefont {Wecker}\ \emph {et~al.}(2015)\citenamefont {Wecker},
  \citenamefont {Hastings},\ and\ \citenamefont {Troyer}}]{wecker2015progress}%
  \BibitemOpen
  \bibfield  {author} {\bibinfo {author} {\bibfnamefont {D.}~\bibnamefont
  {Wecker}}, \bibinfo {author} {\bibfnamefont {M.~B.}\ \bibnamefont
  {Hastings}},\ and\ \bibinfo {author} {\bibfnamefont {M.}~\bibnamefont
  {Troyer}},\ }\href {https://doi.org/10.1103/PhysRevA.92.042303} {\bibfield
  {journal} {\bibinfo  {journal} {Physical Review A}\ }\textbf {\bibinfo
  {volume} {92}},\ \bibinfo {pages} {042303} (\bibinfo {year}
  {2015})}\BibitemShut {NoStop}%
\bibitem [{\citenamefont {Stokes}\ \emph {et~al.}(2020)\citenamefont {Stokes},
  \citenamefont {Izaac}, \citenamefont {Killoran},\ and\ \citenamefont
  {Carleo}}]{stokes2020quantum}%
  \BibitemOpen
  \bibfield  {author} {\bibinfo {author} {\bibfnamefont {J.}~\bibnamefont
  {Stokes}}, \bibinfo {author} {\bibfnamefont {J.}~\bibnamefont {Izaac}},
  \bibinfo {author} {\bibfnamefont {N.}~\bibnamefont {Killoran}},\ and\
  \bibinfo {author} {\bibfnamefont {G.}~\bibnamefont {Carleo}},\ }\href
  {https://doi.org/10.22331/q-2020-05-25-269} {\bibfield  {journal} {\bibinfo
  {journal} {Quantum}\ }\textbf {\bibinfo {volume} {4}},\ \bibinfo {pages}
  {269} (\bibinfo {year} {2020})}\BibitemShut {NoStop}%
\bibitem [{\citenamefont {Koczor}\ and\ \citenamefont
  {Benjamin}(2019)}]{koczor2019quantum}%
  \BibitemOpen
  \bibfield  {author} {\bibinfo {author} {\bibfnamefont {B.}~\bibnamefont
  {Koczor}}\ and\ \bibinfo {author} {\bibfnamefont {S.~C.}\ \bibnamefont
  {Benjamin}},\ }\href {https://arxiv.org/abs/1912.08660} {\bibfield  {journal}
  {\bibinfo  {journal} {arXiv preprint arXiv:1912.08660}\ } (\bibinfo {year}
  {2019})}\BibitemShut {NoStop}%
\bibitem [{\citenamefont {Nakanishi}\ \emph {et~al.}(2020)\citenamefont
  {Nakanishi}, \citenamefont {Fujii},\ and\ \citenamefont
  {Todo}}]{nakanishi2020sequential}%
  \BibitemOpen
  \bibfield  {author} {\bibinfo {author} {\bibfnamefont {K.~M.}\ \bibnamefont
  {Nakanishi}}, \bibinfo {author} {\bibfnamefont {K.}~\bibnamefont {Fujii}},\
  and\ \bibinfo {author} {\bibfnamefont {S.}~\bibnamefont {Todo}},\ }\href
  {https://doi.org/10.1103/PhysRevResearch.2.043158} {\bibfield  {journal}
  {\bibinfo  {journal} {Physical Review Research}\ }\textbf {\bibinfo {volume}
  {2}},\ \bibinfo {pages} {043158} (\bibinfo {year} {2020})}\BibitemShut
  {NoStop}%
\bibitem [{\citenamefont {K{\"u}bler}\ \emph {et~al.}(2020)\citenamefont
  {K{\"u}bler}, \citenamefont {Arrasmith}, \citenamefont {Cincio},\ and\
  \citenamefont {Coles}}]{kubler2020adaptive}%
  \BibitemOpen
  \bibfield  {author} {\bibinfo {author} {\bibfnamefont {J.~M.}\ \bibnamefont
  {K{\"u}bler}}, \bibinfo {author} {\bibfnamefont {A.}~\bibnamefont
  {Arrasmith}}, \bibinfo {author} {\bibfnamefont {L.}~\bibnamefont {Cincio}},\
  and\ \bibinfo {author} {\bibfnamefont {P.~J.}\ \bibnamefont {Coles}},\ }\href
  {https://doi.org/10.22331/q-2020-05-11-263} {\bibfield  {journal} {\bibinfo
  {journal} {Quantum}\ }\textbf {\bibinfo {volume} {4}},\ \bibinfo {pages}
  {263} (\bibinfo {year} {2020})}\BibitemShut {NoStop}%
\bibitem [{\citenamefont {Gu}\ \emph {et~al.}(2021)\citenamefont {Gu},
  \citenamefont {Lowe}, \citenamefont {Dub}, \citenamefont {Coles},\ and\
  \citenamefont {Arrasmith}}]{gu2021adaptive}%
  \BibitemOpen
  \bibfield  {author} {\bibinfo {author} {\bibfnamefont {A.}~\bibnamefont
  {Gu}}, \bibinfo {author} {\bibfnamefont {A.}~\bibnamefont {Lowe}}, \bibinfo
  {author} {\bibfnamefont {P.~A.}\ \bibnamefont {Dub}}, \bibinfo {author}
  {\bibfnamefont {P.~J.}\ \bibnamefont {Coles}},\ and\ \bibinfo {author}
  {\bibfnamefont {A.}~\bibnamefont {Arrasmith}},\ }\href
  {https://arxiv.org/abs/2108.10434} {\bibfield  {journal} {\bibinfo  {journal}
  {arXiv preprint arXiv:2108.10434}\ } (\bibinfo {year} {2021})}\BibitemShut
  {NoStop}%
\bibitem [{\citenamefont {Sweke}\ \emph {et~al.}(2020)\citenamefont {Sweke},
  \citenamefont {Wilde}, \citenamefont {Meyer}, \citenamefont {Schuld},
  \citenamefont {F{\"a}hrmann}, \citenamefont {Meynard-Piganeau},\ and\
  \citenamefont {Eisert}}]{sweke2020stochastic}%
  \BibitemOpen
  \bibfield  {author} {\bibinfo {author} {\bibfnamefont {R.}~\bibnamefont
  {Sweke}}, \bibinfo {author} {\bibfnamefont {F.}~\bibnamefont {Wilde}},
  \bibinfo {author} {\bibfnamefont {J.~J.}\ \bibnamefont {Meyer}}, \bibinfo
  {author} {\bibfnamefont {M.}~\bibnamefont {Schuld}}, \bibinfo {author}
  {\bibfnamefont {P.~K.}\ \bibnamefont {F{\"a}hrmann}}, \bibinfo {author}
  {\bibfnamefont {B.}~\bibnamefont {Meynard-Piganeau}},\ and\ \bibinfo {author}
  {\bibfnamefont {J.}~\bibnamefont {Eisert}},\ }\href
  {https://doi.org/10.22331/q-2020-08-31-314} {\bibfield  {journal} {\bibinfo
  {journal} {Quantum}\ }\textbf {\bibinfo {volume} {4}},\ \bibinfo {pages}
  {314} (\bibinfo {year} {2020})}\BibitemShut {NoStop}%
\bibitem [{\citenamefont {Tamiya}\ and\ \citenamefont
  {Yamasaki}(2022)}]{tamiya2022stochastic}%
  \BibitemOpen
  \bibfield  {author} {\bibinfo {author} {\bibfnamefont {S.}~\bibnamefont
  {Tamiya}}\ and\ \bibinfo {author} {\bibfnamefont {H.}~\bibnamefont
  {Yamasaki}},\ }\bibfield  {journal} {\bibinfo  {journal} {npj Quantum
  Information}\ }\textbf {\bibinfo {volume} {8}},\ \href
  {https://doi.org/10.1038/s41534-022-00592-6} {10.1038/s41534-022-00592-6}
  (\bibinfo {year} {2022})\BibitemShut {NoStop}%
\bibitem [{\citenamefont {Arrasmith}\ \emph {et~al.}(2020)\citenamefont
  {Arrasmith}, \citenamefont {Cincio}, \citenamefont {Somma},\ and\
  \citenamefont {Coles}}]{arrasmith2020operator}%
  \BibitemOpen
  \bibfield  {author} {\bibinfo {author} {\bibfnamefont {A.}~\bibnamefont
  {Arrasmith}}, \bibinfo {author} {\bibfnamefont {L.}~\bibnamefont {Cincio}},
  \bibinfo {author} {\bibfnamefont {R.~D.}\ \bibnamefont {Somma}},\ and\
  \bibinfo {author} {\bibfnamefont {P.~J.}\ \bibnamefont {Coles}},\ }\href
  {https://arxiv.org/abs/2004.06252} {\bibfield  {journal} {\bibinfo  {journal}
  {arXiv preprint arXiv:2004.06252}\ } (\bibinfo {year} {2020})}\BibitemShut
  {NoStop}%
\bibitem [{\citenamefont {Beer}\ \emph {et~al.}(2020)\citenamefont {Beer},
  \citenamefont {Bondarenko}, \citenamefont {Farrelly}, \citenamefont
  {Osborne}, \citenamefont {Salzmann}, \citenamefont {Scheiermann},\ and\
  \citenamefont {Wolf}}]{beer2020training}%
  \BibitemOpen
  \bibfield  {author} {\bibinfo {author} {\bibfnamefont {K.}~\bibnamefont
  {Beer}}, \bibinfo {author} {\bibfnamefont {D.}~\bibnamefont {Bondarenko}},
  \bibinfo {author} {\bibfnamefont {T.}~\bibnamefont {Farrelly}}, \bibinfo
  {author} {\bibfnamefont {T.~J.}\ \bibnamefont {Osborne}}, \bibinfo {author}
  {\bibfnamefont {R.}~\bibnamefont {Salzmann}}, \bibinfo {author}
  {\bibfnamefont {D.}~\bibnamefont {Scheiermann}},\ and\ \bibinfo {author}
  {\bibfnamefont {R.}~\bibnamefont {Wolf}},\ }\href
  {https://doi.org/10.1038/s41467-020-14454-2} {\bibfield  {journal} {\bibinfo
  {journal} {Nature {C}ommunications}\ }\textbf {\bibinfo {volume} {11}},\
  \bibinfo {pages} {808} (\bibinfo {year} {2020})}\BibitemShut {NoStop}%
\bibitem [{\citenamefont {Romero}\ \emph {et~al.}(2017)\citenamefont {Romero},
  \citenamefont {Olson},\ and\ \citenamefont
  {Aspuru-Guzik}}]{romero2017quantum}%
  \BibitemOpen
  \bibfield  {author} {\bibinfo {author} {\bibfnamefont {J.}~\bibnamefont
  {Romero}}, \bibinfo {author} {\bibfnamefont {J.~P.}\ \bibnamefont {Olson}},\
  and\ \bibinfo {author} {\bibfnamefont {A.}~\bibnamefont {Aspuru-Guzik}},\
  }\href {https://doi.org/10.1088/2058-9565/aa8072} {\bibfield  {journal}
  {\bibinfo  {journal} {Quantum Science and Technology}\ }\textbf {\bibinfo
  {volume} {2}},\ \bibinfo {pages} {045001} (\bibinfo {year}
  {2017})}\BibitemShut {NoStop}%
\bibitem [{\citenamefont {LaRose}\ \emph {et~al.}(2019)\citenamefont {LaRose},
  \citenamefont {Tikku}, \citenamefont {O'Neel-Judy}, \citenamefont {Cincio},\
  and\ \citenamefont {Coles}}]{larose2019variational}%
  \BibitemOpen
  \bibfield  {author} {\bibinfo {author} {\bibfnamefont {R.}~\bibnamefont
  {LaRose}}, \bibinfo {author} {\bibfnamefont {A.}~\bibnamefont {Tikku}},
  \bibinfo {author} {\bibfnamefont {{\'E}.}~\bibnamefont {O'Neel-Judy}},
  \bibinfo {author} {\bibfnamefont {L.}~\bibnamefont {Cincio}},\ and\ \bibinfo
  {author} {\bibfnamefont {P.~J.}\ \bibnamefont {Coles}},\ }\href
  {https://doi.org/10.1038/s41534-019-0167-6} {\bibfield  {journal} {\bibinfo
  {journal} {npj Quantum Information}\ }\textbf {\bibinfo {volume} {5}},\
  \bibinfo {pages} {1} (\bibinfo {year} {2019})}\BibitemShut {NoStop}%
\bibitem [{\citenamefont {Cerezo}\ \emph {et~al.}(2022)\citenamefont {Cerezo},
  \citenamefont {Sharma}, \citenamefont {Arrasmith},\ and\ \citenamefont
  {Coles}}]{cerezo2020variational}%
  \BibitemOpen
  \bibfield  {author} {\bibinfo {author} {\bibfnamefont {M.}~\bibnamefont
  {Cerezo}}, \bibinfo {author} {\bibfnamefont {K.}~\bibnamefont {Sharma}},
  \bibinfo {author} {\bibfnamefont {A.}~\bibnamefont {Arrasmith}},\ and\
  \bibinfo {author} {\bibfnamefont {P.~J.}\ \bibnamefont {Coles}},\ }\href
  {https://doi.org/10.1038/s41534-022-00611-6} {\bibfield  {journal} {\bibinfo
  {journal} {npj Quantum Information}\ }\textbf {\bibinfo {volume} {8}},\
  \bibinfo {pages} {1} (\bibinfo {year} {2022})}\BibitemShut {NoStop}%
\bibitem [{\citenamefont {Schuld}\ \emph {et~al.}(2020)\citenamefont {Schuld},
  \citenamefont {Bocharov}, \citenamefont {Svore},\ and\ \citenamefont
  {Wiebe}}]{schuld2020circuit}%
  \BibitemOpen
  \bibfield  {author} {\bibinfo {author} {\bibfnamefont {M.}~\bibnamefont
  {Schuld}}, \bibinfo {author} {\bibfnamefont {A.}~\bibnamefont {Bocharov}},
  \bibinfo {author} {\bibfnamefont {K.~M.}\ \bibnamefont {Svore}},\ and\
  \bibinfo {author} {\bibfnamefont {N.}~\bibnamefont {Wiebe}},\ }\href
  {https://doi.org/10.1103/PhysRevA.101.032308} {\bibfield  {journal} {\bibinfo
   {journal} {Physical Review A}\ }\textbf {\bibinfo {volume} {101}},\ \bibinfo
  {pages} {032308} (\bibinfo {year} {2020})}\BibitemShut {NoStop}%
\bibitem [{\citenamefont {Mitarai}\ \emph {et~al.}(2018)\citenamefont
  {Mitarai}, \citenamefont {Negoro}, \citenamefont {Kitagawa},\ and\
  \citenamefont {Fujii}}]{mitarai2018quantum}%
  \BibitemOpen
  \bibfield  {author} {\bibinfo {author} {\bibfnamefont {K.}~\bibnamefont
  {Mitarai}}, \bibinfo {author} {\bibfnamefont {M.}~\bibnamefont {Negoro}},
  \bibinfo {author} {\bibfnamefont {M.}~\bibnamefont {Kitagawa}},\ and\
  \bibinfo {author} {\bibfnamefont {K.}~\bibnamefont {Fujii}},\ }\href
  {https://doi.org/10.1103/PhysRevA.98.032309} {\bibfield  {journal} {\bibinfo
  {journal} {Physical Review A}\ }\textbf {\bibinfo {volume} {98}},\ \bibinfo
  {pages} {032309} (\bibinfo {year} {2018})}\BibitemShut {NoStop}%
\bibitem [{\citenamefont {Schuld}\ \emph {et~al.}(2019)\citenamefont {Schuld},
  \citenamefont {Bergholm}, \citenamefont {Gogolin}, \citenamefont {Izaac},\
  and\ \citenamefont {Killoran}}]{schuld2019evaluating}%
  \BibitemOpen
  \bibfield  {author} {\bibinfo {author} {\bibfnamefont {M.}~\bibnamefont
  {Schuld}}, \bibinfo {author} {\bibfnamefont {V.}~\bibnamefont {Bergholm}},
  \bibinfo {author} {\bibfnamefont {C.}~\bibnamefont {Gogolin}}, \bibinfo
  {author} {\bibfnamefont {J.}~\bibnamefont {Izaac}},\ and\ \bibinfo {author}
  {\bibfnamefont {N.}~\bibnamefont {Killoran}},\ }\href
  {https://doi.org/10.1103/PhysRevA.99.032331} {\bibfield  {journal} {\bibinfo
  {journal} {Physical Review A}\ }\textbf {\bibinfo {volume} {99}},\ \bibinfo
  {pages} {032331} (\bibinfo {year} {2019})}\BibitemShut {NoStop}%
\bibitem [{\citenamefont {Balles}\ \emph {et~al.}(2017)\citenamefont {Balles},
  \citenamefont {Romero},\ and\ \citenamefont {Hennig}}]{balles2017coupling}%
  \BibitemOpen
  \bibfield  {author} {\bibinfo {author} {\bibfnamefont {L.}~\bibnamefont
  {Balles}}, \bibinfo {author} {\bibfnamefont {J.}~\bibnamefont {Romero}},\
  and\ \bibinfo {author} {\bibfnamefont {P.}~\bibnamefont {Hennig}},\ }in\
  \href {http://auai.org/uai2017/proceedings/papers/141.pdf} {\emph {\bibinfo
  {booktitle} {Proceedings of the Thirty-Third Conference on Uncertainty in
  Artificial Intelligence (UAI)}}}\ (\bibinfo {year} {2017})\ pp.\ \bibinfo
  {pages} {410--419}\BibitemShut {NoStop}%
\bibitem [{\citenamefont {Liu}\ and\ \citenamefont
  {Wang}(2018)}]{liu2018differentiable}%
  \BibitemOpen
  \bibfield  {author} {\bibinfo {author} {\bibfnamefont {J.-G.}\ \bibnamefont
  {Liu}}\ and\ \bibinfo {author} {\bibfnamefont {L.}~\bibnamefont {Wang}},\
  }\href {https://doi.org/10.1103/PhysRevA.98.062324} {\bibfield  {journal}
  {\bibinfo  {journal} {Phys. Rev. A}\ }\textbf {\bibinfo {volume} {98}},\
  \bibinfo {pages} {062324} (\bibinfo {year} {2018})}\BibitemShut {NoStop}%
\bibitem [{\citenamefont {Johnson}\ \emph {et~al.}(2017)\citenamefont
  {Johnson}, \citenamefont {Romero}, \citenamefont {Olson}, \citenamefont
  {Cao},\ and\ \citenamefont {Aspuru-Guzik}}]{johnson2017qvector}%
  \BibitemOpen
  \bibfield  {author} {\bibinfo {author} {\bibfnamefont {P.~D.}\ \bibnamefont
  {Johnson}}, \bibinfo {author} {\bibfnamefont {J.}~\bibnamefont {Romero}},
  \bibinfo {author} {\bibfnamefont {J.}~\bibnamefont {Olson}}, \bibinfo
  {author} {\bibfnamefont {Y.}~\bibnamefont {Cao}},\ and\ \bibinfo {author}
  {\bibfnamefont {A.}~\bibnamefont {Aspuru-Guzik}},\ }\href
  {https://arxiv.org/abs/1711.02249} {\bibfield  {journal} {\bibinfo  {journal}
  {arXiv preprint arXiv:1711.02249}\ } (\bibinfo {year} {2017})}\BibitemShut
  {NoStop}%
\bibitem [{\citenamefont {Bondarenko}\ and\ \citenamefont
  {Feldmann}(2020)}]{bondarenko2020quantumautoencoders}%
  \BibitemOpen
  \bibfield  {author} {\bibinfo {author} {\bibfnamefont {D.}~\bibnamefont
  {Bondarenko}}\ and\ \bibinfo {author} {\bibfnamefont {P.}~\bibnamefont
  {Feldmann}},\ }\href {https://doi.org/10.1103/PhysRevLett.124.130502}
  {\bibfield  {journal} {\bibinfo  {journal} {Phys. Rev. Lett.}\ }\textbf
  {\bibinfo {volume} {124}},\ \bibinfo {pages} {130502} (\bibinfo {year}
  {2020})}\BibitemShut {NoStop}%
\bibitem [{\citenamefont {Gibbs}\ \emph {et~al.}(2022)\citenamefont {Gibbs},
  \citenamefont {Holmes}, \citenamefont {Caro}, \citenamefont {Ezzell},
  \citenamefont {Huang}, \citenamefont {Cincio}, \citenamefont {Sornborger},\
  and\ \citenamefont {Coles}}]{gibbs2022dynamical}%
  \BibitemOpen
  \bibfield  {author} {\bibinfo {author} {\bibfnamefont {J.}~\bibnamefont
  {Gibbs}}, \bibinfo {author} {\bibfnamefont {Z.}~\bibnamefont {Holmes}},
  \bibinfo {author} {\bibfnamefont {M.~C.}\ \bibnamefont {Caro}}, \bibinfo
  {author} {\bibfnamefont {N.}~\bibnamefont {Ezzell}}, \bibinfo {author}
  {\bibfnamefont {H.-Y.}\ \bibnamefont {Huang}}, \bibinfo {author}
  {\bibfnamefont {L.}~\bibnamefont {Cincio}}, \bibinfo {author} {\bibfnamefont
  {A.~T.}\ \bibnamefont {Sornborger}},\ and\ \bibinfo {author} {\bibfnamefont
  {P.~J.}\ \bibnamefont {Coles}},\ }\href {https://arxiv.org/abs/2204.10269}
  {\bibfield  {journal} {\bibinfo  {journal} {arXiv preprint arXiv:2204.10269}\
  } (\bibinfo {year} {2022})}\BibitemShut {NoStop}%
\bibitem [{\citenamefont {Cirstoiu}\ \emph {et~al.}(2020)\citenamefont
  {Cirstoiu}, \citenamefont {Holmes}, \citenamefont {Iosue}, \citenamefont
  {Cincio}, \citenamefont {Coles},\ and\ \citenamefont
  {Sornborger}}]{cirstoiu2020variational}%
  \BibitemOpen
  \bibfield  {author} {\bibinfo {author} {\bibfnamefont {C.}~\bibnamefont
  {Cirstoiu}}, \bibinfo {author} {\bibfnamefont {Z.}~\bibnamefont {Holmes}},
  \bibinfo {author} {\bibfnamefont {J.}~\bibnamefont {Iosue}}, \bibinfo
  {author} {\bibfnamefont {L.}~\bibnamefont {Cincio}}, \bibinfo {author}
  {\bibfnamefont {P.~J.}\ \bibnamefont {Coles}},\ and\ \bibinfo {author}
  {\bibfnamefont {A.}~\bibnamefont {Sornborger}},\ }\href
  {https://doi.org/10.1038/s41534-020-00302-0} {\bibfield  {journal} {\bibinfo
  {journal} {npj Quantum Information}\ }\textbf {\bibinfo {volume} {6}},\
  \bibinfo {pages} {1} (\bibinfo {year} {2020})}\BibitemShut {NoStop}%
\bibitem [{\citenamefont {Caro}\ \emph {et~al.}(2022)\citenamefont {Caro},
  \citenamefont {Huang}, \citenamefont {Ezzell}, \citenamefont {Gibbs},
  \citenamefont {Sornborger}, \citenamefont {Cincio}, \citenamefont {Coles},\
  and\ \citenamefont {Holmes}}]{caro2022outofdistribution}%
  \BibitemOpen
  \bibfield  {author} {\bibinfo {author} {\bibfnamefont {M.~C.}\ \bibnamefont
  {Caro}}, \bibinfo {author} {\bibfnamefont {H.-Y.}\ \bibnamefont {Huang}},
  \bibinfo {author} {\bibfnamefont {N.}~\bibnamefont {Ezzell}}, \bibinfo
  {author} {\bibfnamefont {J.}~\bibnamefont {Gibbs}}, \bibinfo {author}
  {\bibfnamefont {A.~T.}\ \bibnamefont {Sornborger}}, \bibinfo {author}
  {\bibfnamefont {L.}~\bibnamefont {Cincio}}, \bibinfo {author} {\bibfnamefont
  {P.~J.}\ \bibnamefont {Coles}},\ and\ \bibinfo {author} {\bibfnamefont
  {Z.}~\bibnamefont {Holmes}},\ }\href {https://arxiv.org/abs/2204.10268}
  {\bibfield  {journal} {\bibinfo  {journal} {arXiv preprint arXiv:2204.10268}\
  } (\bibinfo {year} {2022})}\BibitemShut {NoStop}%
\bibitem [{\citenamefont {Gordon}\ \emph {et~al.}(2022)\citenamefont {Gordon},
  \citenamefont {Cerezo}, \citenamefont {Cincio},\ and\ \citenamefont
  {Coles}}]{gordon2022covariance}%
  \BibitemOpen
  \bibfield  {author} {\bibinfo {author} {\bibfnamefont {M.~H.}\ \bibnamefont
  {Gordon}}, \bibinfo {author} {\bibfnamefont {M.}~\bibnamefont {Cerezo}},
  \bibinfo {author} {\bibfnamefont {L.}~\bibnamefont {Cincio}},\ and\ \bibinfo
  {author} {\bibfnamefont {P.~J.}\ \bibnamefont {Coles}},\ }\href
  {https://doi.org/10.1103/PRXQuantum.3.030334} {\bibfield  {journal} {\bibinfo
   {journal} {PRX Quantum}\ }\textbf {\bibinfo {volume} {3}},\ \bibinfo {pages}
  {030334} (\bibinfo {year} {2022})}\BibitemShut {NoStop}%
\bibitem [{\citenamefont {Kundu}\ \emph {et~al.}(2023)\citenamefont {Kundu},
  \citenamefont {Bedełek}, \citenamefont {Ostaszewski}, \citenamefont
  {Danaci}, \citenamefont {Patel}, \citenamefont {Dunjko},\ and\ \citenamefont
  {Miszczak}}]{kundu2023enhancing}%
  \BibitemOpen
  \bibfield  {author} {\bibinfo {author} {\bibfnamefont {A.}~\bibnamefont
  {Kundu}}, \bibinfo {author} {\bibfnamefont {P.}~\bibnamefont {Bedełek}},
  \bibinfo {author} {\bibfnamefont {M.}~\bibnamefont {Ostaszewski}}, \bibinfo
  {author} {\bibfnamefont {O.}~\bibnamefont {Danaci}}, \bibinfo {author}
  {\bibfnamefont {Y.~J.}\ \bibnamefont {Patel}}, \bibinfo {author}
  {\bibfnamefont {V.}~\bibnamefont {Dunjko}},\ and\ \bibinfo {author}
  {\bibfnamefont {J.~A.}\ \bibnamefont {Miszczak}},\ }\href@noop {} {\bibfield
  {journal} {\bibinfo  {journal} {arXiv preprint arXiv:2306.11086}\ } (\bibinfo
  {year} {2023})}\BibitemShut {NoStop}%
\bibitem [{\citenamefont {Thanasilp}\ \emph {et~al.}(2021)\citenamefont
  {Thanasilp}, \citenamefont {Wang}, \citenamefont {Nghiem}, \citenamefont
  {Coles},\ and\ \citenamefont {Cerezo}}]{thanasilp2021subtleties}%
  \BibitemOpen
  \bibfield  {author} {\bibinfo {author} {\bibfnamefont {S.}~\bibnamefont
  {Thanasilp}}, \bibinfo {author} {\bibfnamefont {S.}~\bibnamefont {Wang}},
  \bibinfo {author} {\bibfnamefont {N.~A.}\ \bibnamefont {Nghiem}}, \bibinfo
  {author} {\bibfnamefont {P.~J.}\ \bibnamefont {Coles}},\ and\ \bibinfo
  {author} {\bibfnamefont {M.}~\bibnamefont {Cerezo}},\ }\href
  {https://arxiv.org/abs/2110.14753} {\bibfield  {journal} {\bibinfo  {journal}
  {arXiv preprint arXiv:2110.14753}\ } (\bibinfo {year} {2021})}\BibitemShut
  {NoStop}%
\bibitem [{\citenamefont {van Opheusden}\ \emph {et~al.}(2020)\citenamefont
  {van Opheusden}, \citenamefont {Acerbi},\ and\ \citenamefont
  {Ma}}]{van2020unbiased}%
  \BibitemOpen
  \bibfield  {author} {\bibinfo {author} {\bibfnamefont {B.}~\bibnamefont {van
  Opheusden}}, \bibinfo {author} {\bibfnamefont {L.}~\bibnamefont {Acerbi}},\
  and\ \bibinfo {author} {\bibfnamefont {W.~J.}\ \bibnamefont {Ma}},\ }\href
  {https://doi.org/10.1371/journal.pcbi.1008483} {\bibfield  {journal}
  {\bibinfo  {journal} {PLoS computational biology}\ }\textbf {\bibinfo
  {volume} {16}},\ \bibinfo {pages} {e1008483} (\bibinfo {year}
  {2020})}\BibitemShut {NoStop}%
\bibitem [{\citenamefont {Hoeffding}(1948)}]{hoeffding1948class}%
  \BibitemOpen
  \bibfield  {author} {\bibinfo {author} {\bibfnamefont {W.}~\bibnamefont
  {Hoeffding}},\ }\href
  {https://doi.org/https://doi.org/10.1214/aoms/1177730196} {\bibfield
  {journal} {\bibinfo  {journal} {The Annals of Mathematical Statistics}\
  }\textbf {\bibinfo {volume} {19}},\ \bibinfo {pages} {293} (\bibinfo {year}
  {1948})}\BibitemShut {NoStop}%
\bibitem [{\citenamefont {Lloyd}\ \emph {et~al.}(2014)\citenamefont {Lloyd},
  \citenamefont {Mohseni},\ and\ \citenamefont
  {Rebentrost}}]{lloyd2014quantum}%
  \BibitemOpen
  \bibfield  {author} {\bibinfo {author} {\bibfnamefont {S.}~\bibnamefont
  {Lloyd}}, \bibinfo {author} {\bibfnamefont {M.}~\bibnamefont {Mohseni}},\
  and\ \bibinfo {author} {\bibfnamefont {P.}~\bibnamefont {Rebentrost}},\
  }\href {http://www.nature.com/articles/nphys3029} {\bibfield  {journal}
  {\bibinfo  {journal} {Nature Physics}\ }\textbf {\bibinfo {volume} {10}},\
  \bibinfo {pages} {631} (\bibinfo {year} {2014})}\BibitemShut {NoStop}%
\bibitem [{\citenamefont {Bergholm}\ \emph {et~al.}(2018)\citenamefont
  {Bergholm}, \citenamefont {Izaac}, \citenamefont {Schuld}, \citenamefont
  {Gogolin}, \citenamefont {Alam}, \citenamefont {Ahmed}, \citenamefont
  {Arrazola}, \citenamefont {Blank}, \citenamefont {Delgado}, \citenamefont
  {Jahangiri} \emph {et~al.}}]{bergholm2018pennylane}%
  \BibitemOpen
  \bibfield  {author} {\bibinfo {author} {\bibfnamefont {V.}~\bibnamefont
  {Bergholm}}, \bibinfo {author} {\bibfnamefont {J.}~\bibnamefont {Izaac}},
  \bibinfo {author} {\bibfnamefont {M.}~\bibnamefont {Schuld}}, \bibinfo
  {author} {\bibfnamefont {C.}~\bibnamefont {Gogolin}}, \bibinfo {author}
  {\bibfnamefont {M.~S.}\ \bibnamefont {Alam}}, \bibinfo {author}
  {\bibfnamefont {S.}~\bibnamefont {Ahmed}}, \bibinfo {author} {\bibfnamefont
  {J.~M.}\ \bibnamefont {Arrazola}}, \bibinfo {author} {\bibfnamefont
  {C.}~\bibnamefont {Blank}}, \bibinfo {author} {\bibfnamefont
  {A.}~\bibnamefont {Delgado}}, \bibinfo {author} {\bibfnamefont
  {S.}~\bibnamefont {Jahangiri}}, \emph {et~al.},\ }\href
  {https://arxiv.org/abs/1811.04968} {\bibfield  {journal} {\bibinfo  {journal}
  {arXiv preprint arXiv:1811.04968}\ } (\bibinfo {year} {2018})}\BibitemShut
  {NoStop}%
\bibitem [{\citenamefont {Zhang}\ \emph {et~al.}(2022)\citenamefont {Zhang},
  \citenamefont {Chen}, \citenamefont {Huang},\ and\ \citenamefont
  {Jia}}]{zhang2022quark}%
  \BibitemOpen
  \bibfield  {author} {\bibinfo {author} {\bibfnamefont {Z.}~\bibnamefont
  {Zhang}}, \bibinfo {author} {\bibfnamefont {Z.}~\bibnamefont {Chen}},
  \bibinfo {author} {\bibfnamefont {H.}~\bibnamefont {Huang}},\ and\ \bibinfo
  {author} {\bibfnamefont {Z.}~\bibnamefont {Jia}},\ }\href@noop {} {\bibinfo
  {title} {Quark: A gradient-free quantum learning framework for classification
  tasks}} (\bibinfo {year} {2022}),\ \Eprint
  {https://arxiv.org/abs/arXiv:2210.01311} {arXiv:2210.01311} \BibitemShut
  {NoStop}%
\bibitem [{\citenamefont {Bonet-Monroig}\ \emph {et~al.}(2023)\citenamefont
  {Bonet-Monroig}, \citenamefont {Wang}, \citenamefont {Vermetten},
  \citenamefont {Senjean}, \citenamefont {Moussa}, \citenamefont {B\"ack},
  \citenamefont {Dunjko},\ and\ \citenamefont
  {O'Brien}}]{bonet2023performance}%
  \BibitemOpen
  \bibfield  {author} {\bibinfo {author} {\bibfnamefont {X.}~\bibnamefont
  {Bonet-Monroig}}, \bibinfo {author} {\bibfnamefont {H.}~\bibnamefont {Wang}},
  \bibinfo {author} {\bibfnamefont {D.}~\bibnamefont {Vermetten}}, \bibinfo
  {author} {\bibfnamefont {B.}~\bibnamefont {Senjean}}, \bibinfo {author}
  {\bibfnamefont {C.}~\bibnamefont {Moussa}}, \bibinfo {author} {\bibfnamefont
  {T.}~\bibnamefont {B\"ack}}, \bibinfo {author} {\bibfnamefont
  {V.}~\bibnamefont {Dunjko}},\ and\ \bibinfo {author} {\bibfnamefont {T.~E.}\
  \bibnamefont {O'Brien}},\ }\href
  {https://doi.org/10.1103/PhysRevA.107.032407} {\bibfield  {journal} {\bibinfo
   {journal} {Phys. Rev. A}\ }\textbf {\bibinfo {volume} {107}},\ \bibinfo
  {pages} {032407} (\bibinfo {year} {2023})}\BibitemShut {NoStop}%
\bibitem [{\citenamefont {Kulshrestha}\ \emph {et~al.}(2023)\citenamefont
  {Kulshrestha}, \citenamefont {Liu}, \citenamefont {Ushijima\-Mwesigwa},\ and\
  \citenamefont {Safro}}]{kulshrestha2023learning}%
  \BibitemOpen
  \bibfield  {author} {\bibinfo {author} {\bibfnamefont {A.}~\bibnamefont
  {Kulshrestha}}, \bibinfo {author} {\bibfnamefont {X.}~\bibnamefont {Liu}},
  \bibinfo {author} {\bibfnamefont {H.}~\bibnamefont {Ushijima\-Mwesigwa}},\
  and\ \bibinfo {author} {\bibfnamefont {I.}~\bibnamefont {Safro}},\
  }\href@noop {} {\bibfield  {journal} {\bibinfo  {journal} {arXiv preprint
  arXiv:2304.07442}\ } (\bibinfo {year} {2023})}\BibitemShut {NoStop}%
\bibitem [{\citenamefont {Endo}\ \emph {et~al.}(2018)\citenamefont {Endo},
  \citenamefont {Benjamin},\ and\ \citenamefont {Li}}]{endo2018practical}%
  \BibitemOpen
  \bibfield  {author} {\bibinfo {author} {\bibfnamefont {S.}~\bibnamefont
  {Endo}}, \bibinfo {author} {\bibfnamefont {S.~C.}\ \bibnamefont {Benjamin}},\
  and\ \bibinfo {author} {\bibfnamefont {Y.}~\bibnamefont {Li}},\ }\href
  {https://journals.aps.org/prx/abstract/10.1103/PhysRevX.8.031027} {\bibfield
  {journal} {\bibinfo  {journal} {Physical Review X}\ }\textbf {\bibinfo
  {volume} {8}},\ \bibinfo {pages} {031027} (\bibinfo {year}
  {2018})}\BibitemShut {NoStop}%
\bibitem [{\citenamefont {Endo}\ \emph {et~al.}(2021)\citenamefont {Endo},
  \citenamefont {Cai}, \citenamefont {Benjamin},\ and\ \citenamefont
  {Yuan}}]{endo2021hybrid}%
  \BibitemOpen
  \bibfield  {author} {\bibinfo {author} {\bibfnamefont {S.}~\bibnamefont
  {Endo}}, \bibinfo {author} {\bibfnamefont {Z.}~\bibnamefont {Cai}}, \bibinfo
  {author} {\bibfnamefont {S.~C.}\ \bibnamefont {Benjamin}},\ and\ \bibinfo
  {author} {\bibfnamefont {X.}~\bibnamefont {Yuan}},\ }\href
  {https://doi.org/10.7566/JPSJ.90.032001} {\bibfield  {journal} {\bibinfo
  {journal} {Journal of the Physical Society of Japan}\ }\textbf {\bibinfo
  {volume} {90}},\ \bibinfo {pages} {032001} (\bibinfo {year}
  {2021})}\BibitemShut {NoStop}%
\bibitem [{\citenamefont {Huang}\ \emph {et~al.}(2023)\citenamefont {Huang},
  \citenamefont {Xu}, \citenamefont {Guo}, \citenamefont {Tian}, \citenamefont
  {Wei}, \citenamefont {Sun}, \citenamefont {Bao},\ and\ \citenamefont
  {Long}}]{huang2023near}%
  \BibitemOpen
  \bibfield  {author} {\bibinfo {author} {\bibfnamefont {H.-L.}\ \bibnamefont
  {Huang}}, \bibinfo {author} {\bibfnamefont {X.-Y.}\ \bibnamefont {Xu}},
  \bibinfo {author} {\bibfnamefont {C.}~\bibnamefont {Guo}}, \bibinfo {author}
  {\bibfnamefont {G.}~\bibnamefont {Tian}}, \bibinfo {author} {\bibfnamefont
  {S.-J.}\ \bibnamefont {Wei}}, \bibinfo {author} {\bibfnamefont
  {X.}~\bibnamefont {Sun}}, \bibinfo {author} {\bibfnamefont {W.-S.}\
  \bibnamefont {Bao}},\ and\ \bibinfo {author} {\bibfnamefont {G.-L.}\
  \bibnamefont {Long}},\ }\bibfield  {journal} {\bibinfo  {journal} {Science
  China Physics, Mechanics \& Astronomy}\ }\textbf {\bibinfo {volume} {66}},\
  \href {https://doi.org/10.1007/s11433-022-2057-y} {10.1007/s11433-022-2057-y}
  (\bibinfo {year} {2023})\BibitemShut {NoStop}%
\bibitem [{\citenamefont {Cao}\ \emph {et~al.}(2021)\citenamefont {Cao},
  \citenamefont {Lin}, \citenamefont {Kribs}, \citenamefont {Poon},
  \citenamefont {Zeng},\ and\ \citenamefont {Laflamme}}]{cao2021nisq}%
  \BibitemOpen
  \bibfield  {author} {\bibinfo {author} {\bibfnamefont {N.}~\bibnamefont
  {Cao}}, \bibinfo {author} {\bibfnamefont {J.}~\bibnamefont {Lin}}, \bibinfo
  {author} {\bibfnamefont {D.}~\bibnamefont {Kribs}}, \bibinfo {author}
  {\bibfnamefont {Y.-T.}\ \bibnamefont {Poon}}, \bibinfo {author}
  {\bibfnamefont {B.}~\bibnamefont {Zeng}},\ and\ \bibinfo {author}
  {\bibfnamefont {R.}~\bibnamefont {Laflamme}},\ }\href
  {https://arxiv.org/abs/2111.02345} {\bibfield  {journal} {\bibinfo  {journal}
  {arXiv preprint arXiv:2111.02345}\ } (\bibinfo {year} {2021})}\BibitemShut
  {NoStop}%
\bibitem [{\citenamefont {Bonet-Monroig}\ \emph {et~al.}(2018)\citenamefont
  {Bonet-Monroig}, \citenamefont {Sagastizabal}, \citenamefont {Singh},\ and\
  \citenamefont {O'Brien}}]{bonet2018low}%
  \BibitemOpen
  \bibfield  {author} {\bibinfo {author} {\bibfnamefont {X.}~\bibnamefont
  {Bonet-Monroig}}, \bibinfo {author} {\bibfnamefont {R.}~\bibnamefont
  {Sagastizabal}}, \bibinfo {author} {\bibfnamefont {M.}~\bibnamefont
  {Singh}},\ and\ \bibinfo {author} {\bibfnamefont {T.}~\bibnamefont
  {O'Brien}},\ }\href {https://doi.org/10.1103/PhysRevA.98.062339} {\bibfield
  {journal} {\bibinfo  {journal} {Physical Review A}\ }\textbf {\bibinfo
  {volume} {98}},\ \bibinfo {pages} {062339} (\bibinfo {year}
  {2018})}\BibitemShut {NoStop}%
\bibitem [{\citenamefont {Strikis}\ \emph {et~al.}(2021)\citenamefont
  {Strikis}, \citenamefont {Qin}, \citenamefont {Chen}, \citenamefont
  {Benjamin},\ and\ \citenamefont {Li}}]{strikis2020learningbased}%
  \BibitemOpen
  \bibfield  {author} {\bibinfo {author} {\bibfnamefont {A.}~\bibnamefont
  {Strikis}}, \bibinfo {author} {\bibfnamefont {D.}~\bibnamefont {Qin}},
  \bibinfo {author} {\bibfnamefont {Y.}~\bibnamefont {Chen}}, \bibinfo {author}
  {\bibfnamefont {S.~C.}\ \bibnamefont {Benjamin}},\ and\ \bibinfo {author}
  {\bibfnamefont {Y.}~\bibnamefont {Li}},\ }\href
  {https://doi.org/10.1103/PRXQuantum.2.040330} {\bibfield  {journal} {\bibinfo
   {journal} {PRX Quantum}\ }\textbf {\bibinfo {volume} {2}},\ \bibinfo {pages}
  {040330} (\bibinfo {year} {2021})}\BibitemShut {NoStop}%
\bibitem [{\citenamefont {Lowe}\ \emph {et~al.}(2021)\citenamefont {Lowe},
  \citenamefont {Gordon}, \citenamefont {Czarnik}, \citenamefont {Arrasmith},
  \citenamefont {Coles},\ and\ \citenamefont {Cincio}}]{lowe2020unified}%
  \BibitemOpen
  \bibfield  {author} {\bibinfo {author} {\bibfnamefont {A.}~\bibnamefont
  {Lowe}}, \bibinfo {author} {\bibfnamefont {M.~H.}\ \bibnamefont {Gordon}},
  \bibinfo {author} {\bibfnamefont {P.}~\bibnamefont {Czarnik}}, \bibinfo
  {author} {\bibfnamefont {A.}~\bibnamefont {Arrasmith}}, \bibinfo {author}
  {\bibfnamefont {P.~J.}\ \bibnamefont {Coles}},\ and\ \bibinfo {author}
  {\bibfnamefont {L.}~\bibnamefont {Cincio}},\ }\href
  {https://doi.org/10.1103/PhysRevResearch.3.033098} {\bibfield  {journal}
  {\bibinfo  {journal} {Phys. Rev. Research}\ }\textbf {\bibinfo {volume}
  {3}},\ \bibinfo {pages} {033098} (\bibinfo {year} {2021})}\BibitemShut
  {NoStop}%
\bibitem [{\citenamefont {Botelho}\ \emph {et~al.}(2022)\citenamefont
  {Botelho}, \citenamefont {Glos}, \citenamefont {Kundu}, \citenamefont
  {Miszczak}, \citenamefont {Salehi},\ and\ \citenamefont
  {Zimbor\'as}}]{botelho2022error}%
  \BibitemOpen
  \bibfield  {author} {\bibinfo {author} {\bibfnamefont {L.}~\bibnamefont
  {Botelho}}, \bibinfo {author} {\bibfnamefont {A.}~\bibnamefont {Glos}},
  \bibinfo {author} {\bibfnamefont {A.}~\bibnamefont {Kundu}}, \bibinfo
  {author} {\bibfnamefont {J.~A.}\ \bibnamefont {Miszczak}}, \bibinfo {author}
  {\bibfnamefont {O.}~\bibnamefont {Salehi}},\ and\ \bibinfo {author}
  {\bibfnamefont {Z.}~\bibnamefont {Zimbor\'as}},\ }\href
  {https://doi.org/10.1103/PhysRevA.105.022441} {\bibfield  {journal} {\bibinfo
   {journal} {Phys. Rev. A}\ }\textbf {\bibinfo {volume} {105}},\ \bibinfo
  {pages} {022441} (\bibinfo {year} {2022})}\BibitemShut {NoStop}%
\bibitem [{\citenamefont {Pérez-Salinas}\ \emph {et~al.}(2023)\citenamefont
  {Pérez-Salinas}, \citenamefont {Wang},\ and\ \citenamefont
  {Bonet-Monroig}}]{perezsalinas2023analyzing}%
  \BibitemOpen
  \bibfield  {author} {\bibinfo {author} {\bibfnamefont {A.}~\bibnamefont
  {Pérez-Salinas}}, \bibinfo {author} {\bibfnamefont {H.}~\bibnamefont
  {Wang}},\ and\ \bibinfo {author} {\bibfnamefont {X.}~\bibnamefont
  {Bonet-Monroig}},\ }\href {https://arxiv.org/abs/2303.16893} {\bibinfo
  {title} {Analyzing variational quantum landscapes with information content}}
  (\bibinfo {year} {2023}),\ \Eprint {https://arxiv.org/abs/2303.16893}
  {arXiv:2303.16893 [quant-ph]} \BibitemShut {NoStop}%
\bibitem [{\citenamefont {Moussa}\ \emph {et~al.}(2020)\citenamefont {Moussa},
  \citenamefont {Calandra},\ and\ \citenamefont
  {Dunjko}}]{moussa2020toquantum}%
  \BibitemOpen
  \bibfield  {author} {\bibinfo {author} {\bibfnamefont {C.}~\bibnamefont
  {Moussa}}, \bibinfo {author} {\bibfnamefont {H.}~\bibnamefont {Calandra}},\
  and\ \bibinfo {author} {\bibfnamefont {V.}~\bibnamefont {Dunjko}},\ }\href
  {https://doi.org/10.1088/2058-9565/abb8e5} {\bibfield  {journal} {\bibinfo
  {journal} {Quantum Science and Technology}\ }\textbf {\bibinfo {volume}
  {5}},\ \bibinfo {pages} {044009} (\bibinfo {year} {2020})}\BibitemShut
  {NoStop}%
\bibitem [{\citenamefont {Moussa}\ \emph {et~al.}(2022)\citenamefont {Moussa},
  \citenamefont {van Rijn}, \citenamefont {B{\"a}ck},\ and\ \citenamefont
  {Dunjko}}]{moussa2022hyperparameter}%
  \BibitemOpen
  \bibfield  {author} {\bibinfo {author} {\bibfnamefont {C.}~\bibnamefont
  {Moussa}}, \bibinfo {author} {\bibfnamefont {J.~N.}\ \bibnamefont {van
  Rijn}}, \bibinfo {author} {\bibfnamefont {T.}~\bibnamefont {B{\"a}ck}},\ and\
  \bibinfo {author} {\bibfnamefont {V.}~\bibnamefont {Dunjko}},\ }in\
  \href@noop {} {\emph {\bibinfo {booktitle} {Discovery Science}}},\ \bibinfo
  {editor} {edited by\ \bibinfo {editor} {\bibfnamefont {P.}~\bibnamefont
  {Pascal}}\ and\ \bibinfo {editor} {\bibfnamefont {D.}~\bibnamefont {Ienco}}}\
  (\bibinfo  {publisher} {Springer Nature Switzerland},\ \bibinfo {address}
  {Cham},\ \bibinfo {year} {2022})\ pp.\ \bibinfo {pages} {32--46}\BibitemShut
  {NoStop}%
\bibitem [{\citenamefont {{Moussa, Charles}}\ \emph {et~al.}(2022)\citenamefont
  {{Moussa, Charles}}, \citenamefont {{Wang, Hao}}, \citenamefont {{B\"ack,
  Thomas}},\ and\ \citenamefont {{Dunjko, Vedran}}}]{moussa2022unsupervised}%
  \BibitemOpen
  \bibfield  {author} {\bibinfo {author} {\bibnamefont {{Moussa, Charles}}},
  \bibinfo {author} {\bibnamefont {{Wang, Hao}}}, \bibinfo {author}
  {\bibnamefont {{B\"ack, Thomas}}},\ and\ \bibinfo {author} {\bibnamefont
  {{Dunjko, Vedran}}},\ }\href
  {https://doi.org/10.1140/epjqt/s40507-022-00131-4} {\bibfield  {journal}
  {\bibinfo  {journal} {EPJ Quantum Technol.}\ }\textbf {\bibinfo {volume}
  {9}},\ \bibinfo {pages} {11} (\bibinfo {year} {2022})}\BibitemShut {NoStop}%
\bibitem [{\citenamefont {Ito}(2023)}]{ito2023latency}%
  \BibitemOpen
  \bibfield  {author} {\bibinfo {author} {\bibfnamefont {K.}~\bibnamefont
  {Ito}},\ }\href@noop {} {\bibfield  {journal} {\bibinfo  {journal} {arXiv
  preprint arXiv:2302.04422}\ } (\bibinfo {year} {2023})}\BibitemShut {NoStop}%
\bibitem [{\citenamefont {Perrier}\ \emph {et~al.}(2022)\citenamefont
  {Perrier}, \citenamefont {Youssry},\ and\ \citenamefont
  {Ferrie}}]{perrier2022qdataset}%
  \BibitemOpen
  \bibfield  {author} {\bibinfo {author} {\bibfnamefont {E.}~\bibnamefont
  {Perrier}}, \bibinfo {author} {\bibfnamefont {A.}~\bibnamefont {Youssry}},\
  and\ \bibinfo {author} {\bibfnamefont {C.}~\bibnamefont {Ferrie}},\ }\href
  {https://doi.org/10.1038/s41597-022-01639-1} {\bibfield  {journal} {\bibinfo
  {journal} {Scientific Data}\ }\textbf {\bibinfo {volume} {9}},\ \bibinfo
  {pages} {1} (\bibinfo {year} {2022})}\BibitemShut {NoStop}%
\bibitem [{\citenamefont {Placidi}\ \emph {et~al.}(2023)\citenamefont
  {Placidi}, \citenamefont {Hataya}, \citenamefont {Mori}, \citenamefont
  {Aoyama}, \citenamefont {Morisaki}, \citenamefont {Mitarai},\ and\
  \citenamefont {Fujii}}]{placidi2023mnisq}%
  \BibitemOpen
  \bibfield  {author} {\bibinfo {author} {\bibfnamefont {L.}~\bibnamefont
  {Placidi}}, \bibinfo {author} {\bibfnamefont {R.}~\bibnamefont {Hataya}},
  \bibinfo {author} {\bibfnamefont {T.}~\bibnamefont {Mori}}, \bibinfo {author}
  {\bibfnamefont {K.}~\bibnamefont {Aoyama}}, \bibinfo {author} {\bibfnamefont
  {H.}~\bibnamefont {Morisaki}}, \bibinfo {author} {\bibfnamefont
  {K.}~\bibnamefont {Mitarai}},\ and\ \bibinfo {author} {\bibfnamefont
  {K.}~\bibnamefont {Fujii}},\ }\href@noop {} {\bibfield  {journal} {\bibinfo
  {journal} {arXiv preprint arXiv:2306.16627}\ } (\bibinfo {year}
  {2023})}\BibitemShut {NoStop}%
\end{thebibliography}%
\bibliographystyle{apsrev4-2}

\appendix

\section{Constructing unbiased estimators for the Log-Likelihood} \label{app:log_likelihood}
Here we review a conceptually different approach for a shot-frugal optimization of a problem utilizing a Log-likelihood loss function. In particular, we will follow the results in~\cite{van2020unbiased}. First, we assume that one is given a dataset of the form $\SC=\{\rho_i,y_i\}_{i=1}^N$ where $\rho_i$ are quantum states in some domain $\RC$, and $y_i$ are discrete real-valued labels from some domain $\YC$ (e.g., in binary classification $\YC=\{0,1\}$). Then, we assume that one has a trainable parametrized model $h_{\thv}:\RC\rightarrow \YC$ whose goal is to predict the correct labels. Note that here one cannot directly apply the techniques previously used, as discrete outputs preclude gradient calculations and the use of the results in~\cite{ balles2017coupling,kubler2020adaptive}. Hence, a different shot-frugal method must be employed. 

Given $\SC$ and $h_{\thv}$, we can define  the likelihood function 
\begin{equation}
    \Pr(\SC|\thv):=\prod_{i=1}^N\Pr(y_i|\rho_i;\thv)=\prod_{i=1}^Np_i\,,
\end{equation}
where we have defined $p_i=\Pr(y_i|\rho_i;\thv)$ and where assumed that the model's predictions only depend on the current input. From here, we can define the negative log-likelihood loss function as
\begin{equation}\label{eq:neg-log-loss-app}
    \LC(\thv)=-\log\left(\Pr(\SC|\thv)\right)=-\sum_{i=1}^N\log\left(p_i\right)\,.
\end{equation}
Interestingly, we can reduce the task of estimating the negative log-likelihood loss function to a Bernoulli sampling problem. Namely, we can assume without loss of generality that the model specifies a stochastic function $g$ that takes as input $\rho_i$, the set of parameters $\thv$ and outputs a prediction $y_i$. That is, $y_i\sim g(\rho,
\thv)$. Then, we remark that to estimate the loss function, we do not need to know the explicit form of $g$, one only needs to know the probability for a random sample $y$ from the model to match the true label $y_i$. Thus, we can convert each output from the model into a variable
\begin{equation}
    x=\begin{cases}1\quad \text{if $y=y_i$}\,\\
    0\quad \text{otherwise}
    \end{cases}\,.
\end{equation}
By construction, $x$ follows a Bernoulli distribution with probability $p_i$. Thus, we can estimate the log-likelihood loss function by drawing samples $(x_1,x_2,\ldots)$ from said Bernoulli distribution. In most common cases these samples can be thought of as single shot-estimates of the model's output, and our goal is to minimize this number of samples/shots.  

In~\cite{van2020unbiased} the authors study the task of constructing an unbiased estimator for the negative log-likelihood loss function of Eq.~\eqref{eq:neg-log-loss-app}, where they consider a fixed number of shots strategy (one uses $s$ fixed shots) and the inverse binomial sampling (IBS) strategy (where one keeps drawing samples until $x_k=1$). One can show that the fixed number of shots strategy always leads to a  biased estimator of $\log(p)$ (for the special case when one estimates $p=\sum_{k=1}^s x_k/s$, the estimator can even be infinitely biased if $x_k$ is always zero!). On the other hand, using IBS one can define the estimator 
\begin{equation}
    \widehat{\LC}(\thv)=\begin{cases}0\quad \text{for $K$=1}\,\\
    \sum_{k=1}^{K-1}\frac{1}{k}\quad \text{for $K>1$}
    \end{cases}\,,
\end{equation}
where $K$ is the number of shots needed to obtain a result $x_k=1$. Notably, it can be shown that $E[\widehat{\LC}(\thv)]=\LC(\thv)$, meaning that IBS provides an unbiased estimator for the negative log-likelihood loss function (see ~\cite{van2020unbiased}). The number of shots that IBS takes trial with probability $p_i$ is geometrically distributed with mean $1/p_i$. For instance, if $p_i\geq 0.5$, IBS can require only 1 or 2 shots, while it will assign more shots to cases of small $p_i$. We also remark that in Ref.~\cite{van2020unbiased} the authors derive the estimator variance, and show that is minimal and non-diverging if $p_i\rightarrow 0$. This method then provides a shot-frugal, minimum-variance, unbiased estimator for the negative log-likelihood loss function which can be combined with gradient-free optimizers to train the parameters in the mode. We further refer the reader to Ref.~\cite{van2020unbiased} for a discussion on how to extend these results to models $\ell_{\thv}$ with continuous outputs. 
\section{Using U-statistics to construct unbiased estimators}
\label{app:U-statistics}
The theory of U-statistics was initially introduced by Hoeffding in the late 1940s \cite{hoeffding1948class} and has a wide range of applications. U-statistics presents a methodology on how to use observations of estimable parameters to construct minimum variance unbiased estimates of more complex functions of the estimable parameters. In the cases we are interested in here, the estimable parameter is the expectation value of an observable used when constructing the cost function.

First, let $\mathcal{P}$ be a family of probability measures on an arbitrary measurable space. A probability measure is defined as the probability of some particular event occurring. Let $\Theta(P)$ be a real-valued function defined for $P \in \mathcal{P}$. We define $\Theta(P)$ to be an estimable parameter within $\mathcal{P}$ if there exists some integer $r$ for which an unbiased estimator of $\Theta(P)$ exists constructed from $r$ independent identically distributed (i.i.d) random variables according to $P$. In other words, there is a measurable, real-valued function $h(x_1,...,x_r)$ such that
\begin{equation}
    \mathbb{E}_{P}[(h)(x_1,...,x_r)] = \Theta(P) \quad \forall P \in \mathcal{P},
\end{equation}
where $x_1,...,x_m$ are the i.i.d with distribution $P$. The degree of $\Theta(P)$ is defined to be the smallest integer $r$ with this property. 

The function $h(x_1,..,x_r)$ is also referred to as the kernel. This function can be symmetrized by summing over $h$ with all possible permutations of the random variables $x_1,...,x_n$. That is, the symmetrized version of the kernel can be defined as follows
\begin{equation}
    h^{*}(x_1,...,x_r) = \frac{1}{r!} \sum_{\sigma \in S_r} h(x_{\sigma(1)},...,x_{\sigma(r)}),
\end{equation}
where the summation is over $S_r$, the group of all permutations of a set of $r$ indexes. With these definitions in hand, we can construct a U-statistic. 

Assume we are given $s>r$ samples from $\mathcal{P}$, the U-statistic with kernel $h$ is defined as,
\begin{equation}
    U_s = U(h) = \frac{1}{\binom{s}{r}} \sum_{C^{s}_{r}} h^{*}(x_{i_{1}},...,x_{i_{r}}),
\end{equation}
where the summation is over the set $C^{s}_{r}$ of $ \binom{s}{r} = \frac{s!}{r!(s-r)!}$ combinations of $r$ integers chosen from $(1,2,...,s)$.

As the U-statistic defined above is simply an average of unbiased estimates for $\Theta(P)$ and is therefore also an unbiased estimator for $\Theta(P)$. Indeed, it can be proved that the U-statistic is the best-unbiased estimate for $\Theta(P)$ in that it has the smallest variance. 

\subsection{Examples of U-statistics}
Now that we have introduced U-statistics we can build intuition with some simple examples. Firstly, let us consider how to construct an unbiased estimate for the mean $\mu$. Recall we first need to construct a kernel function $h$. In this case, the kernel will have degree $1$ and can simply be defined as $h(x_i) = x_i$. Note this has the required property that $\mathbb{E}[h(x_{i})] = \mathbb{E}[x_{i}] = \mu$. Furthermore, in this case $h^{*}(x_{i}) = h(x_{i})$, i.e. the kernel is already symmetric. Therefore, the U-statistic for $\mu$ constructed from $s$ samples can be written as,
\begin{align}
    U_s(h) &=  \frac{1}{\binom{s}{1}} \sum_{C^{s}_{1}} x_i \\
    &=\frac{1}{s}\sum_{i=1}^{s}x_i,
\end{align}
which is the well-known formula of the sample average, and which is used to estimate the mean. 

Similarly, we can consider how to construct a U-statistic for the variance. Consider the kernel $h(x_i, x_j) = x_{i}^{2} - x_ix_j$. This kernel has degree $2$. Note once again, this has the required property that $\mathbb{E}[h(x_i, x_j)] = \mathbb{E}(x_{i}^{2}) - \mathbb{E}(x_ix_j) = \mathbb{E}(x_{i}^{2}) - \mu^{2}$, which is the definition of the variance. Then this needs to be symmetrized, $h^{*} = 1/2 (x_{i}^{2} - x_ix_j + x_{j}^{2} - x_jx_i) = 1/2( x_{i}^{2} + x_{j}^{2} - 2x_ix_j) = 1/2(x_i - x_j)^{2}$. Therefore, the U-statistic for the variance constructed with the kernel $h$ using $s$ samples is,
\begin{align}
    U_s(h) &=  \frac{1}{\binom{s}{2}} \sum_{C^{s}_{2}} \frac{(x_i - x_j)^{2}}{2} \\
    &= \frac{1}{{s(s-1)}} \sum_{C^{s}_{2}} (x_i - x_j)^{2},
\end{align}
which corresponds to the sample variance.

As a final example let's consider constructing an estimator of the function $\Theta(P) = P^{k}$. Lets define the kernel as $h(x_1,..,x_k) = \prod_{i} x_{i}$, which clearly fulfills the properties we require. This kernel is already symmetric such that $h(x_1,...,x_k) = h^{*}(x_1,...,x_m)$. Therefore, the U-statistic can be defined as,
\begin{align}
    U_s(h) &=  \frac{1}{\binom{s}{k}} \sum_{C^{s}_{k}} \prod_{i} x_{i} \,.
\end{align}
This is exactly the form of estimators that we require in the main text. For our purposes, the number of samples $s$ is a random variable that is constructed probabilistically.

The framework of U-statistics is perfectly suited to our purposes. Indeed, it provides a recipe to construct unbiased estimators for polynomial functions of variables with some probability distribution. Furthermore, these estimators have the smallest possible variance. Finally, the framework shows the minimum number of shots necessary to estimate a given term through the degree of the kernel. 

\section{Variance of estimator for linear loss functions}
\label{app:varlinear}

In section~\ref{unbiasedestimatorslinearsection}, we constructed an unbiased estimator for the loss function. Denoting the latter $\widehat{\LC}$, we proved in Prop.~\ref{prop-unbiased-mother-cost}  that $\mathbb{E}[\widehat{\LC}] = \LC(\thv)$. We also obtained the variance of such an estimator, which can be used to compare different shot-allocation strategies, similar to \cite{arrasmith2020operator} for the Rosalin optimizer designed for VQE. However, we use elements in the variance calculation in Appendix \ref{app:mse} to construct an unbiased estimator for the MSE loss function. 

\begin{proposition}
\label{prop-var-mother-cost}
The variance of $\widehat{\LC}$ is
\begin{align}
    {\rm Var}(\widehat{\LC}) &= \sum_{\vec{i}, j}  \frac{q_{\vec{i},j}^2}{\mathbb{E}[s_{\vec{i},j}]} \sigma_{\vec{i},j}^2 \\ \nonumber
    &+ \sum_{\vec{i}, j, \vec{i'}, j'}  \frac{q_{\vec{i},j} q_{\vec{i'},j'}}{\mathbb{E}[s_{\vec{i},j}] \mathbb{E}[s_{\vec{i'},j'}]} \\ \nonumber
  & \quad \big[\langle h_{\vec{i},j} (\thv)\rangle \langle h_{\vec{i'},j'} (\thv) \rangle {\rm Cov}[ s_{\vec{i},j}, s_{\vec{i'},j'} \big],
 \label{eq:genvariance}
\end{align}
where $\sigma_{\vec{i},j}^2 = \langle h^{2}_{\vec{i},j} (\thv) \rangle - (\langle h_{\vec{i},j} (\thv) \rangle)^2$ and

The covariance is defined as
\begin{equation}
  {\rm Cov}[ s_{\vec{i},j}, s_{\vec{i'},j'}] = \mathbb{E}[ s_{\vec{i},j} s_{\vec{i'},j'} ] - \mathbb{E}[s_{\vec{i},j}] \mathbb{E}[s_{\vec{i'},j'}].  
\end{equation}

\end{proposition}
\label{app:proof-prop-unbiased-simple-mother-cost}
\begin{proof}
Using the definition of the variance and recalling that  $\mathbb{E}[\widehat{\LC}] = \LC(\thv)$, we can write
\begin{align}\label{vareq}
    {\rm Var}(\widehat{\LC})
    &= \mathbb{E}[\widehat{\LC}^2] - \mathbb{E}[\widehat{\LC}]^2 \nonumber \\
    &=  \sum_{\vec{i}, j, \vec{i'}, j'} \frac{q_{\vec{i},j} q_{\vec{i'},j'}}{\mathbb{E}[s_{\vec{i},j}] \mathbb{E}[s_{\vec{i'},j'}]} \mathbb{E}[\sum_{k=1}^{s_{\vec{i},j}} \sum_{k'=1}^{s_{\vec{i}',j'}} r_{\vec{i},j,k} r_{\vec{i'},j',k'} ] \nonumber \\
    & \quad - \LC(\thv)^2\,.
\end{align}
We can then decompose the first term into its constituent parts where indices in the sum are different and equal 
\begin{align}
    & \sum_{\vec{i}, j, \vec{i'}, j'} \frac{q_{\vec{i},j} q_{\vec{i'},j'}}{\mathbb{E}[s_{\vec{i},j}] \mathbb{E}[s_{\vec{i'},j'}]} \mathbb{E}[\sum_{k=1}^{s_{\vec{i},j}} \sum_{k'=1}^{s_{\vec{i}',j'}} r_{\vec{i},j,k} r_{\vec{i'},j',k'} ] \nonumber\\ 
    &= \sum_{\vec{i}, j, \vec{i'}, j', \vec{i} \ne \vec{i'}, j \ne j'} \frac{q_{\vec{i},j} q_{\vec{i'},j'}}{\mathbb{E}[s_{\vec{i},j}] \mathbb{E}[s_{\vec{i'},j'}]} E_1 \nonumber\\
    &+ \sum_{\vec{i}, j, \vec{i'}, j', \vec{i} = \vec{i'}, j \ne j'} \frac{q_{\vec{i},j} q_{\vec{i},j'}}{\mathbb{E}[s_{\vec{i},j}] \mathbb{E}[s_{\vec{i},j'}]} E_2 \nonumber\\
    &+ \sum_{\vec{i}, j, \vec{i'}, j', \vec{i} \ne \vec{i'}, j = j'} \frac{q_{\vec{i},j} q_{\vec{i'},j}}{\mathbb{E}[s_{\vec{i},j}] \mathbb{E}[s_{\vec{i'},j}]} E_3 \nonumber\\
    &+ \sum_{\vec{i}, j} \frac{q_{\vec{i},j}^2}{\mathbb{E}^2[s_{\vec{i},j}]} E_4,
\end{align}
where 
\begin{align} \label{separationsumvar}
    & E_1 =  \mathbb{E}[\sum_{k=1}^{s_{\vec{i},j}} \sum_{k'=1}^{s_{\vec{i}',j'}} r_{\vec{i},j,k} r_{\vec{i'},j',k'} ] \\
  & E_2 =  \mathbb{E}[\sum_{k=1}^{s_{\vec{i},j}} \sum_{k'=1}^{s_{\vec{i},j'}} r_{\vec{i},j,k} r_{\vec{i},j',k'} ] \\
  & E_3 =  \mathbb{E}[\sum_{k=1}^{s_{\vec{i},j}} \sum_{k'=1}^{s_{\vec{i}',j}} r_{\vec{i},j,k} r_{\vec{i'},j,k'} ] \\
  & E_4 =  \mathbb{E}[\sum_{k=1}^{s_{\vec{i},j}} \sum_{k'=1}^{s_{\vec{i},j}} r_{\vec{i},j,k} r_{\vec{i},j,k'} ].
\end{align}
Note that since $\frac{1}{\mathbb{E}[s_{\vec{i},j}]}\sum_{k=1}^{s_{\vec{i},j}} r_{\vec{i},j,k}$ is an unbiased estimator of $\langle h_{\vec{i},j} (\thv) \rangle$, by application of Wald's equation we obtain  
\begin{align}\label{App:Varnotequalcase}
    E_1 &= \langle h_{\vec{i},j} (\thv) \langle h_{\vec{i'},j'} (\thv) \rangle \mathbb{E}[ s_{\vec{i},j} s_{\vec{i'},j'}] \\
    E_2 &= \langle h_{\vec{i},j} (\thv) \rangle \langle h_{\vec{i},j'} (\thv) \rangle \mathbb{E}[ s_{\vec{i},j} s_{\vec{i},j'}] \\
    E_3 &= \langle h_{\vec{i},j} (\thv) \rangle \langle h_{\vec{i'},j} (\thv) \rangle \mathbb{E}[ s_{\vec{i},j} s_{\vec{i'},j}].
\end{align}\label{App:Varnotequalcasebutequali}
For the case $\{\vec{i} = \vec{i'}, j = j'\}$, care needs to be taken when expanding the sum as the terms are not independent when $k' = k$. Therefore, we decompose the sum into the cases where $k \neq k'$ and $k =k'$,
\begin{align}\label{App:Varequalcase}
&E_4 = \mathbb{E}[\sum_{k=1}^{s_{\vec{i},j}} \sum_{k'=1, k' \ne k}^{s_{\vec{i},j}} r_{\vec{i},j,k} r_{\vec{i},j,k'} ] + \mathbb{E}[\sum_{k=1}^{s_{\vec{i},j}}  r_{\vec{i},j,k}^2 ] \nonumber \\
&= E_{4, 1} + E_{4, 2}.
\end{align}
Once again, using Wald's equation, but taking into account that the number of entries for the $k \ne k'$ is $(s_{\vec{i},j}-1)$, one can show
\begin{equation}\label{App:Varequalcasesquareoutside}
  E_{4, 1} =  \mathbb{E}[s_{\vec{i},j} (s_{\vec{i},j} - 1)] (\langle h_{\vec{i},j} (\thv) \rangle)^2. 
\end{equation}
Finally, using the fact that $ \frac{1}{\mathbb{E}[s_{\vec{i},j}]} \sum_{k=1}^{s_{\vec{i},j}}  r_{\vec{i},j,k}^2$ is an unbiased estimator of $\langle h^{2}_{\vec{i},j} (\thv) \rangle$, we obtain:
\begin{equation}\label{App:Varequalcasesquareinside}
    E_{4, 2} = \mathbb{E}[s_{\vec{i},j}] \langle h^{2}_{\vec{i},j} (\thv) \rangle\,.
\end{equation}
Bringing all the terms together and writing the expressions containing $E_1$, $E_2$ and $E_3$ in a single sum, we obtain

\begin{align}
    {\rm Var}(\widehat{\LC}) &= \sum_{\vec{i}, j} \frac{q_{\vec{i},j}^2}{\mathbb{E}^2[s_{\vec{i},j}]}  \mathbb{E}[s_{\vec{i},j} (s_{\vec{i},j} - 1)] (\langle h_{\vec{i},j} (\thv) \rangle)^2 \nonumber\\
    &+ \sum_{\vec{i}, j} \frac{q_{\vec{i},j}^2}{\mathbb{E}^2[s_{\vec{i},j}]}  \mathbb{E}[s_{\vec{i},j}] \langle h^{2}_{\vec{i},j} (\thv) \rangle - \LC(\thv)^2 \nonumber \\
    &+ \sum_{\substack{\vec{i}, j, \vec{i'}, j'  \\  \vec{i} \neq \vec{i'} \text{ if } j = j' \\ j \neq j' \text{ if } \vec{i} = \vec{i'}}}  \frac{q_{\vec{i},j} q_{\vec{i'},j'}}{\mathbb{E}[s_{\vec{i},j}] \mathbb{E}[s_{\vec{i'},j'}]} \langle h_{\vec{i},j} (\thv) \rangle \langle h_{\vec{i'},j'} (\thv) \rangle \nonumber \\
    & \mathbb{E}[ s_{\vec{i},j} s_{\vec{i'},j'} ]\,.
\end{align}

Expanding the second term and factorizing leads to,
\begin{align}
    {\rm Var}(\widehat{\LC}) 
    &= \sum_{\vec{i}, j}\frac{q_{\vec{i},j}^2}{\mathbb{E}[s_{\vec{i},j}]} [\langle h^{2}_{\vec{i},j} (\thv) \rangle - [\langle h_{\vec{i},j} (\thv) \rangle ]^2] \nonumber \\
     & + \sum_{\vec{i}, j, \vec{i'}, j'}  \frac{q_{\vec{i},j} q_{\vec{i'},j'}}{\mathbb{E}[s_{\vec{i},j}] \mathbb{E}[s_{\vec{i'},j'}]} \mathbb{E}[ s_{\vec{i},j} s_{\vec{i'},j'} ] \nonumber \\ 
     & \langle h_{\vec{i},j} (\thv) \rangle  \langle h_{\vec{i'},j'} (\thv) \rangle   
     - \LC(\thv)^2.
\end{align}
Which we can then rewrite to prove the desired result,
\begin{align}
&{\rm Var}(\widehat{\LC}) = \sum_{\vec{i}, j}  \frac{q_{\vec{i},j}^2}{\mathbb{E}[s_{\vec{i},j}]} \sigma_{\vec{i},j}^2 \nonumber \\
 & + \sum_{\vec{i}, j, \vec{i'}, j'}  \frac{q_{\vec{i},j} q_{\vec{i'},j'}}{\mathbb{E}[s_{\vec{i},j}] \mathbb{E}[s_{\vec{i'},j'}]} \langle h_{\vec{i},j} (\thv) \rangle \langle h_{\vec{i'},j'} (\thv) \rangle {\rm Cov}[ s_{\vec{i},j}, s_{\vec{i'},j'}]\,. 
\end{align}
\end{proof}

\section{Constructing an unbiased estimators of loss functions with polynomial dependence}
\label{app:unbiasedlosspolyD}

We consider here how to construct estimators for the loss functions with an arbitrary polynomial of degree $D$ in the observables, expressed as:
\begin{align} \label{costpolynomialintrace}
    \LC(\thv) &= \sum_{\vec{i}} \sum_{z=0}^D p_{\vec{i}} a_z (\sum_{j} c_{\vec{i},j} \langle h_{\vec{i}, j} (\thv)\rangle)^z \\
    &= \sum_{\vec{i},z} p_{\vec{i},z} (\sum_{j} c_{\vec{i},j} \langle h_{\vec{i}, j} (\thv)\rangle)^z,
\end{align}
where $ p_{\vec{i},z} = p_{\vec{i}} a_z, a_z \in \mathbb{R}$.
Using the multinomial theorem, given $J$ hamiltonian terms, we can develop the power of the sum into a sum of $\binom{z+J-1}{J}$ terms:
\begin{equation} \label{costpolynomialintracemultinomial}
     \LC(\thv) =  \sum_{\vec{i},z} p_{\vec{i},z} \sum_{\substack{b_1+ b_2+\\ \cdots +b_J=z}} \binom{z}{b_1, b_2, \cdots b_J} \prod_j (c_{\vec{i},j} \langle h_{\vec{i}, j} (\thv)\rangle)^{b_j}\,,
\end{equation}
where $\binom{z}{b_1, b_2, \cdots b_J} = \frac{z!}{b_1! b_2! \cdots b_J!}$ and $b_j$ are non-negative integers. First, one can generalize the formula for an unbiased estimator of $(\langle h_{\vec{i}, j} (\thv)\rangle)^{b_j}$:
\begin{align} \label{app:estimatorwithpowers}
        & \widehat{\mathcal{E}}_{\vec{i},j} = \frac{1}{  \mathbb{E}[ \prod_{t=0}^{b_j-1} (s_{\vec{i},j} - t)]} \prod_{\alpha=1}^{b_j}(\sum_{k_{\alpha}=1}^{s_{\vec{i},j}} r_{\vec{i}, j, k_{\alpha}})\,.
\end{align}
The case $b_j=2$ was elaborated in Eq.\ref{moment2tijest}.

For $ \prod_j \langle h_{\vec{i}, j} (\thv)\rangle^{b_j}$, we define an unbiased estimator as:
\begin{equation} \label{estimatorofmomentsforpowers}
    \widehat{\mathcal{R}}_{\vec{i},z} = \frac{1}{  \mathbb{E}[\prod_j \prod_{t=0}^{b_j-1} (s_{\vec{i},j} - t)]} \prod_j \prod_{\alpha=1}^{b_j} (\sum_{k_{j,\alpha}}^{s_{\vec{i},j}} r_{\vec{i}, j, k_{j, \alpha}})\,.
\end{equation}
This involves estimating $\binom{z+J-1}{J}J\sum\limits_{j = 1}^{J}b_{j}$ terms to construct the following estimator for the cost:
\begin{equation} \label{estimatorpolynomialintracemultinomial}
     \widehat{\LC} =  \sum_{\vec{i},z} p_{\vec{i},z} \sum_{\substack{b_1+ b_2+\\ \cdots +b_J=z}} \binom{z}{b_1, b_2, \cdots b_J} [\prod_j (c_{\vec{i},j})^{b_j}] \widehat{\mathcal{R}}_{\vec{i},z}\,,
\end{equation}
and establish the following proposition.
\begin{proposition}
\label{prop5}
$\widehat{\LC}$ is an unbiased estimator for $\LC$.
\end{proposition}

\begin{proof}
Taking the expectation over the above definitions:
\begin{align}
    & \mathbb{E} [\widehat{\LC}] =  \sum_{\vec{i},z} p_{\vec{i},z} \sum_{\substack{b_1+ b_2+\\ \cdots +b_J=z}} \binom{z}{b_1, b_2, \cdots b_J} [\prod_j (c_{\vec{i},j})^{b_j}] \mathbb{E} [\widehat{\mathcal{R}}_{\vec{i},z}] \nonumber \\
    &= \sum_{\vec{i},z} p_{\vec{i},z} \sum_{b_1+ b_2+ \cdots +b_J=z} \binom{z}{b_1, b_2, \cdots b_J} [\prod_j (c_{\vec{i},j})^{b_j}] \nonumber \\ 
    &\frac{1}{  \mathbb{E}[\prod_j \prod_{t=0}^{b_j-1} (s_{\vec{i},j} - t)]} \mathbb{E} [\prod_j \prod_{\alpha=1}^{b_j} \sum_{k_{j,\alpha}}^{s_{\vec{i},j}} r_{\vec{i}, j, k_{j, \alpha}}]\,.
\end{align}
Using again Wald's inequality, we obtain:
\begin{align*}
    \mathbb{E} [\prod_j \prod_{\alpha=1}^{b_j} \sum_{k_{j,\alpha}}^{s_{\vec{i},j}} r_{\vec{i}, j, k_{j, \alpha}}] = \mathbb{E}[\prod_j \prod_{t=0}^{b_j-1} (s_{\vec{i},j} - t)] \prod_j \langle h_{\vec{i}, j} (\thv)\rangle^{b_j}\,.
\end{align*}
for which, when replacing in the previous equation, we obtain $ \mathbb{E} [\widehat{\LC}] = \LC(\thv)$.
\end{proof}

\section{Constructing an unbiased estimator for the Mean Squared Error Loss Function}
\label{app:mse}

The mean squared error (MSE) is widely used in machine learning. Here, we show how to construct an unbiased estimator for the MSE as an illustrative example of previously-studied polynomial function of the observables in Appendix~\ref{app:estimatorwithpowers}. We start by considering a simplified case for the MSE, where we do not decompose each measurement operator $H_i$ into a linear combination of Hermitian matrices. We then tackle the more complex case where $H_i$ is of the form presented in Eq.~\eqref{measurement_decomposed}. 

\subsection{Step 1: Simplified MSE}

Let $\{ \rho_{\vec{i}}, y_{\vec{i}}\}$ be a dataset of $N$ points, where the task consists of minimizing the following cost:

\begin{equation}\label{eq:msecostnocompact}
    \LC_{MSE} = \sum\limits_{\vec{i}} p_{\vec{i}} (y_{\vec{i}} - \langle h_{\vec{i}}(\thv) \rangle)^2,
\end{equation}
This corresponds to a weighted MSE cost. The most common case is when $p_{\vec{i}} = 1 / N$. We will take a step-by-step approach to construct unbiased estimators for such cost.

If we expand the brackets, we obtain,
\begin{equation*}
    \LC_{MSE}= \sum_{\vec{i}} p_{\vec{i}} (y_{\vec{i}}^2 - 2 y_{\vec{i}} \langle h_{\vec{i}}(\thv) \rangle + \langle h_{\vec{i}}(\thv) \rangle^2).
\end{equation*}
We denote $\widehat{\mathcal{E}_{\vec{i}}} $ as the estimator for $ \langle h_{\vec{i}}(\thv) \rangle$, $\widehat{\mathcal{Q}}_{\vec{i}} $ as the estimator for $ \langle h_{\vec{i}}(\thv) \rangle^2$.

\begin{align}\label{eq:estimatorsmse}
    \widehat{\mathcal{E}_{\vec{i}}} &= \frac{1}{\mathbb{E}[s_{\vec{i}}]}\sum_{k=1}^{s_{\vec{i}}} r_{\vec{i},k}\,, \\
    \widehat{\mathcal{Q}}_{\vec{i}} &= \frac{1}{\mathbb{E}[s_{\vec{i}} (s_{\vec{i}} - 1)]}\sum_{k=1}^{s_{\vec{i}}} \sum_{k'=1, k' \ne k}^{s_{\vec{i}}} r_{\vec{i},k} r_{\vec{i},k'}\,,
\end{align}
where $\widehat{\mathcal{E}_{\vec{i}}} $ is an unbiased estimator which is straightforward from Wald's equation. $\widehat{\mathcal{Q}}_{\vec{i}} $ is also an unbiased estimator for $ \langle h_{\vec{i}}(\thv) \rangle^2$, following the reasoning used in Eq.~\eqref{App:Varequalcasesquareoutside} in the proof of Prop.~\ref{prop-var-mother-cost}. Then, defining $\widehat{\LC}_{MSE}$ as
\begin{equation}\label{eq:unbiasedestimatorsmse}
    \widehat{\LC}_{MSE}= \sum_{\vec{i}} p_{\vec{i}} (y_{\vec{i}}^2 - 2 y_{\vec{i}} \widehat{\mathcal{E}_{\vec{i}}} + \widehat{\mathcal{Q}}_{\vec{i}}),
\end{equation}
we establish the following proposition.
\begin{proposition}
\label{prop-unbiased-simple-mse-cost}
$\widehat{\LC}_{MSE}$ is an unbiased estimator for $\LC_{MSE}$.
\end{proposition}
\begin{proof}
Using the definitions in Eq.~\eqref{eq:estimatorsmse} and Eq.~\eqref{eq:unbiasedestimatorsmse} we can write
\begin{align}
    & \mathbb{E}[\widehat{\LC}_{MSE}] = \mathbb{E}[\sum_{\vec{i}} p_{\vec{i}}   y_{\vec{i}}^2] - 2 \mathbb{E}[\sum_{\vec{i}} y_{\vec{i}} p_{\vec{i}} \frac{1}{\mathbb{E}[s_{\vec{i}}]}\sum_{k=1}^{s_{\vec{i}}} r_{\vec{i},k}] \nonumber\\ 
    &+ \mathbb{E}[ \sum_i p_i \frac{1}{\mathbb{E}[s_{\vec{i}} (s_{\vec{i}} - 1)]}\sum_{k=1}^{s_{\vec{i}}} \sum_{k'=1, k' \ne k}^{s_{\vec{i}}} r_{\vec{i},k} r_{\vec{i},k'}]\,. 
\end{align}
Using Wald's equation we can evaluate this expression to give
\begin{align}
    & \mathbb{E}[\widehat{\LC}_{MSE}] = \sum_{\vec{i}} p_i y_i^2 - 2 \sum_i y_i p_i \frac{\mathbb{E}[s_{\vec{i}}]}{\mathbb{E}[s_{\vec{i}}]} \langle h_{\vec{i}}(\thv) \rangle\nonumber\\
    &+ \sum_{\vec{i}} p_i \frac{\mathbb{E}[s_{\vec{i}} (s_{\vec{i}} - 1)]}{\mathbb{E}[s_{\vec{i}} (s_{\vec{i}} - 1)]} (\langle h_{\vec{i}}(\thv) \rangle)^{2}.
\end{align}
Finally, we can evaluate these terms and factorize, which leads to,
\begin{align}
     \mathbb{E}[\widehat{\LC}_{MSE}] &= \sum_{\vec{i}} p_i  (y_i^2 - 2 y_i \langle h_{\vec{i}}(\thv) \rangle + \langle h_{\vec{i}}(\thv) \rangle^2) \\
     & = \LC_{MSE}\,.
\end{align}
\end{proof}
This means that we can estimate this simplified MSE loss function by evaluating the first term exactly, sampling from the second term according to the probability distribution $p_{\vec{i}}y_{\vec{i}}$ and sampling the third term according to the distribution $p_{\vec{i}}$.

\subsection{Step 2: Full MSE}\label{hamiltoniantermsmsesection}

Now, we assume the Hamiltonians $H_{\vec{i}}$ can be decomposed into many terms such that $H_{\vec{i}} = \sum_j c_{\vec{i},j} H_{\vec{i},j}$. Hence the new full MSE cost can be expressed as
\begin{align}\label{eq:msecostnocompactwithhamiltonianterms}
    & \LC'_{MSE}= \sum_{\vec{i}} p_{\vec{i}} (y_{\vec{i}} - \sum_j c_{\vec{i},j} \langle h_{\vec{i}, j} (\thv)\rangle)^2 \nonumber\\
    &= \sum_{\vec{i}} p_{\vec{i}} y_{\vec{i}}^2 - 2 \sum_{\vec{i},j} p_{\vec{i}} y_{\vec{i}} c_{\vec{i},j} \langle h_{\vec{i}, j} (\thv)\rangle \nonumber \\
    &+ \sum_{\vec{i},j, j', j' \ne j} p_{\vec{i}} c_{\vec{i},j} c_{\vec{i}, j'} \langle h_{\vec{i}, j} (\thv)\rangle \langle h_{\vec{i}, j'} (\thv)\rangle+ \sum_{\vec{i},j} p_{\vec{i}} c_{\vec{i},j}^2 \langle h_{\vec{i}, j} (\thv)\rangle^2.
\end{align}
Defining $\widehat{\mathcal{E}}_{\vec{i},j} $ as the estimator for $ T_{\vec{i},j}$ as previously,
\begin{equation*}
    \widehat{\mathcal{E}}_{\vec{i},j} = \frac{1}{\mathbb{E}[s_{\vec{i},j}]}\sum_{k=1}^{s_{\vec{i},j}} r_{\vec{i},j,k}.
\end{equation*}
When $j \ne j'$ we define,
\begin{equation}
    \widehat{\mathcal{Q}}_{\vec{i},j,j'} = \frac{1}{\mathbb{E}[s_{\vec{i},j}s_{\vec{i},j'}]} \sum_{k=1}^{s_{\vec{i},j}} \sum_{k'=1}^{s_{\vec{i},j'}} r_{\vec{i},j,k} r_{\vec{i},j',k'}\,,
\end{equation}
and finally,
\begin{equation}\label{moment2tijest}
\widehat{\mathcal{Q}}_{\vec{i},j} = \frac{1}{\mathbb{E}[s_{\vec{i},j} (s_{\vec{i},j} - 1)]}\sum_{k=1}^{s_{\vec{i},j}} \sum_{k'=1, k' \ne k}^{s_{\vec{i},j}} r_{\vec{i},j,k} r_{\vec{i},j,k'}.
\end{equation}
Collecting these terms we define an estimator of the full MSE loss function as
\begin{align}\label{eq:unbiasedestimatorsmsewithhamiltonianterms}
    &\widehat{\LC}'_{MSE}= \sum_{\vec{i}} p_{\vec{i}} y_{\vec{i}}^2 - 2 \sum_{\vec{i},j} p_{\vec{i}} y_{\vec{i}} c_{\vec{i},j} \widehat{\mathcal{E}}_{\vec{i},j} \nonumber \\ 
    &+ \sum_{\vec{i}, j,j', j' \ne j} p_{\vec{i}} c_{\vec{i},j} c_{\vec{i},j'} \widehat{\mathcal{Q}}_{\vec{i},j,j'} + \sum_{\vec{i}, j} p_{\vec{i}} c_{\vec{i},j}^2 \widehat{\mathcal{Q}}_{\vec{i},j},
\end{align}
which can be used to establish the following proposition.
\begin{proposition}
\label{prop4}
$\widehat{\LC}'_{MSE}$ is an unbiased estimator for $\LC'_{MSE}$.
\end{proposition}
\begin{proof}
Taking the above definition we can write the following result:
\begin{align}
    & \mathbb{E}[\widehat{\LC}'_{MSE}]= \sum_{\vec{i}} p_{\vec{i}} y_{\vec{i}}^2 - 2 \sum_{\vec{i},j} p_{\vec{i}} y_{\vec{i}} c_{\vec{i},j}  \frac{1}{\mathbb{E}[s_{\vec{i},j}]} \mathbb{E}[\sum_{k=1}^{s_{\vec{i},j}} r_{\vec{i},j,k}] \nonumber \\ 
    &+ \sum_{\vec{i}, j,j', j' \ne j} p_{\vec{i}} c_{\vec{i},j} c_{\vec{i},j'} \frac{1}{\mathbb{E}[s_{\vec{i},j}s_{\vec{i},j'}]} \mathbb{E}[\sum_{k=1}^{s_{\vec{i},j}} \sum_{k'=1}^{s_{\vec{i},j'}} r_{\vec{i},j,k} r_{\vec{i},j',k'}] \nonumber \\ 
    &+ \sum_{\vec{i}, j} p_{\vec{i}} c_{\vec{i},j}^2 \frac{1}{\mathbb{E}[s_{\vec{i},j} (s_{\vec{i},j} - 1)]} \mathbb{E}[\sum_{k=1}^{s_{\vec{i},j}} \sum_{k'=1, k' \ne k}^{s_{\vec{i},j}} r_{\vec{i},j,k} r_{\vec{i},j,k'}].
\end{align}
Given that
\begin{align}\label{givenestimatorsmseunbiased}
    &\mathbb{E}[\sum_{k=1}^{s_{\vec{i},j}} r_{\vec{i},j,k}] = \mathbb{E}[s_{\vec{i},j}] \langle h_{\vec{i}, j} (\thv)\rangle\\
    & \mathbb{E}[\sum_{k=1}^{s_{\vec{i},j}} \sum_{k'=1}^{s_{\vec{i},j'}} r_{\vec{i},j,k} r_{\vec{i},j',k'}] = \mathbb{E}[s_{i,j} s_{i,j'}] \langle h_{\vec{i}, j} (\thv)\rangle \langle h_{\vec{i}, j'} (\thv)\rangle \\
    & \mathbb{E}[\sum_{k=1}^{s_{\vec{i},j}} \sum_{k'=1, k' \ne k}^{s_{\vec{i},j}} r_{\vec{i},j,k} r_{\vec{i},j,k'}] = \mathbb{E}[s_{\vec{i},j} (s_{\vec{i},j} - 1)] (\langle h_{\vec{i}, j} (\thv)\rangle)^2\,.
\end{align}
We can combine the above expressions to obtain
\begin{align}
      \mathbb{E}[\widehat{\LC}'_{MSE}]&= \sum_{\vec{i}} p_{\vec{i}} y_{\vec{i}}^2  - 2 \sum_{\vec{i},j} p_{\vec{i}} y_{\vec{i}} c_{\vec{i},j} \langle h_{\vec{i}, j} (\thv)\rangle \nonumber \\
     & \quad + \sum_{\vec{i}, j,j', j' \ne j} p_{\vec{i}} c_{\vec{i},j} c_{\vec{i}, j'} \langle h_{\vec{i}, j} (\thv)\rangle \langle h_{\vec{i}, j'} (\thv)\rangle \nonumber \\
     & \quad + \sum_{\vec{i}, j} p_{\vec{i}} c_{\vec{i},j}^2  (\langle h_{\vec{i}, j} (\thv)\rangle)^2 \\
     &= \LC'_{MSE} \,.
\end{align}
\end{proof}

This shows we can estimate the full MSE cost function by evaluating 4 terms. One can be evaluated exactly and the others can be estimated by sampling from their corresponding probability distributions built with the weights as parameters.

\end{document}